\documentclass[11pt]{article}

\usepackage[utf8]{inputenc}
\usepackage[T1]{fontenc}

\usepackage{amsmath, amssymb, amsthm}

\usepackage{arxiv}
\bibpunct{(}{)}{;}{a}{,}{,}
\usepackage{amsfonts}
\usepackage{graphicx}
\usepackage{booktabs}
\usepackage{microtype}
\usepackage{url}
\usepackage{xcolor}
\usepackage{physics}
\usepackage[autostyle]{csquotes}
\usepackage{capt-of}
\usepackage{xspace}
\usepackage{enumitem}
\usepackage{tikz}
\usetikzlibrary{arrows,patterns,calc,shapes,positioning,arrows.meta}
\usepackage{tikz-cd}
\usepackage{float}

\newenvironment{romanenumerate}{\begin{enumerate}[label=(\roman*), ref=\roman*]}{\end{enumerate}}

\graphicspath{{./downloads/}}
\numberwithin{equation}{section}
\theoremstyle{plain}
\newtheorem{theorem}{Theorem}[section]
\newtheorem{lemma}[theorem]{Lemma}
\newtheorem{proposition}[theorem]{Proposition}
\newtheorem{corollary}[theorem]{Corollary}

\theoremstyle{definition}
\newtheorem{definition}[theorem]{Definition}

\theoremstyle{remark}
\newtheorem{remark}[theorem]{Remark}

\newcounter{inlineitem}
\newcommand{\inlineitemlabel}{(\roman{inlineitem})}

\newcommand{\soumik}[1]{{\color{black}  #1}}\newcommand{\sou}[1]{{\color{black}  #1}} 
\newcommand{\rch}{\operatorname{rch}}

\newcommand{\man}{\mathcal{M}}
\newcommand{\tp}[2]{\mathcal{T}_{#1}^{#2}\mathcal{M}}
\newcommand{\vol}{\operatorname{vol}}
\newcommand{\diam}{\operatorname{diam}}
\newcommand{\R}{\mathbb{R}}

\DeclareMathOperator*{\E}{\mathbb{E}}
\newcommand{\Ex}{\mathbb{E}}
\renewcommand{\Pr}{\mathbb{P}}
\newcommand{\Prob}[1]{\Pr\!\left[#1\right]}
\newcommand{\Var}{\operatorname{Var}}
\newcommand{\calN}{\mathcal{N}}
\newcommand{\pth}[1]{\left(#1\right)}

\newcommand{\e}{\varepsilon}

\usepackage{hyperref}
\hypersetup{
  colorlinks=true,
  linkcolor=blue,
  citecolor=blue,
  urlcolor=blue,
}

\let\oldcite\cite
\let\oldcitet\citet
\let\oldcitep\citep
\renewcommand{\cite}{\oldcite*}
\renewcommand{\citet}{\oldcitet*}
\renewcommand{\citep}{\oldcitep*}

\title{Low-Dimensional Embeddings for Gaussian Kernels on Manifolds}

\author{Soumik Dutta\thanks{Institute of Informatics, University of Warsaw, Poland. Supported by the Polish NCN SONATA Grant no.\ 2019/35/D/ST6/04525. \texttt{s.dutta2@uw.edu.pl}, \texttt{k.dutta@uw.edu.pl}} \and Kunal Dutta\footnotemark[1]}

\begin{document}
\maketitle

\begin{abstract}
The Gaussian kernel is a widely used similarity measure that captures nonlinear relationships in data. The Gaussian kernel 
function also gives rise to the \emph{Gaussian kernel distance}, which has several applications in areas such as kernel PCA, 
spectral clustering, etc.
However, many applications involve the computation of Gaussian kernel distances between a large number of pairs of data points, 
which can be computationally prohibitive.  
Using the \emph{Random Fourier Features} sampling method of Rahimi and Recht [NeurIPS 2007], Chen and Phillips [ALT. 2017] proved that 
for points in a $d$-dimensional Euclidean ball in $\R^N$, having radius-\(R\), sampling \(t=\Omega\!\bigl(\tfrac{d}{\varepsilon^{2}}\log\tfrac{dR}{\varepsilon}\bigr)\) 
random features suffice to preserve all pairwise Gaussian–kernel distances up to a \((1\!\pm\!\varepsilon)\) multiplicative factor with high probability. 
For the more general case when the data points are on an arbitrary submanifold $\mathcal{M} \subset \R^N$ having positive reach, and with intrinsic dimensionality $d$, we establish the first 
\emph{uniform, relative-error} embedding theorem for Gaussian kernel distances. We show that drawing  
\[
O\!\Bigl(
      \tfrac{d}{\varepsilon^{2}}\,
      \log\!\bigl(
         \tfrac{\vol(\mathcal{M})^{2}\,N^{2d}}
              {\vol(B_{1}^{d}(0))^{2}\rch(\mathcal{M})^{2d}\,\varepsilon^{2d+1}\,\delta}
      \bigr)
   \Bigr) \;\;\approx\;\; O\Bigl(\tfrac{d^2}{\e^2}\bigl(\log N + \log \frac{1}{\e\delta}\bigr)\Bigr), 
\]
Random Fourier Features are enough to guarantee, with probability $1-\delta$, that the Gaussian kernel distance between each pair of points in the manifold, is 
preserved up to a relative error of $\e$.
Thus, our bound depends only logarithmically on the ambient dimension and the manifold parameters such as volume and reach, while 
maintaining the optimal Euclidean rate of $1/\e^2$.
Moreover, our analysis shows that the entire \emph{continuous} geometry of any smooth manifold can be sketched faithfully, regardless of ambient 
curvature or embedding dimension.  

We believe this has consequences for kernel PCA, spectral clustering, Laplacian eigenmaps, and topological data-analysis,
and can enable pipelines to operate in $O(nt)$ time and memory—independent of the ambient dimension $N$ and of the sample size $n$—without 
sacrificing the manifold structure that these methods are designed to exploit.

In the case of topological data analysis, we prove that under the same RFF embedding, the persistent homology of the manifold is preserved: 
the weighted Čech and Rips filtrations built with Gaussian kernel power distance are $(1\pm\varepsilon_\star)$-interleaved, where 
$\varepsilon_\star$ incorporates both distance distortion and kernel weight approximation.

\end{abstract}
 
\noindent\textbf{Keywords:} random Fourier features, manifold learning, kernel methods, dimensionality reduction, reach

\section{Introduction}

Kernels are fundamental similarity measures between pairs of points in an ambient space, and are extensively utilized in data analysis
and machine learning for capturing nonlinear relationships among data points~\cite{mercer1909xvi, smolalearning}. Some of the most popular 
types of kernels used in data analysis are Gaussian and Laplace kernels, polynomial kernels, Dirichlet and other kernels..
A wide class of kernels induce inner products between pairs of data points in high-dimensional or infinite-dimensional feature spaces.
For \emph{shift-invariant} kernels, these inner products can be used to induce the notion of \emph{kernel distances} in the usual way - an inner product gives 
rise to a norm, and the distance between a pair of points is the norm of the difference vector of the two points. For many applications such as clustering, 
dimensionality reduction, and manifold learning~\cite{hofmann2008kernel, ghojogh2021reproducing}, kernel distances are more natural to use than kernel inner products.

However, directly computing kernel inner products and distances in large datasets can be computationally
prohibitive due to the complexity involved in high-dimensional computations. A popular
solution is to approximate the kernel using low-dimensional Euclidean embeddings, significantly
improving computational efficiency~\cite{liu2021random}.

For a broad class of kernels, the \emph{Random Fourier Features} (RFF) method of Rahimi and Recht~\cite{rahimi2007random} provides such
embeddings by mapping data points into a lower-dimensional Euclidean space, approximating
kernel inner products via standard dot products~\cite{rahimi2007random}. Prior work~\cite{rahimi2007random, sutherland2015error}
provides refined bounds for approximating kernel function values via RFFs, enabling scalable
kernel ridge regression and SVMs. However, several applications in areas such as manifold
learning~\cite{jayasumana2015kernel}, kernel-based clustering~\cite{ghojogh2021reproducing} and topological data
analysis~\cite{boissonnat2024euclidean} critically require the preservation of not only the kernel function values,
but the kernel distances between pairs of points --- a strictly stronger geometric requirement. This requirement was
subsequently addressed by Chen and Phillips~\cite{chen2017relative}, who showed that the RFF embedding
also preserves pairwise Gaussian kernel distances up to a $(1 \pm \varepsilon)$-factor, provided the target
dimension is at least $O\!\left(\frac{\log n}{\varepsilon^2}\right)$, where $n$ is the number of data points and $\varepsilon \in (0,1)$ is the target relative error.

For Euclidean distances, the celebrated result of Johnson and
Lindenstrauss~\cite{johnson1984extensions, dasgupta2003elementary} has led to a well-developed theory of low-dimensional embeddings 
using random projections. These ideas have found
several applications in computer science such as hashing~\cite{indyk1998approximate} or clustering~\cite{boutsidis2010random} high-dimensional data,
high-dimensional topological data analysis~\cite{lotz2019persistent}, etc. In this context, an interesting line of
research has been the investigation of the case when all data points share some property such
as being on a lower-dimensional flat or manifold. This assumption, often referred to as the
manifold hypothesis~\cite{fefferman2016testing, ma2012manifold}, has been highly influential in Machine Learning, Artificial Intelligence
and several related areas. Thus, Baraniuk and
Wakin~\cite{baraniuk2009random} showed that random projections preserve all pairwise distances between arbibtrary set of points  
on a manifold, i.e.\ the target dimension becomes independent of the number of
points and the ambient dimension. These bounds were later improved and generalized by
Clarkson~\cite{clarkson2008tighter} and Verma~\cite{verma2011note}.

The emerging theory of low-dimensional embeddings for Gaussian and other kernels, initiated by Rahimi and Recht, parallels in some sense the 
development of the above random projection based ideas for Euclidean distances. While the underlying distance functions are quite 
different -- Euclidean distances are inner products on finite-dimensional spaces and scale linearly, whereas kernel distances are often highly non-linear and involve 
inner products in infinite-dimensional spaces -- some similarities exist. For instance given a relative error parameter $\e$, the required embedding dimension 
for Euclidean distances and Gaussian kernels is both $O(\e^{-2}\log n)$. Moreover, the low-dimensional embeddings for Euclidean distances and Gaussian kernels 
both involve sampling with random Gaussian vectors (though the Gaussian kernel embedding uses trigonometric functions of the random projections).

These unexpected similarities naturally raise some intriguing questions. For instance, just as the Johnson-Lindenstrauss bound on the 
embedding dimension can be improved 
when the points lie on a smooth submanifold, having a lower intrinsic dimension, could it be possible to improve the Chen-Phillips bound 
under similar conditions? More formally, we ask the following question.

\medskip
\noindent\textbf{Problem.} Given points lying on a submanifold $\mathcal{M} \subset \mathbb{R}^N$
having intrinsic dimension $d$ and a relative error parameter $\varepsilon \in (0,1]$, is it possible to obtain a Euclidean
embedding of the points in a space whose target dimension depends only on $d$, $\varepsilon$
and $\mathcal{M}$?
\medskip

A special case of this problem was already investigated by Chen and Phillips~\cite{chen2017relative}, when the 
points lie inside the unit $d$-dimensional ball in $\mathbb{R}^N$. In this case their result
implies that the target dimension only needed to be $\Omega\left(\frac{d}{\varepsilon^2}\right)$, a bound 
independent of the number of points and the ambient dimension. For a general point set, it is possible to 
think of the smallest bounding ball for the data points and then apply the Chen-Phillips bound.
However, these bounds could be very loose, as a manifold of low intrinsic dimensionality (e.g. a curve) could be embedded over 
all of $\mathbb{R}^N$, so that the bounding ball would be full-dimensional in $\mathbb{R}^N$. Thus, it would be advantageous 
to obtain a bound on the target dimension, which does not depend on the dimensionality of the bounding ball.

In our main result, we address the general version of this question, and obtain nearly optimal bounds --- we show
that using RFFs, essentially
\[
   O\!\left(\frac{d^2 \log N}{\varepsilon^2}\right)
\]
dimensions suffice to obtain a Euclidean embedding which preserves all pairwise Gaussian
kernel distances, up to a relative error of $(1 \pm \varepsilon)$. (Here the constant in 
the $O$-notation depends on the invariants of the manifold $\mathcal{M}$.

Thus, in this paper, we provide uniform relative error bounds for kernel distance approximations
over general manifolds using Random Fourier Features. We show that when the manifold's
reach is positive, the entire geometric structure induced by the Gaussian kernel can be
approximated uniformly with high fidelity using RFF embeddings. Notably, our results
indicate that the required embedding dimension $t$ scales logarithmically with the ambient
dimension $N$, facilitating efficient and high-fidelity applications of standard algorithms
without necessitating computations of full kernel matrices.

Further, we also state and prove several algorithmic consequences for downstream applications of 
our low-dimensional embeddings, in areas such as $(i)$ topological data analysis, 
$(ii)$ kernel $k$-means clustering, $(iii)$ kernel distance matching, 
$(iv)$ kernel nearest neighbour search, and $(v)$ kernel learning. In general, our algorithmic improvements 
result from replacing kernel computations by computations involving low-dimensional 
Euclidean distances.

\subsection{Preliminaries and Related Work}
Random Fourier Features were introduced in the seminal work of Rahimi and Recht~\cite{rahimi2007random} as a 
way to scale kernel methods by reducing the cost of kernel computations~\cite{rahimi2007random}. For the Gaussian 
kernel, they approximate the kernel value between two points by a dot product in a lower-dimensional space. This 
has proved extremely useful in many settings, leading to extensive usage in Machine Learning (see, e.g.,~\cite{li2019towards}).

The Gaussian kernel is defined as
\begin{equation}\label{eq:gaussian-kernel}
  K_\sigma(x,y) = \exp\!\Bigl(-\tfrac{\|x-y\|^{2}}{2\sigma^{2}}\Bigr), \quad x,y\in\mathbb{R}^{N}.
\end{equation}
By Mercer's theorem~\cite{mercer1909xvi}, $K_\sigma$ can be expressed as an inner product
$K_\sigma(x,y)=\langle\psi(x),\psi(y)\rangle_{\mathcal{H}}$ in an infinite-dimensional
reproducing kernel Hilbert space (RKHS) $\mathcal{H}$, with an explicit but intractable
feature map $\psi\colon\mathbb{R}^{N}\to\mathcal{H}$~\cite{ghojogh2021reproducing}.
A kernel $K(x,y)$ is said to be \emph{positive-definite}, if for every finite set of points $\{p_1,\ldots,p_n\}\in \R^N$, 
(where $N$ is finite), the matrix whose $i,j$-th entries are $K(p_i,p_j)$ is positive definite. 
For positive-definite kernels, Bochner's theorem implies 
the existence of a distribution over the feature space whose expected value at a pair of points is the kernel function
evaluated at those points. The key observation of Rahimi and Recht~\cite{rahimi2007random} was that for \emph{shift-invariant} 
kernels such as the Gaussian kernel, this expectation can be made explicit:
\begin{equation}\label{eq:bochner}
  K_\sigma(x,y)
  = \mathbb{E}_{\omega\sim\mathcal{N}(0,\,\sigma^{-2}I_{N})}\!\bigl[\cos(\langle\omega,x-y\rangle)\bigr].
\end{equation}
Approximating this expectation by drawing $t$ i.i.d.\ frequencies
$\omega^{1},\ldots,\omega^{t}\sim\mathcal{N}(0,\sigma^{-2}I_{N})$, and averaging, yields the
finite-dimensional feature map
\begin{equation}\label{eq:rff-map}
  \phi(x) = \tfrac{1}{\sqrt{t}}\bigl[\cos(\langle\omega^{1},x\rangle),\,
  \sin(\langle\omega^{1},x\rangle),\,\ldots,\,
  \cos(\langle\omega^{t},x\rangle),\,\sin(\langle\omega^{t},x\rangle)\bigr]^{\!\top}\in\mathbb{R}^{2t},
\end{equation}
so that $\hat{K}(x,y)=\langle\phi(x),\phi(y)\rangle\approx K_\sigma(x,y)$.

For positive-definite kernels, Mercer's theorem also implies the existence of a \emph{kernel distance} from the inner product given by the kernel function,  
i.e. the kernel distance between two points is the square of the inner product (the kernel function) of their difference vector with itself. Thus we can define 
\[ D_K^2(x,y) = D_K^2(x-y) = 2(1-\exp\pth{-\|x-y\|_2^2}), \; x,y \in \R^N.\]  
While the kernel distance is not used as widely as the kernel-based similarity function, in several applications such as clustering, kernel distances are more 
important than kernel function values. \sou{There are related works that highlight connections between RFFs and such applications 
\cite{gedon2023invertible}.} Chen and Phillips~\cite{chen2017relative} showed that RFF can preserve Gaussian kernel distances with relative error over a ball, 
specifically requiring dimensions $t = \Omega\left(\frac{N}{\varepsilon^2} \log\left(\frac{N}{\varepsilon} \frac{r}{\delta}\right)\right)$ to achieve uniform 
relative error bounds, where $N$ is the ambient dimension of the bounding ball, $r$ is the radius of that ball, $\varepsilon \in (0,1)$ is the target relative 
error, and $\delta \in (0,1)$ is the failure probability (the bound holds with probability at least $1 - \delta$).

Rahimi and Recht also demonstrated that the Gaussian kernel function is preserved up to additive error for points inside a bounded-radius ball around 
the origin~\cite{rahimi2007random}, where the target dimension depends only the radius of the ball. However, their analyses do not extend efficiently to arbitrary manifolds.

In a different line of research, earlier work on embeddings based on the Johnson–Lindenstrauss lemma provides bounds for relative preservation of Euclidean distances, 
requiring projection dimensions of $O(\log n/\varepsilon^2)$ for $n$ points~\cite{johnson1984extensions, dasgupta2003elementary}. Yet these methods are inherently 
linear, making them unsuitable for directly handling nonlinear kernel distances.

In this context, Baraniuk and Wakin~\cite{baraniuk2009random} established that random linear projections preserve pairwise distances on smooth manifolds with high 
probability, requiring a number of projections linear in the manifold's intrinsic dimension, specifically 
$M = O\!\left(\frac{d\log(NVR_{\mathrm{geo}}\tau^{-1}\varepsilon^{-1})\log(1/\delta)}{\varepsilon^2}\right)$ 
\sou{where $M$ is number of projections, $d$ is the intrinsic manifold dimension, $N$ is ambient dimension, $V$ is manifold volume, $R_{\mathrm{geo}}$ is geodesic covering regularity, $\tau$ is the reach $\rch(\mathcal{M})$, $\delta$ is probability of failure.}
Similarly, Verma~\cite{verma2011note} proved that random projections preserve geodesic path lengths on manifolds with distortion bounded by $(1 \pm \varepsilon)$, 
independent of the ambient dimension. Clarkson~\cite{clarkson2008tighter} provided tighter bounds for random projections preserving manifold geometry, emphasizing 
the importance of extrinsic properties depending on embedding curvature such as total curvature and reach.

Other works have taken complementary directions. Tai~\cite{tai2020optimal} constructed small-sized coresets for Gaussian kernel density estimation, significantly reducing dataset size but without directly addressing uniform preservation of pairwise distances. Lotz~\cite{lotz2019persistent} demonstrated approximations of persistent homology for data structures of low complexity, showing that Gaussian width measures effectively guide projection dimension requirements.
Cheng, Jiang, Wei, and Wei~\cite{cheng2023rff_kernel_distance} study relative-error preservation of kernel distance by RFF for \emph{general} shift-invariant kernels: they prove that for wide kernel families---including standard Laplacian kernels---low feature dimension cannot yield small relative error, whereas for \emph{analytic} shift-invariant kernels (in particular the Gaussian) RFF with $\mathrm{poly}(\varepsilon^{-1}\log n)$ features achieves $\varepsilon$-relative error for all pairwise kernel distances among $n$ points.

\sou{While Random Fourier Features (RFFs) primarily offer relative error guarantees, alternative methods like Phillips and Tai's Gaussian Sketch~\cite{jeff2020gaussiansketch} provide provably superior almost relative error for kernel distance, alongside significantly improved, near-linear runtimes for applications. Avron et al.~\cite{pmlr-v70-avron17a} also studied modified RFFs for kernel ridge regression, giving improved spectral approximation bounds and statistical guarantees.
}

In the topological domain, Lotz demonstrated that persistence modules can be preserved under random projections using Gaussian width complexity~\cite{lotz2019persistent}, while Arya et al. proved interleaving guarantees for weighted filtrations by showing that simplex radii admit convex decompositions in preserved distances~\cite{arya2021dimensionality}.
Boissonnat and Dutta showed that RFF embeddings preserve the persistent homology of GKPD-based filtrations~\cite{boissonnat2024euclidean}. Other complementary approaches include Tai's coresets for kernel density estimation~\cite{tai2020optimal} and Kusano et al.'s use of RFF for vectorizing persistence 
diagrams~\cite{kusano2016pwgk}, demonstrating the broad applicability of random features in topological data analysis.

Thus, prior literature spans various aspects of distance preservation under linear or kernelized embeddings. Yet, a unified, dimensionally efficient framework 
that achieves uniform relative error guarantees for Gaussian kernel distances on general manifolds has been lacking, motivating the current study.
While these works establish important foundations, a unified framework that achieves both efficient dimensionality reduction and certified topological preservation on general manifolds has been lacking. Our work bridges this gap by providing intrinsic dimension bounds for kernel distance preservation while ensuring interleaving guarantees for the resulting persistent homology modules.

\subsection{Our Contributions}
\label{subsec:intro_main_result}
 Rahimi--Recht~\cite{rahimi2007random} show that $t = O\!\bigl(\frac{d}{\varepsilon^{2}}\log\!\frac{\sigma\,\diam(\mathcal{M})}{\varepsilon}\bigr)$ random features make the inner–product estimate $\hat K$ \emph{additively} $\varepsilon$–close to the true, shift-invariant kernel $K$ \emph{uniformly} over a compact set~$\mathcal{M}$.
Chen--Phillips~\cite{chen2017relative} upgrade this to a \emph{relative} $(1\!\pm\!\varepsilon)$ bound for \emph{Gaussian‐kernel distances}, but only for points lying in a ball of radius~$r$, with
$t=\Omega\!\bigl(\frac{N}{\varepsilon^{2}}\log\!\frac{Nr}{\varepsilon\delta}\bigr)$.
Both results depend on the \emph{ambient} diameter (or radius) and on the ambient dimension, and they do not exploit intrinsic geometry.

Our Theorem~\ref{thm:intro_uniform} removes the Chen--Phillips limitation, guaranteeing $(1\!\pm\!2\varepsilon)$ preservation of \emph{all} Gaussian kernel distances on any $C^{2}$ manifold $\man$ with positive reach. The resulting bound has effective intrinsic-dimension dependence roughly $d^{2}/\varepsilon^{2}$, up to the remaining logarithmic and geometric factors, while retaining only logarithmic dependence on the ambient dimension $N$.
Our Theorem~\ref{thm:additive_uniform} removes the Rahimi--Recht limitation: the same feature dimension suffices for \emph{uniform additive} $\varepsilon$-accuracy of kernel \emph{values} on $\man$, replacing the ambient-diameter dependence with intrinsic geometry.

We consider a compact, $d$–dimensional, $\mathcal C^{2}$ submanifold
$\mathcal{M}\subset\R^{N}$ with positive reach
$\rch(\mathcal{M})>0$. The reach of a manifold---a classical regularity measure introduced by Federer~\cite{federer1959curvature}---quantifies the largest distance up to which each point in the ambient space has a unique nearest point on the manifold. We work with the Gaussian kernel $K_\sigma$ of~\eqref{eq:gaussian-kernel} and the RFF map $\phi\colon\R^{N}\to\R^{2t}$ of~\eqref{eq:rff-map}. As in the kernel-distance convention $D_K$ above, write
\[
   D_{K_\sigma}(x,y)^2
   \;:=\;
   K_\sigma(x,x)+K_\sigma(y,y)-2K_\sigma(x,y)
   \;=\;
   2\bigl(1-K_\sigma(x,y)\bigr)
\]
for the associated Gaussian kernel distance at bandwidth $\sigma$. Below, $B^{d}_{1}(0)\subset\R^{d}$ denotes the closed unit ball (radius $1$, centered at the origin); see Section~\ref{sec: prelim} for the general ball notation $B^{d}_{r}(p)$.

\begin{theorem}[Uniform kernel–distance preservation]
\label{thm:intro_uniform}
Fix accuracy $\varepsilon\in(0,\rch(\man)/2)$ and confidence
$\delta\in(0,1)$.
If
\begin{equation}
   t
   \;=\;
   \Omega\!\Bigl(
   \frac{d}{\varepsilon^{2}}\,
   \log\!\Bigl(
      \frac{\vol(\mathcal{M})^{2}\,N^{\,2d}}
           {\vol(B^{d}_{1}(0))^{2}\,
           \rch(\mathcal{M})^{\,2d}\,
            \varepsilon^{\,2d+1}\,
            \delta}
   \Bigr)\Bigr),
   \label{eq:intro_sample_complexity}
\end{equation}

then with probability at least $1-\delta$ the RFF embedding satisfies
\[
   (1-2\varepsilon)\;
   D_{K_{\sigma}}(p,q)^{2}
   \;\le\;
\norm{\phi(p)-\phi(q)}^2
   \;\le\;
   (1+2\varepsilon)\;
   D_{K_{\sigma}}(p,q)^{2},
   \qquad
   \forall p,q\in\mathcal{M}.
\]
\end{theorem}

\begin{remark}\label{ex:slow-helix}
\sou{Our contribution lies in reducing the target embedding dimension from a linear dependence on the ambient dimension $N$ 
to the intrinsic manifold dimension $d$.} One natural way to apply the result of Chen and Phillips~\cite{chen2017relative} 
to a manifold $\mathcal{M} \subset \mathbb{R}^N$ is to first embed $\mathcal{M}$ within a Euclidean ball of radius 
$r = \diam(\mathcal{M})$, and then invoke their result over this enclosing ball. Their bound requires 
$t = \Omega\bigl(\frac{N}{\varepsilon^2} \log(\frac{N}{\varepsilon} \cdot \frac{r}{\delta})\bigr)$ dimensions to obtain a 
$(1 \pm \varepsilon)$ relative error approximation to Gaussian kernel distances.

In contrast, our Theorem~\ref{thm:intro_uniform} provides a uniform relative error bound directly over the manifold 
$\mathcal{M}$, without relying on any ambient Euclidean ball. The sample complexity in our result depends on the 
intrinsic dimension $d$, the reach $\rch(\mathcal{M})$, and the intrinsic volume $\vol(\mathcal{M})$, with only logarithmic 
dependence on the ambient dimension $N$. Notably, our bound avoids any dependence on $\diam(\mathcal{M})$, and the logarithmic 
ambient dependence yields a significant advantage in high dimensions. \\

Below, we give an example showing a case where our bound can be much better than that of Chen and Phillips. 

\paragraph*{Example [Slow Helix with Bounded Curvature in High Dimensions]:} 
Let \( N \in \mathbb{N} \) be large. Define the 1D manifold
\(
\mathcal{M}=\{x(s):s\in[0,2\pi]\}
\)
with
\(
x(s)=\frac{1}{\sqrt{N}}\bigl(\cos\frac{s}{N},\sin\frac{s}{N},\ldots,\cos\frac{s}{2},\sin\frac{s}{2}\bigr)\in\mathbb{R}^N
\)
(the $j$th pair has frequency $j/N$, ending at $\cos(s/2),\sin(s/2)$).
This defines a \emph{slow} helix where the frequency of each harmonic decreases inversely with \( N \).
Before we apply our result to this slow helix, let us first consider a few geometric properties of this slow helix.

\smallskip\noindent\textbf{Norm.}~We have \(\|x(s)\|^2=\frac{1}{N}\sum_{j=1}^{N/2}\bigl(\cos^2(\tfrac{js}{N})+\sin^2(\tfrac{js}{N})\bigr)=\tfrac12\), hence \(\mathcal{M}\subset B^N_2(0)\) and \(d=1\).

\smallskip\noindent\textbf{Reach.}~For each~\(j\), the scaled pair \(\bigl(N^{-1/2}\cos(js/N),\,N^{-1/2}\sin(js/N)\bigr)\) has coordinatewise second derivatives of magnitude \(j^2/N^{5/2}\), so \(\|\ddot{x}(s)\|^2=\sum_{j=1}^{N/2} j^4/N^5\).
Since \(\sum_{j=1}^{N/2} j^4\sim \tfrac{1}{5}(N/2)^5=\Theta(N^5)\), this gives \(\|\ddot{x}(s)\|=\Theta(1)\) and therefore \(\rch(\mathcal{M})=\Omega(1)\) in~\(N\).

\smallskip\noindent\textbf{Volume.}~Each $(\cos,\sin)$ block contributes \(j^2/N^3\) to \(\|\dot{x}(s)\|^2\), so \(\|\dot{x}(s)\|^2=\sum_{j=1}^{N/2} j^2/N^3\) is independent of~\(s\), and \(\|\dot{x}(s)\|=N^{-3/2}\sqrt{\sum_{j=1}^{N/2} j^2}=\frac{1}{\sqrt{N}}\sqrt{\sum_{j=1}^{N/2}(j/N)^2}\).
Since \(\sum_{j=1}^{N/2} j^2\sim \tfrac{1}{3}(N/2)^3=\Theta(N^3)\), we obtain \(\|\dot{x}(s)\|=\Theta(1)\) and \(\mathrm{Vol}(\mathcal{M})=\int_0^{2\pi}\|\dot{x}(s)\|\,ds=2\pi\,\|\dot{x}(0)\|=\Theta(1)\).

\smallskip
\noindent
Note that the slow helix is genuinely $N$-dimensional: there is no nonzero $a\in\mathbb{R}^N$ with $a\cdot x(s)=0$ for every $s\in[0,2\pi]$.
Equivalently, the curve is not contained in any hyperplane through the origin, so its linear span in $\mathbb{R}^N$ is all of~$\mathbb{R}^N$.

\smallskip
\noindent
For $j=1,\ldots,N/2$, write $u_j(s):=\frac{1}{\sqrt{N}}\cos(js/N)$ and $v_j(s):=\frac{1}{\sqrt{N}}\sin(js/N)$; these have pairwise distinct frequencies~$j/N$.
Regarded as elements of $L^2([0,2\pi])$ with the usual inner product
\(
\langle f,g\rangle:=\int_0^{2\pi} f(s)\,g(s)\,ds
\),
the family $\{u_j,v_j\}_{j=1}^{N/2}$ is linearly independent; hence no nontrivial linear combination $\sum_{\ell=1}^N a_\ell x_\ell(s)$ vanishes identically on~$[0,2\pi]$.

\smallskip
\noindent
Consequently, $\mathcal{M}$ cannot lie in a ball in any strictly lower-dimensional linear subspace of~$\mathbb{R}^N$.

\smallskip
\noindent
We are now in a position to compare the two guarantees on this helix.
Theorem~\ref{thm:intro_uniform} applies directly to~$\mathcal{M}$ with its intrinsic reach and volume.
The Chen--Phillips bound, by contrast, is stated for points in a Euclidean ball in~$\mathbb{R}^N$; since $\mathcal{M}\subset B_2^N(0)$, invoking it requires enlarging to that full unit ball, where their analysis yields
\(t_{\text{CP}} = \Omega\bigl(\frac{N}{\varepsilon^2} \log(\frac{N}{\varepsilon \delta})\bigr)\)
feature dimensions, whereas on~$\mathcal{M}$ alone, Theorem~\ref{thm:intro_uniform} needs only
\(t = \Omega\bigl(\frac{1}{\varepsilon^2} \log(\frac{N}{\varepsilon^3 \delta})\bigr)\).

This example shows a setting where both methods operate within a unit ball, yet
the Chen–Phillips bound incurs a full linear cost in the ambient dimension $N$, while our intrinsic
RFF approach maintains efficient scaling.
\end{remark}

\begin{remark}\sou{It should be noted that the logarithmic term in our bound hides factors exponential in $d$. Thus, 
the target dimension should be interpreted as $O(d^2/\varepsilon^2)$. 
We emphasize that the non-linearity in $d$ is offset by the reduction in dependence on the ambient dimension $N$, from linear to logarithmic.}
\end{remark}

\begin{remark}
Cheng et al.~\cite{cheng2023rff_kernel_distance} suggest that there may be a gap in the 
Chen--Phillips~\cite{chen2017relative} proof of the relative-error bound for Gaussian kernel distances on a ball, specifically in Lemma~5 of that paper.

In our proof of Theorem~\ref{thm:intro_uniform}, we use a version of the concentration inequality of Chen and Phillips, however it is nevertheless a special case of a stronger inequality due to Boissonnat and Dutta~\cite{boissonnat2024euclidean} (Lemma~8\textup{(2)}) (see also 
the full version~\cite{boissonnatduttaesa24fullversion}).
In particular, taking~$S$ in their lemma to consist of a single column vector (the difference vector between two points), together with the 
trivial bound $r_{\mathrm{st}}\ge 1$, recovers the Chen--Phillips concentration bound. We provide a proof of the lemma of ~\cite{boissonnat2024euclidean}, 
restricted to our case of interest, in the Appendix.
\end{remark}

\medskip

\medskip
\textbf{Uniform additive approximation of kernel values.}
Besides relative control of kernel \emph{distances}, many applications require uniform additive accuracy for the Gaussian kernel \emph{values} $K_\sigma(p,q)$ under the RFF estimator $\widehat K$.
A short sketch of the idea appears at the end of the proof overview (Section~\ref{subsec:proof_overview}); Section~\ref{sec:app-kernel-values} develops the lemmas and ends with the complete proof.

\begin{theorem}[Uniform additive error bound for kernel values]
\label{thm:additive_uniform}
Fix $\varepsilon < \rch(\man)/2$ and $\delta\in(0,1)$. If
\[
t
\;=\;
\Omega\!\Bigl(
\frac{d}{\varepsilon^{2}}\,
\log\!\Bigl(
\frac{\vol(\man)^{2}N^{2d}}
     {\vol(B^d_1(0))^2\rch(\man)^{2d}\,\varepsilon^{2d+1}\,\delta}
\Bigr)\Bigr),
\]
then with probability at least $1-\delta$,
\[
\sup_{p,q\in\man}
\;\bigl|K_{\sigma}(p,q)-\widehat K(p,q)\bigr|
\;\le\;
\varepsilon.
\]
\end{theorem}

\medskip
Equation~\eqref{eq:intro_sample_complexity} shows that the required
feature dimension is
\(
   t=\Theta\!\bigl(\varepsilon^{-2}\bigr)
\)
up to a \emph{logarithmic} factor that depends on intrinsic geometry
($\vol(\mathcal{M})$, $\rch(\mathcal{M})$) and the ambient dimension~$N$. This is similar to the case of random linear projection ~\cite{baraniuk2009random}.
The leading $1/\varepsilon^{2}$ term matches the
optimal Euclidean RFF rate; reach (curvature) appears only in the logarithmic factor.

\medskip
\textbf{Preserving persistent homology.} Beyond kernel distances and values, many downstream pipelines depend on the \emph{persistent homology} of a kernel-weighted filtration built from the data; our embedding preserves this structure as well, up to a controlled multiplicative distortion.

For a finite point set $P\subset\man$, define the \emph{kernel weight} of $p\in P$ by
\[
w(p)
\;:=\;
-\Biggl(\frac1{|P|}\sum_{y\in P} D_{K_\sigma}^2(p,y)
\;-\;
\frac1{2|P|^2}\sum_{x,y\in P} D_{K_\sigma}^2(x,y)\Biggr),
\]
and write $\widehat P := \{(p,w(p)):p\in P\}$ for $P$ equipped with these weights. The induced \emph{Gaussian kernel power distance} between weighted points is $D_{K_\sigma}^2(\widehat p,\widehat q):=D_{K_\sigma}^2(p,q)-w(p)-w(q)$~\cite{phillips1geometric,boissonnat2024euclidean}. The \emph{weighted Čech filtration} $\check C_\alpha(\widehat P)$ includes a simplex once its minimal enclosing radius under this power distance is at most $\alpha$, and the \emph{weighted Rips filtration} $VR_\alpha(\widehat P)$ includes a simplex once every pairwise power distance within it is at most $\alpha$; see Section~\ref{sec: prelim} for the formal definitions. Two filtrations $\{F_\alpha\}$ and $\{G_\alpha\}$ are $(1\pm\varepsilon)$\emph{-interleaved} if $F_\alpha\subseteq G_{(1+\varepsilon)\alpha}\subseteq F_{(1+\varepsilon)^2\alpha}$ for every $\alpha\ge0$ (and symmetrically); by the persistence Stability Theorem (Section~\ref{subsec:interleaving}) this bounds the bottleneck distance between the corresponding persistence diagrams.

\begin{theorem}[Weighted Čech/Rips interleaving on manifolds]
\label{thm:intro_interleaving}
Under the same setting as Theorem~\ref{thm:intro_uniform}, let $P\subset\mathcal{M}$ be a finite sample of $n$ points.
Define the kernel centroid $\mu_P := \tfrac{1}{|P|}\sum_{y\in P}\phi(y)$,
the constant $c_P := \tfrac{2}{(1-\|\mu_P\|)^2}$,
and $\varepsilon_\star := \max\!\bigl\{2\varepsilon,\;\varepsilon\,c_P\bigr\}$.
Then, with probability at least $1-\delta$, the weighted \v{C}ech filtrations
$\check{C}_\alpha(\widehat{P})$ and $\check{C}_\alpha\!\bigl(\phi(P)\bigr)$
built with the Gaussian kernel power distance $D_{K_\sigma}$
are $(1\pm\varepsilon_\star)$–interleaved.
Consequently, the corresponding weighted Rips filtrations
$VR_\alpha(\widehat{P})$ and $VR_\alpha\!\bigl(\phi(P)\bigr)$
are also $(1\pm\varepsilon_\star)$–interleaved.
(Full proof: Theorem~\ref{thm:manifold_interleaving_final}.)
\end{theorem}

This topological guarantee extends the point cloud results of \cite{boissonnat2024euclidean} to the manifold setting. The key technical innovation lies in establishing relative approximation bounds for kernel weights defined through RKHS centroids, ensuring that the weighted simplex structure—and consequently the persistent homology—remains stable under projection.

\begin{remark}[On the centroid constant $c_P$]
\label{rem:centroid_constant}
The constant $c_P = \tfrac{2}{(1-\|\mu_P\|)^2}$ is strictly positive since $\|\mu_P\| < 1$ for any non-degenerate finite set~$P$.
Moreover, using Lemma~\ref{lem:kernel-weight-lower}, one can lower bound $(1-\|\mu_P\|)^2$ by
\[
(1-\|\mu_P\|)^2
\;\ge\;
\Biggl(
\Bigl(1-e^{-r^2/2}\Bigr)
\frac{\mathbb{E}\bigl[\|x-y\|^2\bigr]}{r^2}
\Biggr)^{\!2}\Big/4,
\]
where $r = \diam(\operatorname{supp}P)$.
Hence, $c_P$ is not only finite but bounded above by the reciprocal of this term.
This bound, however, is not sharp—empirically, $\|\mu_P\|$ tends to be small for well-spread data,
so the actual value of $c_P$ is often considerably lower than this theoretical upper bound.
\end{remark}

\subsection{Proof Ideas}
\label{subsec:proof_overview}
In the following, we shall assume unit bandwidth ($\sigma = 1$) for simplicity. For a general 
bandwidth $\sigma > 0$, observe that both the Gaussian kernel and the Random Fourier Feature approximation are invariant under the rescaling $x \mapsto x/\sigma$:
\[
K_\sigma(x, y) = \exp\left(-\frac{\|x - y\|^2}{2\sigma^2}\right)
= K\left(\tfrac{x}{\sigma}, \tfrac{y}{\sigma}\right).
\]
Similarly, the RFF approximation using $\omega \sim \mathcal{N}(0, \sigma^{-2} I_N)$ is equivalent to 
using $\omega' \sim \mathcal{N}(0, I_N)$ on the manifold rescaled by a factor of~$1/\sigma$. Thus in order to obtain our full theorems, it is 
sufficient to consider the case of Gaussian kernels having unit bandwidth. 

Our goal is Theorem~\ref{thm:intro_uniform}: a uniform
error bound on the ratio of the RFF distance between pairs of embedded points and their kernel distances 
in the original space, i.e. 
$D_{\hat{K}}^{2}(x-y)/D_{K_\sigma}^{2}(x-y)$ for all pairs of points $x,y\in \mathcal{M}$, where $\mathcal{M} \subset \R^N$ is a 
compact $d$-dimensional $\mathcal{C}^2$ submanifold with reach $\rch(\mathcal{M})$.

Let us denote by $\widehat{\mathcal{M}}:=\mathcal{M}-\mathcal{M}=\{y-z:y,z\in\mathcal{M}\}$, the set of all pairwise separations.
Since Gaussian kernel distances are \emph{shift-invariant} -- every comparison between $y,z\in\mathcal{M}$ 
depends only on the chord $\Delta=y-z$, so bounding the ratio 
$R = R(\Delta) := D_{\hat{K}}^{2}(\Delta)/D_{K_\sigma}^{2}(\Delta)$, pointwise and uniformly over $\Delta\in\widehat{\mathcal{M}}$, 
is precisely the statement that all kernel distances on~$\mathcal{M}$ are preserved simultaneously. \\ 

As is typical in proofs of this type (see e.g.~\cite{baraniuk2009random,rahimi2007random}), the initial framework of our proof is via a \emph{net} argument -- $(i)$ constructing an appropriately 
fine net over the manifold, $(ii)$ showing that kernel distances between all pairs of points from the net are preserved by the RFF embedding (using a union bound 
over pairs of points from the net), and $(iii)$ extending this distance preservation 
to small neighbourhoods of the net points via a Lipschitz continuity argument, via an upper bound on the magnitude of the gradient of the ratio $R(\Delta)$. 
Observe that a crucial aspect of our analysis is always to work using local tangent-plane approximations (around each point in the net), rather than using 
Euclidean balls or higher-order differential geometry. Subsequently we lift the error bounds obtained in the tangent planes to a small neighbourhood around 
the net-points. This allows us to use the intrinsic dimension of the manifold, as the tangent planes are affine spaces 
having the intrinsic dimension rather than the ambient one.  \\
 
However, an attempt to implement this basic approach encounters some major obstacles, which are conceptual as well as technical.
Handling these requires several technical and some conceptual innovations, which we regard as our main technical contribution. \\

Firstly, $\widehat{\mathcal{M}}$ always contains the origin (i.e. corresponding to points $y\in \mathcal{M}$, giving $\Delta = y-y=0\in \widehat{\mathcal{M}}$), 
where the 
ratio $R$ approaches $0/0$, since both the kernel distance as well as the RFF distance approach zero as $\Delta$ goes to zero.
Moreover, the two distances approach zero at different speeds, which, in the limit, causes the gradient to blow up around zero.
A possible approach to address this issue could be to use the Chen-Phillips result, bounding the relative error of all pairwise 
distances in a small ball around the origin. However this would force the target dimension to depend \emph{linearly} on the ambient 
dimension.

We therefore take a different tack -- using the notion of \emph{relative Lipschitz} continuity.
We show that the RFF distances when normalized by the kernel distance, as well as by the \emph{magnitude} 
of the original position vectors, is bound with high probability. This is sufficient to compare a true separation $\Delta\in \widehat{M}$ 
to its projection on to the tangent plane, and requires only a logarithmic dependence on the ambient dimension. \\ 

Next, although the manifold has bounded reach, its tangent spaces may rotate arbitrarily 
as one moves along the manifold; in particular, the manifold may twist in arbitrary directions. 
Consequently, a direct Lipschitz -- or even relative Lipschitz -- argument would require controlling the approximation error in every ambient direction. 
Such an approach inevitably incurs a dependence on the ambient dimension $N$, since the analysis must account for all 
possible directions in which the tangent plane may rotate, thereby negating our previous attempts to eliminate ambient-dimensional dependence.

Instead, we exploit the defining geometric consequence of bounded reach: locally, the manifold departs from its tangent plane 
only quadratically. More precisely, if a tangent disk has diameter $\varepsilon$, then the corresponding manifold patch remains 
within $O(\varepsilon^2)$ of that disk. Combining this quadratic deviation with the relative Lipschitz property and an appropriately 
constructed net on the tangent disk allows us to transfer estimates from the tangent plane to the manifold without paying a linear 
penalty in the ambient dimension. This is the key geometric insight that enables us to remove the dependence on $N$. \\

Finally, in general, $\widehat{\mathcal{M}}$ is \emph{not} a smooth manifold
(it can be a stratified set), so there is no global tangent bundle or
chart in which to run the linear analysis.
To handle this, we avoid covering $\widehat{\mathcal{M}}$ directly.
Instead, Section~\ref{subsec:local-charts} shows, using only the
reach of~$\mathcal{M}$, that if $U,V\subset\mathcal{M}$ are two small intrinsic neighborhoods about net points $y,z\in\mathcal{M}$, then the ambient set of chord vectors
\(
\{u-v:u\in U,\;v\in V\}
\)
lies near the affine space of \emph{tangent differences}
$T_y\mathcal{M}-T_z\mathcal{M}$ (Lemma~\ref{lem:diff-quadratic}).
That affine model has dimension at most~$2d$ and is described by intrinsic data we control.
The reach $\rch(\mathcal{M})$ enters crucially here: it bounds how far this set difference 
of intrinsic neighborhoods may deviate from the flat tangent-difference model.
Once such reach-controlled proximity is in hand, one combines the relative Lipschitz–type 
property for~$R$ (Lemma~\ref{lem:LipboundR2}) to propagate uniform ratio bounds from 
$T_y\mathcal{M}-T_z\mathcal{M}$ back to the intrinsic neighborhoods, and a union bound 
over the finitely many neighborhood pairs completes the global theorem. \\

\subsection{Applications}
\label{subsec:applications}

Random Fourier Features (RFF), introduced in~\cite{rahimi2007random}, approximate
shift-invariant kernels by an explicit finite-dimensional map so that inner products
in the feature space estimate kernel values; see~\cite{liu2021random,hofmann2008kernel}
for surveys of kernel methods and random-feature approximations.  For Gaussian kernels,
\cite{chen2017relative} showed that RFFs can preserve Gaussian kernel distances with
relative error, which is the relevant geometric quantity in many kernel-based pipelines.

Our contribution is that, when the input is supported on a compact low-dimensional
manifold, the required number of random features is controlled by the intrinsic geometry
of the manifold rather than by the size of the data set or by an ambient Euclidean ball.
Thus any RFF-based pipeline whose analysis reduces to preserving Gaussian kernel
distances or kernel values may use the manifold embedding of
Theorems~\ref{thm:intro_uniform} and~\ref{thm:additive_uniform}.  The formal statements
and the corresponding approximation and complexity calculations are deferred to
Appendix~\ref{app:applications}.

\paragraph{Kernel \(k\)-means clustering.}
Kernel \(k\)-means maps the data to an RKHS and performs Euclidean \(k\)-means there
using feature-space centroids~\cite{girolami2002mercer,dhillon2004kernel}.  RFFs replace
the implicit RKHS representation by explicit Euclidean vectors, making standard Euclidean
clustering tools applicable.  Under the manifold hypothesis, our theorem gives the same
type of kernel-distance preservation using an intrinsic feature count.  The formal
cost-preservation statement and the approximation-transfer calculation are given in
Appendix~\ref{app:manifold-chen-phillips-kmeans}.

\begin{corollary}
Let
\[
  P=\{p_1,\ldots,p_n\}\subset \mathcal M\subset\mathbb R^N,
\]
where \(\mathcal M\) is a compact \(d\)-dimensional \(C^2\) submanifold with positive reach.  Let \(K_\sigma\) be the Gaussian kernel, let
\(\Psi:\mathcal M\to\mathcal H_{K_\sigma}\) be its RKHS feature map, and let
\(\phi:\mathbb R^N\to\mathbb R^{2t}\) be the RFF map from Theorem~\ref{thm:intro_uniform}.  Assume that
\[
  t
  \;\ge\;
  C\,\frac{d}{\varepsilon^2}
  \log\!\left(
      \frac{\operatorname{vol}(\mathcal M)^2 N^{2d}}
           {\operatorname{vol}(B_1^d(0))^2
            \operatorname{rch}(\mathcal M)^{2d}
            \varepsilon^{2d+1}\delta}
  \right)
\]
for the constant \(C\) in Theorem~\ref{thm:intro_uniform}.  Then, with probability at least \(1-\delta\), the following holds simultaneously for every partition
\(\Pi=\{C_1,\ldots,C_k\}\) of \(P\):
\[
  (1-2\varepsilon)\operatorname{cost}_{K_\sigma}(\Pi)
  \;\le\;
  \operatorname{cost}_{\phi}(\Pi)
  \;\le\;
  (1+2\varepsilon)\operatorname{cost}_{K_\sigma}(\Pi),
\]
where
\[
  \operatorname{cost}_{K_\sigma}(\Pi)
  :=
  \sum_{j=1}^k\sum_{p\in C_j}
  \|\Psi(p)-\mu_j\|_{\mathcal H_{K_\sigma}}^2,
  \qquad
  \mu_j:=\frac1{|C_j|}\sum_{q\in C_j}\Psi(q),
\]
and
\[
  \operatorname{cost}_{\phi}(\Pi)
  :=
  \sum_{j=1}^k\sum_{p\in C_j}
  \|\phi(p)-\widehat\mu_j\|_2^2,
  \qquad
  \widehat\mu_j:=\frac1{|C_j|}\sum_{q\in C_j}\phi(q).
\]
Hence, if a Euclidean clustering algorithm applied to \(\phi(P)\subset\mathbb R^{2t}\) returns a partition \(\widehat\Pi\) satisfying
\[
  \operatorname{cost}_{\phi}(\widehat\Pi)
  \le
  \rho\min_{\Pi}\operatorname{cost}_{\phi}(\Pi),
\]
then the same partition satisfies
\[
  \operatorname{cost}_{K_\sigma}(\widehat\Pi)
  \le
  \rho\,\frac{1+2\varepsilon}{1-2\varepsilon}
  \min_{\Pi}\operatorname{cost}_{K_\sigma}(\Pi).
\]
Here \(\rho\ge 1\) denotes the approximation factor of the Euclidean
clustering algorithm applied to the embedded point set \(\phi(P)\); that is,
the algorithm returns a partition whose RFF-space \(k\)-means cost is at most
\(\rho\) times the optimal RFF-space \(k\)-means cost.
\end{corollary}

\paragraph*{Topological inference.}
The Gaussian kernel power distance introduced in~\cite{phillips1geometric} expresses geometric inference on kernel density estimates through weighted power distances; its sublevel sets are stable under $W_2$ perturbations of the measure and Lipschitz in the bandwidth~$\sigma$.
~\cite{boissonnat2024euclidean} showed that an RFF embedding into $\mathbb{R}^{2t}$ with $t=O(\varepsilon^{-2}\log n)$ yields a weighted \v{C}ech filtration that is interleaved with the GKPD filtration in the original space under a stable-rank condition. We prove the manifold version of this guarantee, Theorem~\ref{thm:intro_interleaving} (weighted Čech/Rips interleaving), as one of our main contributions in Section~\ref{subsec:intro_main_result}.
For kernel-distance coresets, near-linear-size summaries exist~\cite{tai2020optimal, DBLP:conf/soda/Phillips13}, and RFF makes querying them practical in high ambient dimension.
\cite{kusano2016pwgk} combine RFF with persistence-weighted Gaussian kernels to vectorize persistence diagrams for standard kernel classifiers.

In some applications one may simply restrict attention to those data points whose kernel-power weights (see e.g. ~\cite{boissonnat2024euclidean})
stay uniformly bounded away from zero. For such cases, filtrations built from Gaussian kernel weights enjoy the same $(1\!\pm\!\varepsilon)$
stability after projection.  Persistent homology of $n$ points can therefore be computed faster.

\paragraph{Kernel distance matching.}
Kernel distances are used to compare distributions, point clouds, shapes, and medical
images through RKHS mean embeddings~\cite{smola2007hilbert,gretton2012kernel,
glaunes2006large,durrleman2007measuring,joshi2011shape}.  When the task is to recover a
pointwise alignment, one may instead solve a matching problem with pairwise kernel-distance
costs.  RFFs reduce this to Euclidean geometric matching, and our manifold theorem gives
the corresponding intrinsic-dimensional version.  The reduction and approximation
guarantee are stated in Appendix~\ref{app:kernel-distance-matching}.

\begin{corollary}
Assume that the RFF map \(\phi\) satisfies
\[
    (1-2\varepsilon)D_{K_\sigma}(p,q)^2
    \le
    \|\phi(p)-\phi(q)\|_2^2
    \le
    (1+2\varepsilon)D_{K_\sigma}(p,q)^2
\]
for all \(p,q\in\mathcal M\).  Then, for every matching \(\pi\in S_n\),
\[
    (1-2\varepsilon)C_{K_\sigma}(\pi)
    \le
    C_\phi(\pi)
    \le
    (1+2\varepsilon)C_{K_\sigma}(\pi).
\]
Consequently, if a Euclidean matching algorithm applied to
\(\phi(X),\phi(Y)\subset\mathbb R^m\) has approximation factor \(\rho\ge 1\), meaning
that it returns a permutation \(\widehat\pi\) satisfying
\[
    C_\phi(\widehat\pi)
    \le
    \rho\min_{\pi\in S_n} C_\phi(\pi),
\]
then the same matching satisfies
\[
    C_{K_\sigma}(\widehat\pi)
    \le
    \rho\frac{1+2\varepsilon}{1-2\varepsilon}
    \min_{\pi\in S_n} C_{K_\sigma}(\pi).
\]
Thus, any \(\rho\)-approximate Euclidean matching algorithm applied after the manifold
RFF embedding yields a
\[
    \rho\frac{1+2\varepsilon}{1-2\varepsilon}
    =
    \rho(1+O(\varepsilon))
\]
approximation to the original Gaussian-kernel matching problem.
\end{corollary}

\paragraph{Kernel nearest-neighbor search.}
For the Gaussian kernel, nearest-neighbor search under the kernel distance is equivalent
to maximum Gaussian similarity search.  Since our embedding preserves Gaussian kernel
distances on the manifold, standard Euclidean nearest-neighbor data structures may be
applied after the RFF map.  The precise statement is given in
Appendix~\ref{app:kernel-nearest-neighbor}.

\begin{corollary}
Let \(P\subset\mathcal M\subset\mathbb R^N\) and \(q\in\mathcal M\), where \(\mathcal M\) is a compact \(d\)-dimensional \(C^2\) submanifold with positive reach.  Let \(\phi:\mathbb R^N\to\mathbb R^{2t}\) be the RFF map from Theorem~\ref{thm:intro_uniform} with \(t\) satisfying the bound in that theorem.  Then, with probability at least \(1-\delta\),
\[
(1-2\varepsilon)D_{K_\sigma}(p,q)^2
\le \|\phi(p)-\phi(q)\|_2^2
\le (1+2\varepsilon)D_{K_\sigma}(p,q)^2
\qquad\text{for all }p,q\in\mathcal M.
\]
Consequently, if
\[
\widehat p = \arg\min_{p\in P}\|\phi(p)-\phi(q)\|_2
\]
is the Euclidean nearest neighbor of \(\phi(q)\) in \(\phi(P)\), then
\[
D_{K_\sigma}(\widehat p,q)
\;\le\;
\sqrt{\frac{1+2\varepsilon}{1-2\varepsilon}}\;
\min_{p\in P} D_{K_\sigma}(p,q)
\;=\;
(1+O(\varepsilon))\min_{p\in P} D_{K_\sigma}(p,q).
\]
Thus Euclidean nearest-neighbor search on \(\phi(P)\subset\mathbb R^{2t}\) yields a \((1+O(\varepsilon))\)-approximate nearest neighbor with respect to the original Gaussian kernel distance, with the feature dimension controlled by the intrinsic geometry of \(\mathcal M\) rather than by the ambient dimension \(N\).
\end{corollary}

\paragraph{Kernel SVM.}
For kernel SVMs, the margin
is a distance-to-boundary quantity in the RKHS.  Since our embedding
preserves RKHS distances on \(\mathcal M\), it preserves the large-margin
geometry up to \(1\pm O(\varepsilon)\).  Thus the kernel SVM can be replaced
by a linear SVM in \(\mathbb R^t\), reducing Gram-matrix storage
\(O(n^2)\) to feature storage \(O(nt)\), and replacing kernel-SVM
optimization by linear-SVM optimization on the random features.

\subsection{Roadmap}
\label{subsec:roadmap}

The remainder of the paper is organized as follows.
Section~\ref{sec:experiment} gives a brief empirical illustration of the random-feature construction.
Section~\ref{sec: prelim} collects background on kernel distances and random Fourier features, the Gaussian kernel power distance, persistent homology and interleaving, and the manifold regularity assumptions used throughout.
Section~\ref{subsec:overview} presents the local strategy for the squared-distance ratio $R(\Delta)$: pointwise concentration, gradient control, and uniform bounds on Euclidean patches.
Section~\ref{sec:properties_R} develops the relative Lipschitz-type estimates for~$R$ needed near degeneracies.
Section~\ref{subsec:local-charts} lifts these local bounds to chord differences on~$\mathcal{M}$ via intrinsic nets and reach, completing the proof of Theorem~\ref{thm:intro_uniform}.
Section~\ref{sec:approx-persistent-modules} gives the complete proof of Theorem~\ref{thm:intro_interleaving} (stated above as one of our main contributions in Section~\ref{subsec:intro_main_result}), transferring the kernel-distance guarantees to weighted Čech and Rips filtrations in the GKPD setting.
Section~\ref{sec:app-kernel-values} turns to uniform additive approximation of Gaussian kernel values on~$\mathcal{M}$ and proves Theorem~\ref{thm:additive_uniform}.
Section~\ref{app:applications} spells out the further application corollaries previewed in Section~\ref{subsec:applications}: kernel $k$-means clustering, kernel-distance matching, and kernel nearest-neighbor search, together with the Maximum Mean Discrepancy preservation result of Section~\ref{sec:mmd}.
Appendix~\ref{sec:conc-ineq-gcc} proves the concentration inequality of \cite{boissonnat2024euclidean} invoked in the proof of Theorem~\ref{thm:intro_uniform} (Section~\ref{subsec:proof_overview}).

\section{Numerical experiment: ambient dimension versus RFF dimension}
\label{sec:experiment}

The purpose of this experiment is not to compare the numerical constants in the estimate of~\cite{chen2017relative} with the constants in Theorem~\ref{thm:intro_uniform}.  The constants in both bounds are conservative and are not expected to predict practical feature counts.  Instead, the experiment is designed to test the qualitative dependence on the ambient dimension.  The bound in~\cite{chen2017relative} gives a relative Gaussian-kernel-distance guarantee whose required number of random features contains an ambient-dimensional factor, schematically of the form
\begin{equation}
    T
    \gtrsim
    \frac{N}{\varepsilon^{2}}
    \log\!\left(\frac{N r}{\varepsilon\delta}\right),
    \label{eq:experiment-cp-schematic}
\end{equation}
where \(N\) is the ambient dimension.  In contrast, Theorem~\ref{thm:intro_uniform} predicts that for a fixed-dimensional manifold with controlled geometry, the leading dependence is governed by the intrinsic dimension and geometric quantities of the manifold, rather than by a linear dependence on the ambient dimension.

We test this distinction on the slow helix in \(\mathbb{R}^{N}\).  For even \(N\), define
\begin{equation}
\begin{aligned}
    \gamma_N(s)
    &=
    \frac{1}{\sqrt{N}}
    \bigl(
       \cos(s/N),\sin(s/N),
       \cos(2s/N),\sin(2s/N),
       \ldots, \\
    &\qquad\qquad
       \cos((N/2)s/N),\sin((N/2)s/N)
    \bigr),
    \qquad s\in[0,2\pi].
\end{aligned}
    \label{eq:experiment-slow-helix}
\end{equation}
This is a one-dimensional curve, but it uses all \(N\) ambient coordinates.  Thus the example separates intrinsic dimension from ambient dimension: the intrinsic parameter dimension remains one, while the ambient space dimension is allowed to grow.

For each ambient dimension \(N\), we sample points \(x_1,\ldots,x_n\) on the curve and use the Gaussian kernel with bandwidth \(\sigma=1\).  The exact Gaussian kernel distance is
\begin{equation}
    D_K(x,y)
    =
    \sqrt{
        2\left(
            1-\exp\left(-\frac{\norm{x-y}^{2}}{2\sigma^{2}}\right)
        \right)
    }.
    \label{eq:experiment-kernel-distance}
\end{equation}
Given an RFF map \(\phi_T:\mathbb{R}^{N}\to\mathbb{R}^{T}\), we measure the maximum sampled relative kernel-distance distortion
\begin{equation}
    \operatorname{err}_{T,N}
    :=
    \max_{i<j}
    \left|
       \frac{\norm{\phi_T(x_i)-\phi_T(x_j)}}{D_K(x_i,x_j)}
       -1
    \right|.
    \label{eq:experiment-error}
\end{equation}
For each pair \((N,T)\), the experiment is repeated over independent random Fourier feature seeds.  We record the empirical \(95\)th percentile of \(\operatorname{err}_{T,N}\), and define the empirical target dimension threshold
\begin{equation}
    T_{\min}(N)
    :=
    \min
    \left\{
        T:
        q_{0.95}\bigl(\operatorname{err}_{T,N}\bigr)
        \leq 10^{-2}
    \right\}.
    \label{eq:experiment-threshold}
\end{equation}
The value \(T_{\min}(N)\) is found by an exponential scan followed by a binary search.  The shaded region in the figures below shows the final binary-search bracket.

\begin{figure}[H]
    \centering
    \includegraphics[width=0.72\linewidth]{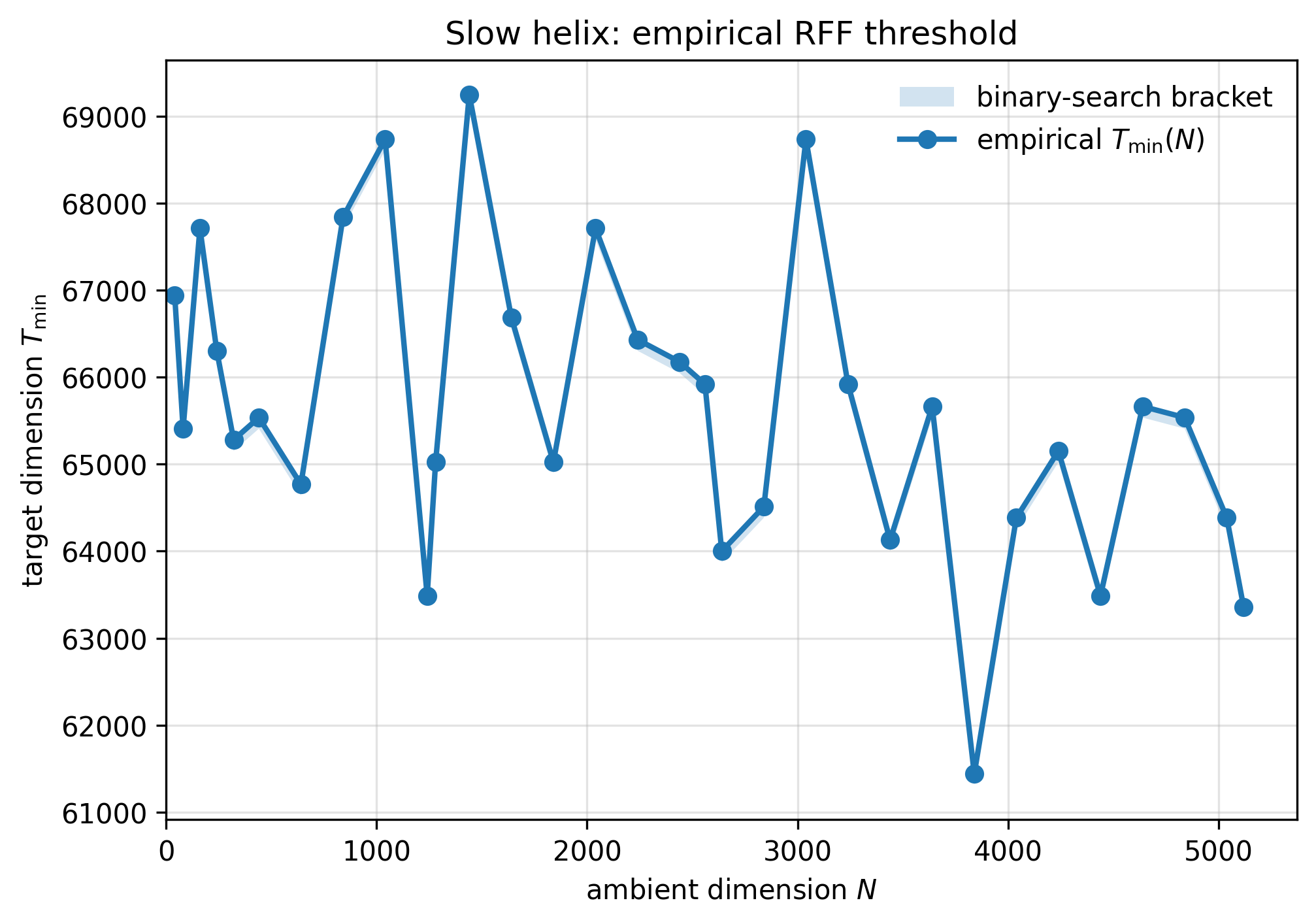}
    \caption{
    Empirical RFF threshold on the slow helix, shown on ordinary axes.  The horizontal axis is the ambient dimension \(N\), and the vertical axis is the empirical target dimension \(T_{\min}(N)\) required to make the empirical \(95\)th percentile of the maximum sampled relative Gaussian-kernel-distance error at most \(10^{-2}\).  The shaded band shows the final binary-search bracket.
    }
    \label{fig:slow-helix-threshold-ordinary}
\end{figure}

\begin{figure}[H]
    \centering
    \includegraphics[width=0.72\linewidth]{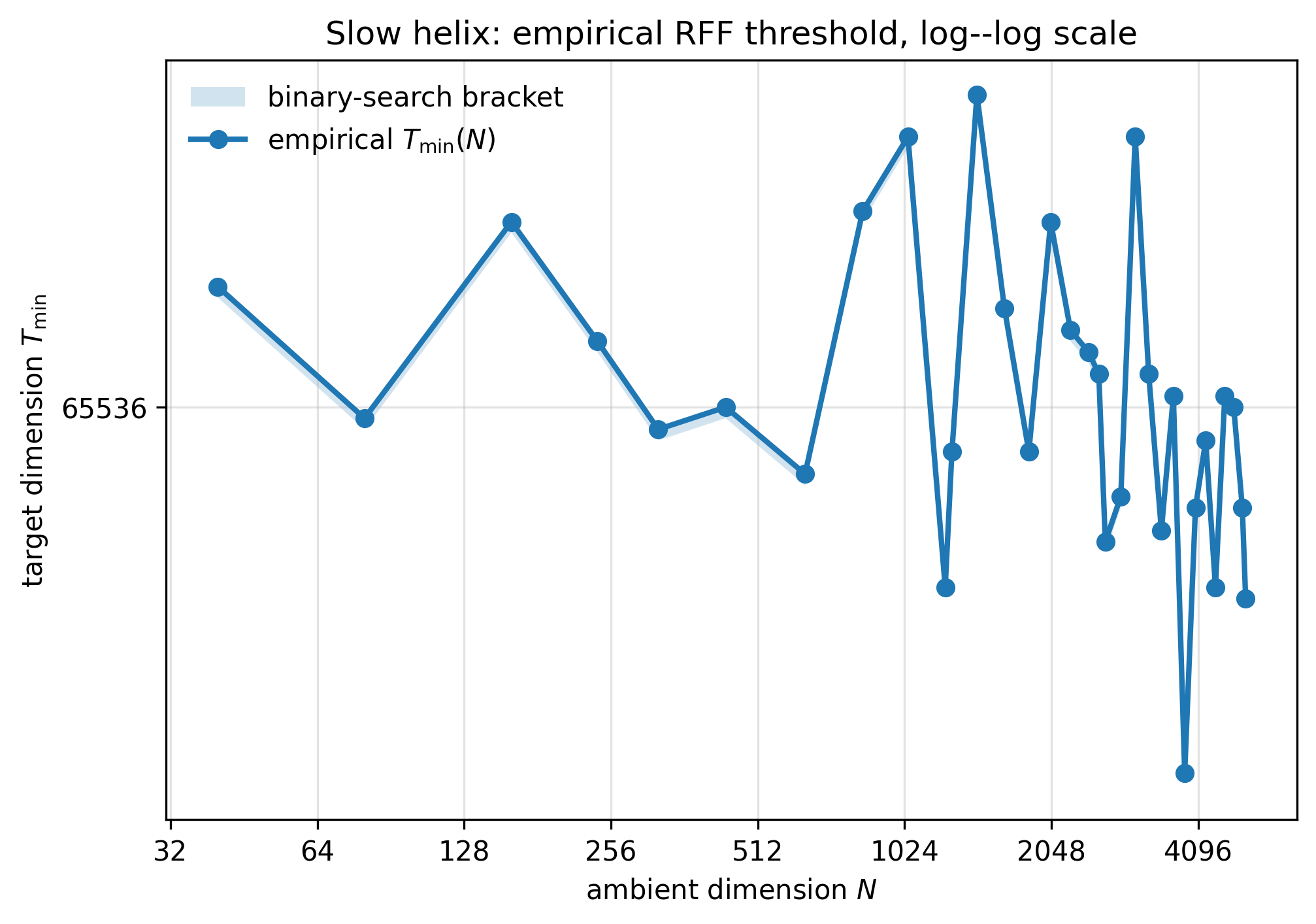}
    \caption{
    The same empirical thresholds plotted on log--log axes.  Both the ambient dimension \(N\) and the target dimension \(T_{\min}(N)\) are displayed on logarithmic scales.  This view is useful for assessing whether the observed growth is compatible with ambient-linear scaling.  The experiment does not show the linear-in-\(N\) growth suggested by applying an ambient-dimensional estimate of the type in~\cite{chen2017relative} to this one-dimensional manifold family.
    }
    \label{fig:slow-helix-threshold-loglog}
\end{figure}

The experiment supports the intrinsic-geometric interpretation of Theorem~\ref{thm:intro_uniform}.  Increasing the ambient dimension by adding more slow-helix coordinates does not force the empirical target dimension to grow linearly with \(N\).  The relevant point is therefore not that the empirical constants match the theoretical constants, but rather that the observed scaling is qualitatively inconsistent with an ambient-linear dependence for this family of low-dimensional manifolds.

This numerical evidence should be interpreted with the usual finite-sample caveats.  The maximum in \eqref{eq:experiment-error} is taken over a sampled point set rather than over all pairs of points on the curve, and the empirical \(95\)th percentile over random seeds is not a substitute for a symbolic high-probability theorem.  Nevertheless, the experiment illustrates the phenomenon captured by the manifold theorem: for structured low-dimensional data embedded in high ambient dimension, the number of random Fourier features needed to preserve Gaussian kernel distances can be governed by intrinsic geometry rather than by the ambient dimension itself. 
\section{Background and Preliminaries}
\label{sec: prelim}

The main theorem and its proof rely on standard notions from kernel methods, probability, and manifold geometry. This section collects the definitions and facts needed to state the approximation guarantees and to carry out the local-to-global argument. The material on kernel distances and random features is used throughout the kernel approximation analysis; the part on manifold geometry supports the net construction and the control of differences on the manifold.

\subsection{Probabilistic and concentration lemmas}
\label{sec:app-prob}

The following lemmas are standard probabilistic tools used in the main text to control concentration of the kernel distance ratio and of gradient terms. They support the pointwise and gradient-concentration steps of the local approximation argument and the net-based union bound.
Later sections rely on them: the local kernel-ratio analysis is in Section~\ref{subsec:overview}, analytic properties of the ratio in Section~\ref{sec:properties_R}, chart arguments in Section~\ref{subsec:local-charts}, and additive uniform kernel-value bounds (including Theorem~\ref{thm:additive_uniform}) in Section~\ref{sec:app-kernel-values}.

\begin{lemma}[Generalized Hoeffding's Inequality]
\label{lemma:generalized-hoeffding}
Let \( \{X_i\}_{i=1}^k \) be independent, identically distributed (i.i.d.) zero-mean random variables with \( \|X_i\|_{\psi_2} \le K \) for all \( i \). Then for any \( \varepsilon > 0 \),
\[ 
\mathbb{P}\left(\left|\sum_{i=1}^k X_i\right| \ge \varepsilon\right)
\;\le\;
2\exp\left(-\frac{c_H \varepsilon^2}{k K^2}\right),
\]
where \( c_H > 0 \) is a universal constant (which may be taken as \( c_H = 1/8 \)).
\end{lemma}

\begin{lemma}[Stein's Lemma~\cite{vershynin2009high}]\label{lem:stein}
Let $g \sim \mathcal{N}(0, I_q)$ be a standard Gaussian vector in $\mathbb{R}^q$, and let $f : \mathbb{R}^q \to \mathbb{R}$ be a differentiable function such that both $f(g)$ and $\nabla f(g)$ are integrable. Then $\mathbb{E}[ f(g) \cdot g ] = \mathbb{E} [ \nabla f(g) ]$, where the expectation on the left is vector-valued and taken coordinatewise.
\end{lemma}

\begin{lemma}[Probabilistic Cauchy-Schwarz in Subspaces]
\label{Probabilistic Cauchy-Schwarz in Subspaces}
Let \( w_1, \dots, w_{k} \overset{\text{iid}}{\sim} \mathcal{N}(0, I_N) \) be independent standard Gaussian vectors in \( \mathbb{R}^N \), and let \( x \in \mathbb{R}^N \) be a fixed vector lying in a \( d \)-dimensional subspace of \( \mathbb{R}^N \). Define \(M := \frac{1}{k} \sum_{i=1}^{k} \abs{\langle w_i, x \rangle}^2\). Then there exists a universal constant \( C > 0 \) such that
\[
\mathbb{P} \left( M > 2d \|x\|^2 \right) \leq 2 \exp\left( -\frac{k d}{C} \right).
\]
\end{lemma}

\begin{proof}
Let \( x \in \mathbb{R}^N \) be supported on its first \( d \) coordinates (i.e., \( x = (x_1, \dots, x_d, 0, \dots, 0) \)). This is without loss of generality because Gaussians are rotationally invariant, so any \( d \)-dimensional subspace can be rotated to align with the first \( d \) coordinates. Now decompose the inner product. For \( w_i \sim \mathcal{N}(0, I_N) \), write \( w_i = (w_i^{(d)}, w_i^{(N-d)}) \), where \( w_i^{(d)} \) is the first \( d \) coordinates. Since \( x \) has no support beyond the first \( d \) coordinates:
\[
M = \frac{1}{k} \sum_{i=1}^k |\langle w_i^{(d)}, x \rangle|^2 \leq \frac{1}{k} \sum_{i=1}^k ||w_i^{(d)}||^2 \, ||x||^2.
\]
Since \( w_i^{(d)} \sim \mathcal{N}(0, I_d) \), the Lemma~\ref{Concentration for Gaussian Norm Sums} gives:
\[
\mathbb{P}\left( \frac{1}{k} \sum_{i=1}^k \|w_i^{(d)}\|^2 \geq 2d \right) \leq 2 \exp\left(-\frac{k d}{C}\right).
\]
Thus for \( w_1, \dots, w_k \sim \mathcal{N}(0, I_N) \) and any fixed \( d \)-dimensional subspace \( V \), we have:
\[
\mathbb{P}\left( \sup_{x \in V} \frac{1}{k} \sum_{i=1}^k |\langle w_i, x \rangle|^2 \leq 2d \|x\|^2 \right) \geq 1 - 2 \exp\left(-\frac{k d}{C}\right),
\]
where \( C > 0 \) is a universal constant. This holds simultaneously for all \( x \in V \). 

\end{proof}

\begin{lemma}[Concentration for Gaussian Norm Sums]
\label{Concentration for Gaussian Norm Sums}
Let \( w_i \sim \mathcal{N}(0, 1) \) be i.i.d.\ standard Gaussian vectors in \( \mathbb{R}^N \). Then for any integer \( L \geq 1 \), we have
\[
\mathbb{P}\left[ \frac{1}{L} \sum_{i=1}^{L} \|w_i\|^2 \geq 2N \right] \leq 2 \exp\left(- \frac{NL}{C_4^4} \right),
\]
for some absolute constant \( C_4 > 0 \).
\end{lemma}

\begin{proof}
\textbf{Expectation \& Gaussian–norm calculation.} $\frac{w_{i}}{\sqrt{L}}\sim\mathcal N\!\bigl(0,\tfrac1L I_{N}\bigr)$, $\Bigl\lVert\tfrac{w_{i}}{\sqrt{L}}\Bigr\rVert_{\psi_{2}} = \frac{c}{\sqrt{L}}$, $\Bigl\lVert\!\bigl(\tfrac{w_{i}}{\sqrt{L}}\bigr)^{2}\Bigr\rVert_{\psi_{1}} = \frac{c^{2}}{L}$, $\E\Bigl[\bigl(\tfrac{w_{i}}{\sqrt{L}}\bigr)^{2}\Bigr] = \frac1L$.

\textbf{Concentration calculation.}
First centre the variable:
\[
   \sum_{i=1}^{k=LN}\!\Bigl(\tfrac{w_{i}}{\sqrt{L}}\Bigr)^{2}
   \;=\;
   \sum_{i=1}^{k}
      \Bigl[\!\bigl(\tfrac{w_{i}}{\sqrt{L}}\bigr)^{2}-\tfrac1L\Bigr]
   \;+\;\frac{k}{L},
\]
a mean-zero, sub-exponential sum with norm $1/L$.  
By Bernstein's inequality,
\[
   \Pr\!\Bigl[
       \sum_{i=1}^{k}
       \Bigl(\!\bigl(\tfrac{w_{i}}{\sqrt{L}}\bigr)^{2}-\tfrac1L\Bigr)
       \;\ge\; N
     \Bigr]
   \;\le\;
   2\exp\!\bigl(-N L / c_{4}^{4}\bigr).
\]
Hence
\[
   \Pr\!\Bigl[
       \frac1L \sum_{i=1}^{LN} w_{i}^{2}\;\ge\;2N
     \Bigr]
   \;\le\;
   2\exp\!\bigl(-N L / c_{4}^{4}\bigr).
\]
\end{proof}

\subsection{Kernel Distance and Random Fourier Features (RFF)}
\label{subsec:kernel_rff}

\paragraph*{Kernel Distance.}
Let \( K: \mathbb{R}^N \times \mathbb{R}^N \to \mathbb{R} \) be a positive-definite kernel. For any two points \( x, y \in \mathbb{R}^N \), the \emph{kernel distance} \( D_K \) induced by \( K \) is
\(
D_K(x,y) := \sqrt{K(x,x) + K(y,y) - 2K(x,y)}.
\)
When \( K \) is the Gaussian kernel with bandwidth \( \sigma > 0 \), i.e.\ \(K_\sigma(x,y) = \exp\!\bigl(-\|x-y\|^2/(2\sigma^2)\bigr)\), the squared kernel distance simplifies to
\(
D_{K_\sigma}(x,y)^2 = 2\bigl(1 - \exp(-\|x-y\|^2/(2\sigma^2))\bigr).
\)
We will frequently use the shorthand \( D_K(\Delta) \) where \( \Delta = x - y \).

\paragraph*{Random Fourier Features (RFF).}
To efficiently approximate shift-invariant kernels like the Gaussian, we employ \emph{Random Fourier Features} (RFF)~\cite{rahimi2007random}. The method proceeds as follows:

\begin{enumerate}
    \item Sample \( t \) i.i.d. frequencies \( \omega^1, \ldots, \omega^t \sim \mathcal{N}(0, \sigma^{-2}I_N) \)
    \item Construct a randomized feature map \( \phi: \mathbb{R}^N \to \mathbb{R}^t \), such as:
    \[
    \phi(x) = \frac{1}{\sqrt{t}}\left[
        \cos(\langle \omega^1, x\rangle), \sin(\langle \omega^1, x\rangle), \ldots, 
        \cos(\langle \omega^{t/2}, x\rangle), \sin(\langle \omega^{t/2}, x\rangle)
    \right]
    \]
    \item Approximate the kernel via \( \hat{K}(x,y) = \langle \phi(x), \phi(y) \rangle \)
\end{enumerate}

This construction yields an unbiased estimator of the true kernel: \( \mathbb{E}[\hat{K}(x,y)] = K(x,y) \).

\paragraph*{Approximate Kernel Distance.}
The RFF approximation induces
\[
   D_{\hat{K}}(x,y) := \sqrt{\hat{K}(x,x) + \hat{K}(y,y) - 2\hat{K}(x,y)}.
\]
Our key objective is to control the relative approximation error \(D_{\hat{K}}(x,y)^2/D_K(x,y)^2 \in [1-\varepsilon, 1+\varepsilon]\) with high probability.

\begin{lemma}[Relative Error Guarantee~\cite{chen2017relative}]
\label{lem:ChenPhilorg}
Let \( B \geq 0 \) bound the maximum pairwise distance \( \|x-y\| \leq \sigma B \). Then for the Gaussian kernel approximation:
\[
\mathbb{P}\left(
    \frac{D_{\hat{K}}(x,y)^2}{D_K(x,y)^2} \in [1-\varepsilon, 1+\varepsilon]
\right)
\geq 1 - O\left(
    \frac{dB}{\varepsilon} \exp\!\left(-\frac{t\varepsilon^2}{d}\right)
\right).
\]
\end{lemma}

\begin{corollary}[Sample Complexity]
\label{cor:rff_sample_complexity}
For any \( \varepsilon, \delta \in (0,1) \), if the number of random features satisfies
\[
t = \Omega\left(
    \frac{d}{\varepsilon^2} \log\left(\frac{d}{\varepsilon\delta}\right)
\right),
\]
then uniformly for all \( \|\Delta\| \leq 1 \):
\[
\frac{D_{\hat{K}}(\Delta)}{D_K(\Delta)} \in [1-\varepsilon, 1+\varepsilon]
\quad \text{with probability at least } 1-\delta.
\]
\end{corollary}

\begin{definition}[Kernel weight]
\label{def:kernel-weights-avg}
For a point set $P$, the kernel weight is $w(p)=\frac{1}{|P|}\sum_{y\in P}D_K^2(p,y)-\frac{1}{2|P|^2}\sum_{x,y\in P}D_K^2(x,y)$.
\end{definition}

\subsection{Gaussian Kernel Power Distance (GKPD)}
\label{subsec:gkpd}

\paragraph*{Definition and Formulation.}
Let $P \subset \mathbb{R}^N$ be a finite point set with empirical measure $\mu$. The \textbf{Gaussian Kernel Power Distance (GKPD)}~\cite{phillips1geometric} is defined as
\[
f^K_{\mu}(x)^2 := \min_{p \in P} \left(D_K^2(x,p) - w(p)\right).
\]
Here $D_K(x,y)$ is the \textbf{Gaussian kernel distance}, with squared form $D_K(x,y)^2 = 2(1 - K(x,y)) = 2\bigl(1 - e^{-\|x-y\|^2 / 2}\bigr)$. The \textbf{kernel weight function} at $p$ is
\[
w(p) := -D_K^2(\mu,p) = -\Biggl(\frac{1}{|P|}\sum_{y\in P} D_K^2(p,y)-\frac{1}{2|P|^2}\sum_{x,y\in P} D_K^2(x,y)\Biggr).
\]
The GKPD is the minimal squared power distance to the weighted set $\widehat{P} = \{(p, w(p)) : p \in P\}$.

\paragraph*{Pairwise GKPD.}
\begin{definition}[Pairwise Gaussian Kernel Power Distance]
\label{def:weighted-pairwise-gkpd}
Let $\widehat{p}_i = (p_i, w(p_i))$ and $\widehat{p}_j = (p_j, w(p_j))$ be weighted points in $\widehat{P}$, where $w(p)$ is the kernel weight from Definition~\ref{def:kernel-weights-avg}.
The \textbf{pairwise Gaussian Kernel Power Distance} between $\widehat{p}_i$ and $\widehat{p}_j$ is
\[
D_K^2(\widehat{p}_i, \widehat{p}_j) := D_K^2(p_i, p_j) - w(p_i) - w(p_j),
\]
with $D_K(p_i, p_j)$ as defined above.
\end{definition}

\paragraph*{Unweighted Čech Complex.}
For an unweighted set $\sigma \subset P$,
\[
\mathrm{rad}(\sigma) = \min_{x \in \mathbb{R}^D} \max_{p_i \in \sigma} \|x - p_i\|.
\]
A simplex $\sigma$ belongs to $\check{C}_\alpha(P)$ if and only if $\mathrm{rad}(\sigma) \leq \alpha$.

\paragraph*{Weighted Čech Complex (Power Distance).}
For weighted points $\hat{X} = \{\hat{p}_1, \ldots, \hat{p}_n\}$ with $\hat{p}_i = (p_i, w(i))$,
\[
\mathrm{rad}^2(\hat{X}) = \min_{x \in \mathbb{R}^D} \max_{\hat{p}_i \in \hat{X}} D(x, \hat{p}_i),
\]
where the power distance is $D(x, \hat{p}_i) = \|x - p_i\|^2 - w(i)$.

\paragraph*{GKPD Context.}
In a Hilbert space $H$ (or RKHS $H_K$), the squared radius for a weighted simplex $\hat{\sigma} \subset \hat{P}$ is defined using the power distance $D(\hat{x}, \hat{p}_i) = \|x - p_i\|^2_H - w(p_i)$.

\begin{definition}[GKPD distortion map]
\label{def:distortion-map}
A map $f : (\mathbb{R}^D, D_K) \to (\mathbb{R}^{2t}, \|\cdot\|)$ is an \textbf{$(\varepsilon,\eta)$-distortion map for the GKPD} if it satisfies:

\begin{enumerate}
    \item \textbf{Pairwise Distance Preservation:} For all $x, y \in P$,
    \[
    (1 - \varepsilon)D_K(x, y)^2 - \eta \leq \|f(x) - f(y)\|^2 \leq (1 + \varepsilon)D_K(x, y)^2 + \eta.
    \]
    
    \item \textbf{Weight Function Preservation:} For all $x \in P$,
    \[
    |w(f(x)) - w(x)| \leq \varepsilon|w(x)| + \eta,
    \]
    where $w(f(x))$ is recomputed in $\mathbb{R}^{2t}$ using Euclidean distances.
\end{enumerate}
\end{definition}

\begin{lemma}[Simplex Distortion Lemma~\cite{boissonnat2024euclidean}]
\label{lem:simplex-distortion}
Let $\widehat{\sigma} \subset \widehat{P}$ be a simplex in the weighted Čech complex $\check{C}_\alpha(\widehat{P})$ built from the GKPD $D_K^2(\widehat{p}_i,\widehat{p}_j)$.
Let $f$ be an $(\varepsilon,\eta)$-distortion map for the pairwise GKPD.
Write $\widehat{f(\sigma)}$ for the image of $\sigma$ in $\mathbb{R}^{2t}$ \textup{(}with weights recomputed\textup{)}.
Then
\[
(1-\varepsilon)\left(\mathrm{rad}^2(\widehat{\sigma})-\eta\right)
\leq
\mathrm{rad}^2\left(\widehat{f(\sigma)}\right)
\leq
(1+\varepsilon)\left(\mathrm{rad}^2(\widehat{\sigma})+\eta\right).
\]
\end{lemma}

Recent work~\cite{boissonnat2024euclidean} demonstrates that Random Fourier Feature (RFF) embeddings can approximate both pairwise kernel distances and weights with relative error, enabling efficient dimensionality reduction for kernel-based persistent homology.

\subsection{Persistent Homology and Interleaving Distance}
\label{subsec:interleaving}

Persistent homology (PH)~\cite{edelsbrunner2008persistent} is a foundational method in Topological Data Analysis (TDA) that quantifies the multiscale geometric and topological structure of data. 
Given a function or filtration $\{F_\alpha\}_{\alpha\ge0}$—for example, the sublevel sets of a distance function or a kernel power distance—PH tracks how topological features such as connected components, holes, and voids appear and disappear as the scale parameter $\alpha$ varies. 
The resulting summary, called a \emph{persistence module}, records these birth–death events algebraically, and its canonical visualization is the \emph{persistence diagram} (PD), a multiset of points $(b,d)$ where each point represents a topological feature born at scale $b$ and dying at scale $d$.

To compare filtrations or persistence modules, we use the \emph{interleaving distance}. 
Given two persistence modules $M$ and $N$, an $\varepsilon$–interleaving consists of morphisms 
$f_t : M_t \to N_{t+\varepsilon}$ and $g_t : N_t \to M_{t+\varepsilon}$ such that composing them shifts indices by at most $2\varepsilon$. 
The smallest such $\varepsilon$ defines the \emph{interleaving distance} $d_I(M,N)$, which measures how much the modules can be “shifted” to align. 
A related geometric notion is the \emph{bottleneck distance} $d_B(D(f), D(g))$ between persistence diagrams $D(f)$ and $D(g)$, defined as the minimal cost of optimally matching their points. 
A fundamental result, the \emph{Stability Theorem}~\cite{cohen2005stability}, states that \(d_B(D(f), D(g)) \le d_I(M(f), M(g))\), ensuring that small perturbations in the input metric or kernel function lead to only small changes in the output persistence diagrams. 
This stability principle underlies the robustness of all topological approximations presented in this work.

\subsection{Notation for Persistent Homology and GKPD}
\noindent We write $\mathcal{H}$ for the RKHS and $\psi: \mathbb{R}^N \to \mathcal{H}$ for the canonical feature map.
The kernel centroid is $\mu_P = \frac{1}{|P|}\sum_{y\in P}\psi(y)$.
Kernel weights are $w(p) = -\|\psi(p)-\mu_P\|^2$ and $\widehat{w}(p)$; $\widehat{P}$ denotes a weighted point set.
We set $c_P = \frac{2}{(1-\|\mu_P\|)^2}$ and $\varepsilon_\star = \max\{2\varepsilon, \varepsilon c_P\}$.
The map $f$ is as in Definition~\ref{def:distortion-map}.
Filtrations $\check{C}_\alpha(\widehat{P})$ and $VR_\alpha(\widehat{P})$ are the weighted \v{C}ech and Rips complexes; $d_I(M, N)$ and $d_B(D(f), D(g))$ denote interleaving and bottleneck distance.

\subsection{Probabilistic Inequalities}
\label{subsec:hoeffding}

We begin by recalling the definition of sub-Gaussian random variables and their associated norms, which quantify the tail decay behavior.

\begin{definition}[Sub-Gaussian Norm]
\label{def:subgaussian}
A random variable \( X \) is called \emph{sub-Gaussian} if there exists a constant \( \alpha > 0 \) such that for all \( t > 0 \),
\[ 
\mathbb{P}(|X| > t) \;\le\; 2\exp(-t^2/\alpha^2).
\]
The \emph{sub-Gaussian norm} (or \(\psi_2\)-norm) of \( X \) is defined as
\[ 
\|X\|_{\psi_2} 
\;:=\; 
\inf\left\{ t > 0 \mid \mathbb{E}\left[\exp\left(X^2/t^2\right)\right] \le 2 \right\}.
\]
\end{definition}

\begin{remark}
The sub-Gaussian norm captures the "Gaussian-like" tail behavior of \( X \). For a standard normal random variable \( Z \sim \mathcal{N}(0,1) \), we have \( \|Z\|_{\psi_2} \approx 1. \)
\end{remark}

Our analysis relies on several probabilistic tools for controlling concentration and expectations. We use the following concentration inequality for sums of sub-Gaussian random variables (Lemma~\ref{lemma:generalized-hoeffding}); Lemma~\ref{lem:stein} and Lemma~\ref{Probabilistic Cauchy-Schwarz in Subspaces} are stated and proved in Section~\ref{sec:app-prob}.

\subsection{Manifold Geometry}
\label{subsec:manifold-net}

\begin{definition}[Reach of a set~\cite{federer1959curvature}]
Let \( A \subset \mathbb{R}^N \) be a closed set. The \emph{reach} of \( A \), denoted \( \rch(A) \), is defined as the largest number \( \tau \ge 0 \) such that every point \( x \in \mathbb{R}^N \) with \( \operatorname{dist}(x, A) < \tau \) has a unique nearest point in \( A \). That is,
\[
\rch(A) := \sup\left\{ r > 0 \;\middle|\; \forall x \in \mathbb{R}^N,\, \operatorname{dist}(x,A) < r \Rightarrow \exists! \, a \in A \text{ with } \|x - a\| = \operatorname{dist}(x,A) \right\}.
\]
\end{definition}

The following result characterizes the size of a dense point set (net) needed to cover a manifold, adapting results from~\cite{niyogi2008finding}. We present both the original probabilistic formulation and a simplified deterministic version.

\begin{proposition}[Probabilistic Net Construction]
\label{prop:probabilistic-net}
Let $\mathcal{M} \subset \mathbb{R}^N$ be a $d$-dimensional manifold with reach $\rch(\mathcal{M})>0$. Let $\bar{x}$ be obtained by uniform random sampling from $\mathcal{M}$. Then for any $\delta > 0$ and $\varepsilon < \rch(\mathcal{M})/2$, the set $\bar{x}$ is $(\varepsilon/2)$-dense in $\mathcal{M}$ (meaning for every $p \in \mathcal{M}$ there exists $x_i \in \bar{x}$ with $\|p - x_i\| < \varepsilon/2$) with probability at least $1-\delta$, provided:

\[
|\bar{x}| > \beta_1 \left( \log(\beta_2) + \log\left(\frac{1}{\delta}\right) \right),
\]

where the constants are:
\[
\beta_1 = \frac{\vol(\mathcal{M})}{\cos^d(\theta_1)\vol(B^d_{\varepsilon/4})}, \quad 
\beta_2 = \frac{\vol(\mathcal{M})}{\cos^d(\theta_2)\vol(B^d_{\varepsilon/8})},
\]
with $\theta_1 = \arcsin(\varepsilon/8\rch(\mathcal{M}))$, $\theta_2 = \arcsin(\varepsilon/16\rch(\mathcal{M}))$, and $B^d_r$ denoting the $d$-dimensional Euclidean ball of radius $r$.
\end{proposition}

For our purposes, we require the following deterministic version that provides explicit bounds on the net size:

\begin{lemma}[Deterministic Net Size Bound]
\label{prop:manifold-net}
Let $\mathcal{M} \subset \mathbb{R}^N$ be a $d$-dimensional manifold with reach $\rch(\mathcal{M})>0$. For any $0 < l < \rch(\mathcal{M})/2$, there exists a net $\Gamma_l \subset \mathcal{M}$ such that:

\begin{romanenumerate}
  \item $\Gamma_l$ is $(l/2)$-dense in $\mathcal{M}$: for every $p \in \mathcal{M}$, there exists $x_j \in \Gamma_l$ with $\|p - x_j\| \le l/2$.
  
  \item The cardinality satisfies:
  \[
  \abs{\Gamma_l} \approx \frac{\vol(\mathcal{M})}{(l/5)^dV_d} \log\left(\frac{\vol(\mathcal{M})}{(l/9)^dV_d}\right)
  \]
  where $V_d$ is the volume of a unit $d$-ball.
\end{romanenumerate}
\end{lemma}

\begin{remark}
The original result in~\cite{niyogi2008finding} includes additional probabilistic aspects related to random sampling and curvature-dependent angle bounds. Here we present a simplified deterministic version focusing on the explicit dependence on manifold volume, dimension.
\end{remark}

\begin{lemma}[Tangent Space Approximation~\cite{boissonnat2019reach}]
\label{lem:distance_tangent_space}
Let \( \mathcal{M} \subset \mathbb{R}^N \) be a manifold with reach \( \rch(\mathcal{M}) \). For any \( p, q \in \mathcal{M} \) with \( \|p - q\| < \rch(\mathcal{M}) \):
\begin{enumerate}
    \item The angle between the secant \( [pq] \) and \( T_p\mathcal{M} \) satisfies
    \[
    \sin \angle([pq], T_p\mathcal{M}) \leq \frac{\|p - q\|}{2\rch(\mathcal{M})}.
    \]
    \item The Euclidean distance from \( q \) to \( T_p\mathcal{M} \) is bounded by
    \[
    d_{\mathcal{E}}(q, T_p\mathcal{M}) \leq \frac{\|p - q\|^2}{2\rch(\mathcal{M})}.
    \]
\end{enumerate}
\end{lemma}
\subsection*{Notation}
Throughout this paper, we use the following notation: $\mathbb{R},\ \mathbb{R}^N,\ \mathbb{R}^d$ (real line and Euclidean spaces), where $N$ always denotes the \emph{ambient} dimension of the embedding space $\mathbb{R}^N$; $\|x\|$ (Euclidean norm), $B^d_r(p)$ and $B_r(p)$ ($d$-dimensional and contextual closed Euclidean balls), $B_{\mathcal{M}}(p,r)$ (intrinsic ball on manifold $\mathcal{M}$), $\mathcal{M}\subset\mathbb{R}^N$ (compact $d$–dimensional $\mathcal C^2$ submanifold), $\rch(\mathcal{M})$ (reach), $\vol(\mathcal{M})$ and $\diam(\mathcal{M})$ (volume and diameter), $K_\sigma(x,y)=\exp\!\bigl(-\|x-y\|^2/(2\sigma^2)\bigr)$ (Gaussian kernel), $\hat{K}$ (RFF approximation to $K_\sigma$), $D_K$ (exact Gaussian kernel distance) and $D_{\hat{K}}$ (kernel distance computed from $\hat{K}$), $R(\Delta) := D_{\hat{K}}(\Delta)^2/D_K(\Delta)^2$ (kernel distance ratio), $\omega^1,\dots,\omega^t\sim\mathcal{N}(0,\sigma^{-2}I_N)$ (RFF frequencies; $\mathcal{N}$ is the Gaussian law), $\phi\colon\mathbb{R}^N\to\mathbb{R}^t$ (finite-dimensional RFF feature map) and $\psi\colon\mathbb{R}^N\to\mathcal{H}$ (canonical RKHS feature map for $K_\sigma$), $P$ (finite point set), $\Gamma_r$ ($r$–net with $m:=|\Gamma_r|$ and $\Gamma_r=\{z_1,\dots,z_m\}$, see Lemma~\ref{prop:manifold-net}), $T_p\mathcal{M}$ (tangent space at $p\in\mathcal{M}$), $O(\cdot)$, $\Theta(\cdot)$ (asymptotic notation); $s$ is the slow-helix parameter in Example~\ref{ex:slow-helix}; $k$, $L$, and $q$ denote generic summation lengths in Section~\ref{sec:app-prob} (Stein's lemma uses $g\in\mathbb{R}^q$); $r$ denotes a Euclidean ball radius where distinguished from the net scale in $\Gamma_r$.

\section{Local Approximation Strategy}
\label{subsec:overview}

In this section we work in the ambient space $\mathbb{R}^N$ and study geometric and analytic properties of the squared kernel distance ratio $R(\Delta):=\frac{D_{\hat{K}}(\Delta)^2}{D_K(\Delta)^2}$ induced by Random Fourier Features (RFF). In particular, we develop useful lemmas that help us lift the ratio from tangent planes to local neighborhoods of the manifold. Our concrete objective is to bound the relative error in $R(\Delta)$ on a small ball \emph{not necessarily centered at the origin}. The analysis proceeds in three stages: (i) \textbf{pointwise concentration}: high-probability bounds on $R(\Delta)$ at fixed locations $\Delta\in\mathbb{R}^N$ (Lemma~\ref{lem:ratio bound at a point}); (ii) \textbf{gradient control}: sensitivity of $R(\Delta)$ by bounding partial derivatives (Lemma~\ref{Derivative of the ratio is small}) and proving Lipschitz continuity (Lemma~\ref{Lipschitz Constant of derivative of ratio}; proof follows that lemma); (iii) \textbf{uniform neighborhood guarantees}: combining (i) and (ii) via Taylor expansion, extending pointwise bounds to uniform control over Euclidean balls (Lemma~\ref{Local Uniform Control in Euclidean Balls}; proof follows that lemma). This three-step approach provides the foundation for extending local accuracy guarantees to the entire manifold in subsequent sections.

\subsection{Pointwise Ratio Concentration at a Fixed Location}
\label{Subsec: Bounding ratio at the center of the ball}

We begin by establishing pointwise control of the kernel distance ratio at a fixed offset vector $\Delta$. This foundational result will later be extended to uniform bounds over neighborhoods.

\begin{lemma}[Pointwise concentration of kernel distance ratio]
\label{lem:ratio bound at a point}
Let $\Delta \in \mathbb{R}^N \setminus \{0\}$ and define the random ratio
\[
R(\Delta) := \frac{D_{\hat{K}}(\Delta)^2}{D_K(\Delta)^2} = \frac{1}{t}\sum_{j=1}^t \frac{1 - \cos\langle \omega^j, \Delta\rangle}{1 - e^{-\|\Delta\|^2/2}},
\]
where $\omega^j \sim \mathcal{N}(0,I_N)$. Then for any $\varepsilon > 0$,
\[
\mathbb{P}\left( \left| R(\Delta) - 1 \right| \geq \varepsilon \right)
\leq 2 \exp\left( -\frac{t(1 - e^{-\|\Delta\|^2/2})^2 \varepsilon^2}{2} \right).
\]
\end{lemma}

\begin{proof}[Proof of Lemma~\ref{lem:ratio bound at a point}]\label{app:proof-lem-ratio-pointwise}
First observe the zero mean of \( R(\Delta) \).
Note that, using Stein's lemma,  for $\|\Delta\| \neq 0$,  
\[ 
\mathbb{E}\bigl[D_{\hat{K}}(\Delta) / D_K(\Delta)\bigr]
~=~
1/(1 -  e^{-\|\Delta\|^2/2})\mathbb{E}\Bigl[1 - 
  \frac{1}{t}\sum_{j=1}^t \cos\bigl(\langle \omega^j,\Delta\rangle\bigr)
\Bigr] = 1.
 \]

Further, each term in 
\( \,\bigl(1 - \cos(\langle \omega^j,\Delta\rangle)\bigr)/(1 -  e^{-\|\Delta\|^2/2}) \) 
is bounded by $2/(1 -  e^{-\|\Delta\|^2/2})$ in absolute value. Hence a direct application of Hoeffding’s inequality ensures the result. 
\end{proof}

\begin{remark}
When $\|\Delta\| \geq 1$, the bound simplifies to
\[
\mathbb{P}\left( \left| R(\Delta) - 1 \right| \geq \varepsilon \right)
\leq 2 \exp\left( -\frac{t\varepsilon^2}{13} \right),
\]
since $1 - e^{-1/2} \geq \frac{1}{\sqrt{13}}$.
\end{remark}

Extending pointwise bounds to neighborhoods requires controlling how the ratio varies with its input; the next lemma provides that control.

\subsection{Gradient Concentration of the Distance Ratio}
\label{Subsec: Bounding the gradient of the ratio at the center of the ball}

Having established pointwise control of the kernel distance ratio $R(\Delta)$, we now analyze its gradient to understand how the ratio varies with respect to its input. This gradient analysis will enable us to extend pointwise guarantees to local neighborhoods.

Consider the partial derivative with respect to the first coordinate (the analysis for other coordinates is identical). From~\cite{chen2017relative,rahimi2007random}, the derivative takes the form:

\begin{equation}
\label{eq:ratio-derivative}
\frac{\partial}{\partial \Delta_1} R(\Delta)
= \frac{1}{t} \sum_{k=1}^{t} 
\frac{
  \omega^k_1 \sin(\langle \omega^k, \Delta \rangle)(1 - e^{-\|\Delta\|^2/2})
  - (1 - \cos \langle \omega^k, \Delta \rangle)\Delta_1 e^{-\|\Delta\|^2/2}
}{
  (1 - e^{-\|\Delta\|^2/2})^2
},
\end{equation}

where $\omega^1, \ldots, \omega^t \sim \mathcal{N}(0, I_N)$ are the RFF frequencies.
The key observation is that the numerator is a zero-mean random variable, enabling concentration as $t$ increases.

For $\|\Delta\| \geq 1$, the denominator $1 - e^{-\|\Delta\|^2/2}$ is bounded away from zero, simplifying our analysis. Let us define:

\begin{align*}
I &:= \langle \omega, \Delta \rangle = \omega^\top \Delta, \\
Z &:= \omega_1 \sin(I)(1 - e^{-\|\Delta\|^2/2}) - (1 - \cos(I))\Delta_1 e^{-\|\Delta\|^2/2}.
\end{align*}

\begin{lemma}[Gradient Concentration]
\label{lem:gradient-concentration}
\label{Derivative of the ratio is small}
Let $\Delta \in \mathbb{R}^N$ with $\|\Delta\| \geq 1$. For any coordinate $i \in \{1, \dots, N\}$ and $\varepsilon > 0$,
\[
\mathbb{P}\left( \left| \frac{\partial}{\partial \Delta_i} R(\Delta) \right| \geq \varepsilon \right)
\leq 2 \exp\left( -\frac{(1 - e^{-\|\Delta\|^2/2})^4 c_H \varepsilon^2 t}{16} \right),
\]
where $c_H > 0$ is an absolute constant.
\end{lemma}

\begin{proof}[Proof of Lemma~\ref{Derivative of the ratio is small}]\label{app:proof-lem-gradient-ratio}
Let \(  \omega \sim N(0,I_n)  \) and \(  x \in \mathbb{R}^n  \), and define 
\[  
X = \omega^\top \Delta.
  \]
First, using Stein's Lemma one can show \(  \mathbb{E}[Z] = 0  \). We start with the following known facts:
\[  
\mathbb{E}[\cos(X)] = e^{-\|\Delta\|^2/2}, \quad \text{and} \quad \mathbb{E}[1-\cos(X)] = 1 - e^{-\|\Delta\|^2/2}.
  \]
By Stein's lemma, we have
\[  
\mathbb{E}[\omega_1 \sin(X)]
=
\mathbb{E}\!\Biggl[\frac{\partial}{\partial \omega_1}\sin(\omega^\top \Delta)\Biggr]
=
\mathbb{E}\bigl[\Delta_1 \cos(X)\bigr]
=
\Delta_1\,e^{-\|\Delta\|^2/2}.
  \]
Multiplying both sides of the last equality by \(  1 - e^{-\|\Delta\|^2/2}  \) yields
\[  
(1 - e^{-\|\Delta\|^2/2}) \, \mathbb{E}[\omega_1 \sin(X)]
=
(1 - e^{-\|\Delta\|^2/2}) \, (\Delta_1 \, e^{-\|\Delta\|^2/2}).
  \]
Since $(1 - e^{-\|\Delta\|^2/2}) = \mathbb{E}[1-\cos(X)]$, we have $(1 - e^{-\|\Delta\|^2/2}) \mathbb{E}[\omega_1 \sin(X)] = \Delta_1 \, e^{-\|\Delta\|^2/2} \, \mathbb{E}[1-\cos(X)]$. Rearranging terms shows that $\bigl(1 - e^{-\|\Delta\|^2/2}\bigr)\,\mathbb{E}[\omega_1 \sin(X)] - \Delta_1\,e^{-\|\Delta\|^2/2}\,\mathbb{E}[1-\cos(X)] = 0$.
Now, we need Sub Gaussian norm of \(  Z  \). As $\abs{\omega_1\,\sin(I)} \;\le\; \abs{\omega_1}$ therefore $
\norm{\omega_1\,\sin(I)}_{\psi_2}
~\le~
\norm{\omega_1}_{\psi_2}
~=~
2.
$
Also,$\abs{1 - \cos(I)} \;\le\; 2.
\Rightarrow
\norm{1 - \cos(I)}_{\psi_2}
~\le~
2.$
Then using property of norm,
\[  
\norm{Z/t\Bigl(1 - e^{-\frac12\|\Delta\|^2}\Bigr)^2}_{\psi_2}
 \leq \frac{2\abs{1 - e^{-\|\Delta\|^2/2}} + \abs{\Delta_1} e^{-\|\Delta\|^2/2}}{t(1 - e^{-\|\Delta\|^2/2})^2} \leq \frac{4}{t(1 - e^{-\|\Delta\|^2/2})^2}.
  \]
Applying Generalized Hoeffding’s inequality~\cite{vershynin2009high}),
\[  
\mathbb{P} \left( \left| \frac{1}{t} \sum_{i=1}^{t} Z_i/\Bigl(1 - e^{-\frac12\|\Delta\|^2}\Bigr)^2 \right| \geq \varepsilon \right)
\leq \]\[2 \exp \left( -\frac{c_H \varepsilon^2}{t \cdot \frac{16}{t^2(1 - e^{-\|\Delta\|^2/2})^4}} \right)\leq 2 \exp \left( -\frac{(1 - e^{-\|\Delta\|^2/2})^4 \cdot c_H \varepsilon^2 t}{16} \right).
  \]

\end{proof}

\subsection{Lipschitz Control of the Ratio Gradient}
\label{sebsec: Bounding Lipschitz constant of the derivative of the ratio}

Before analyzing the Lipschitz properties of the gradient, we require a concentration bound for sums of Gaussian norms. 

To extend pointwise gradient bounds to uniform control over regions, we establish that the gradient of the kernel distance ratio is Lipschitz continuous. This allows us to control how quickly the gradient can vary between nearby points. The proof of Lemma~\ref{Lipschitz Constant of derivative of ratio} follows the lemma statement below.

\begin{lemma}[Lipschitz Constant of Derivative of Ratio]
\label{Lipschitz Constant of derivative of ratio}
Let \( \Delta, \Delta' \in \mathbb{R}^N \). For each coordinate \( i \in \{1, \dots, N\} \), the partial derivatives satisfy
\[
\left\| \frac{\partial}{\partial \Delta_i} R(\Delta) - \frac{\partial}{\partial {\Delta'}_i} R({\Delta'}) \right\|
\leq C_L \|\Delta - \Delta'\|,
\]
where \( C_L \leq 8N + 53\sqrt{N} + 95 \) with probability at least \( 1 - \exp(-tN/C) \) for some absolute constant \( C > 0 \).
\end{lemma}

\begin{proof}[Proof of Lemma~\ref{Lipschitz Constant of derivative of ratio}]\label{app:proof-lip-derivative-ratio}
We analyze the difference in partial derivatives through several key observations:

\begin{romanenumerate}
\item \textbf{Exponential term difference:} 
By the mean value theorem applied to \( f(s) = e^{-s/2} \), there exists \( \xi \) between \( \|\Delta\|^2 \) and \( \|\Delta'\|^2 \) such that
\[
\left| e^{-\|\Delta\|^2/2} - e^{-\|\Delta'\|^2/2} \right|
= \frac{1}{2} e^{-\xi/2} \left| \|\Delta\|^2 - \|\Delta'\|^2 \right|.
\]
Since \( \left| \|\Delta\|^2 - \|\Delta'\|^2 \right| \leq \|\Delta + \Delta'\|\cdot\|\Delta - \Delta'\| \), we have
\[
\left| e^{-\|\Delta\|^2/2} - e^{-\|\Delta'\|^2/2} \right|
\leq \frac{1}{2} e^{-\min\{\|\Delta\|^2, \|\Delta'\|^2\}/2} \|\Delta + \Delta'\|\cdot\|\Delta - \Delta'\|.
\]

\item \textbf{Trigonometric term differences:}
For the sine terms, we have
\[
\left| \sin(\langle \omega, \Delta \rangle) - \sin(\langle \omega, \Delta' \rangle) \right|
\leq \|\omega\|\cdot\|\Delta - \Delta'\|
\]
and similarly for cosine terms:
\[
\left| \cos(\langle \omega, \Delta \rangle) - \cos(\langle \omega, \Delta' \rangle) \right|
\leq \|\omega\|\cdot\|\Delta - \Delta'\|.
\]

\item \textbf{Bound on individual terms:}
Each term in the partial derivative satisfies
\[
\left| \omega_1 \sin(\langle \omega, \Delta \rangle)(1 - e^{-\|\Delta\|^2/2}) - (1 - \cos \langle \omega, \Delta \rangle)\Delta_1 e^{-\|\Delta\|^2/2} \right|
\leq |\omega_1| + 2.
\]
\end{romanenumerate}

Now consider the complete difference expression:
\begin{align*}
E(\Delta,\Delta') 
&= \omega_1\left[ \sin(\langle \omega, \Delta \rangle)(1 - e^{-\|\Delta\|^2/2}) - \sin(\langle \omega, \Delta' \rangle)(1 - e^{-\|\Delta'\|^2/2}) \right] \\
&\quad - \left[ \Delta_1 e^{-\|\Delta\|^2/2} - \Delta'_1 e^{-\|\Delta'\|^2/2} \right] \\
&\quad - \left[ \cos(\langle \omega, \Delta \rangle)\Delta_1 e^{-\|\Delta\|^2/2} - \cos(\langle \omega, \Delta' \rangle)\Delta'_1 e^{-\|\Delta'\|^2/2} \right].
\end{align*}

\textbf{Step 1 (Sine difference):} The first term satisfies
\[
|\omega_1 \Gamma_1| \leq |\omega_1|\left( \frac{1}{2}e^{-\min\{\|\Delta\|^2,\|\Delta'\|^2\}/2}\|\Delta + \Delta'\| + \|\omega\| \right)\|\Delta - \Delta'\|.
\]

\textbf{Step 2 (Linear-exponential difference):} The second term satisfies
\[
|\Gamma_2| \leq \left( 1 + \frac{1}{2}|\Delta'_1|e^{-\min\{\|\Delta\|^2,\|\Delta'\|^2\}/2}\|\Delta + \Delta'\| \right)\|\Delta - \Delta'\|.
\]

\textbf{Step 3 (Cosine-linear-exponential difference):} The third term satisfies
\[
|\Gamma_3| \leq \left( \|\omega\||\Delta'_1|e^{-\|\Delta'\|^2/2} + 1 + \frac{1}{2}|\Delta_1|e^{-\min\{\|\Delta\|^2,\|\Delta'\|^2\}/2}\|\Delta + \Delta'\| \right)\|\Delta - \Delta'\|.
\]

Combining these bounds yields:
\[
|E(\Delta,\Delta')| \leq \left( |\omega_1|(1 + \|\omega\|) + (\|\omega\| + 5) \right)\|\Delta - \Delta'\|.
\]

For the complete partial derivative difference, we obtain:
\[
\left| \frac{\partial}{\partial \Delta_i} R(\Delta) - \frac{\partial}{\partial \Delta'_i} R(\Delta') \right| \leq \frac{1}{t} \sum_{j=1}^t \left( 7\|\omega^j\|^2 + 53\|\omega^j\| + 95 \right) \|\Delta - \Delta'\|.
\]

By Lemma~\ref{Probabilistic Cauchy-Schwarz in Subspaces}
\[
C_L \leq \frac{1}{t} \sum_{j=1}^t \left( 7\|\omega^j\|^2 + 53\|\omega^j\| + 95 \right) \leq  O(N).
\]
with probability at least \( 1 - \exp(-tN/C) \) for some absolute constant \( C > 0 \).
\end{proof}

\begin{remark}[Uniform Control over All Local Difference Patches]
In the local‐chart framework of Section~\ref{subsec:local-charts}, we cover \(\mathcal{M}\) by an \(r\)\nobreakdash‐net \(\Gamma_{r}\) of cardinality \(|\Gamma_{r}|\) where \(r = O(\varepsilon / N)\).  Each ordered pair \((p,q)\in\Gamma_{r}\times\Gamma_{r}\) defines a difference patch \(\Delta_{pq}\), and Lemma~\ref{Lipschitz Constant of derivative of ratio} guarantees that on the tangent‐space approximation to each patch the Lipschitz constant \(C_L\) satisfies
\[
\Pr\Bigl[\forall\,\Delta,\Delta'\in\Delta_{pq}:\;|\partial_iR(\Delta)-\partial_iR(\Delta')|\le C_L\|\Delta-\Delta'\|\Bigr]
\;\ge\;1-2\exp(-t\,d/C).
\]
Since there are \(|\Gamma_{r}|^2\) such patches, a union bound over all \((p,q)\) shows that with probability at least
\[
1 \;-\; 2\,|\Gamma_{r}|^2\,\exp\!\bigl(-t\,d/C\bigr)
\]
the same Lipschitz control holds simultaneously on \emph{every} local difference patch \(\Delta_{pq}\).  In particular, as long as
\[
t\,d \;\gtrsim\; \log\bigl(|\Gamma_{r}|^2\bigr)
\;=\;2\log|\Gamma_{r}|,
\]
all tangent‐based approximations—and hence the propagated ratio bounds of Section~\ref{subsec:lifting_ratio}—hold uniformly over the entire manifold difference set \(\mathcal{M}-\mathcal{M}\).
\end{remark}

\subsection{Uniform Control of Ratio Gradient Within Euclidean Balls}
\label{subsec: Bounding the ratio in the ball}

Equipped with pointwise and differential control, we now derive a uniform guarantee over Euclidean balls using the gradient's Lipschitz continuity.

\begin{lemma}[Local Uniform Control in Euclidean Balls]
\label{Local Uniform Control in Euclidean Balls}
Let \( p \in \mathbb{R}^n \) with \( \|p\| \ge 1 \), and let \( B^{v_1, \ldots, v_d}_p(l) \) denote a ball of radius \( l \) centered at \( p \) in the affine flat \(\vec{p} + \langle\{v_1, \ldots, v_d\}\rangle\). Define the kernel-distance ratio
\[
R(\Delta) := \frac{D_{\hat{K}}(\Delta)^2}{D_K(\Delta)^2}.
\]
Then for all \( \Delta \in B^{v_1, \ldots, v_d}_p(l) \), with probability at least
\[
1 - 2d \exp\left( - \frac{c_H \varepsilon^2 t}{668} \right) - 2\exp\left(-\frac{t \varepsilon_0^2}{13} \right),
\]
we have the uniform bound:
\[
\left| R(\Delta) - 1 \right| \le \varepsilon_0 + l\sqrt{d}(\varepsilon + lC_L),
\]
where \( C_L \) is the Lipschitz constant for \( \nabla R \) from Lemma~\ref{Lipschitz Constant of derivative of ratio}, \( c_H \) is the absolute constant from Lemma~\ref{Derivative of the ratio is small}, and \( \varepsilon, \varepsilon_0 > 0 \) are error parameters.
\end{lemma}

\begin{proof}[Proof of Lemma~\ref{Local Uniform Control in Euclidean Balls}]\label{app:proof-lem-local-uniform-balls}

By Lemma~\ref{Lipschitz Constant of derivative of ratio} and Lemma~\ref{Derivative of the ratio is small}, for all $\Delta \in B_l(p) $
\[ 
\|\nabla_{i} R(\Delta)\|
\;\le\;
\|\nabla_{i} R(p)\| + lC_L
\;\le\;
\varepsilon + lC_L
 \]
with probability at least
\[ 
1 \;-\; 2 \exp \left( -\frac{(1 - e^{-\|p\|^2/2})^4 \cdot c_H \varepsilon^2 t}{16} \right),
 \]
where \( c_H \) is an absolute positive constant. Using union bound over all the directional derivatives with the directions spanning $\langle\{v_1, \ldots, v_d\}\rangle$.
\[ 
\|\nabla R(\Delta)\|
\;\le\;
\sqrt{d}(\varepsilon + lC_L)
\quad
\text{with probability at least }
1 - 2d \exp \left( -\frac{c_H \varepsilon^2 t}{668} \right).
 \]
Then using Lemma~\ref{lem:ratio bound at a point} along with the above bound

\[ 
\|R(\Delta)\|
\;\le\;
\varepsilon_0 +  l\sqrt{d}(\varepsilon + lC_L)
\quad
\text{with probability at least }\]\[
1 - 2d \exp \left( -\frac{c_H \varepsilon^2 t}{668} \right)-2\,\exp\Bigl(-\,\tfrac{t\varepsilon_0^2}{13}\Bigr) \]
\end{proof}

\section{Some Properties of the Function \texorpdfstring{$R$}{R}}
\label{sec:properties_R}

In this section we take a closer look at the analytic and geometric properties of the random kernel-distance ratio~$R$. In particular, we focus on the regime $\|\Delta\|\ll1$. The ratio
\[
   R(x)
   \;:=\;
   \frac{1}{t}
   \sum_{i=1}^{t}
      \frac{1-\cos\langle w_{i},x\rangle}
           {1-\exp\!\bigl(-\tfrac12\|x\|^{2}\bigr)},
   \qquad
   x\in\R^{N},
\]
plays a central rôle in all subsequent error bounds.  
Although $R$ is \emph{not} Lipschitz at the origin, it enjoys a
\emph{relative‐Lipschitz} behavior of the form
\(
   |R(z)-R(y)|
   \lesssim
   \|z-y\|/\max\{\|y\|,\|z\|\},
\)
which is sufficient for our manifold patching arguments
(Section~\ref{subsec:local-charts}).  

The purpose of this section is to establish that estimate in two
steps:
\begin{enumerate}
   \item
      Lemma~\ref{lem:LipboundR1} proves the estimate for the
      \emph{single‐frequency} building block
      \(
         f(x)=
         (1-\cos\langle w,x\rangle)/(1-e^{-\|x\|^{2}/2})
      \)
      inside the unit ball.
   \item
      Lemma~\ref{lem:LipboundR2} upgrades this to the full random
      feature map $R$ via a Gaussian concentration argument,
      showing that the desired bound holds simultaneously for all
      $y,z\in B_{0}(1)$ with high probability.
\end{enumerate}
These results will later allow us to transfer uniform
approximation guarantees from tangent spaces to the manifold itself,
even in regions where $\|x\|\ll1$.
The proofs of Lemmas~\ref{lem:LipboundR1} and~\ref{lem:LipboundR2} follow each lemma statement above.

\begin{lemma}
\label{lem:LipboundR1}
Let
\[
f(x) = \frac{1 - \cos \langle w, x \rangle}{1 - e^{-\|x\|^2/2}}.
\]
Then $|f(z)-f(y)| \leq 6 \|w\|^2 \frac{\|y - z\|}{\max(\|z\|,\|y\|)}$ inside the unit ball.
\end{lemma}
\begin{proof}[Proof of Lemma~\ref{lem:LipboundR1}]\label{app:proof-lem-LipboundR1}
Write $g(x) = 1 - \cos \langle w, x \rangle$ and $h(x) = 1 - e^{-\|x\|^2/2}$. Then
\begin{multline*}
f(z)-f(y)
=\frac{\cos \langle w, y \rangle - \cos \langle w, z \rangle}{h(y)}
-g(z)\Bigl(\frac{1}{h(y)}-\frac{1}{h(z)}\Bigr)\\
=\frac{2\sin\!\bigl\langle \tfrac{w}{2}, y+z \bigr\rangle\,\sin\!\bigl\langle \tfrac{w}{2}, y-z \bigr\rangle}{h(y)}
-g(z)\,\frac{e^{-\|y\|^2/2}-e^{-\|z\|^2/2}}{h(y)\,h(z)}.
\end{multline*}

Assume $\|z\| \leq \|y\|$. For the first grouped term,
\begin{align*}
\frac{2\bigl|\sin\!\langle \tfrac{w}{2}, y+z \rangle\bigr|\,\bigl|\sin\!\langle \tfrac{w}{2}, y-z \rangle\bigr|}{|h(y)|}
&\leq \frac{2\bigl|\bigl\langle \tfrac{w}{2}, y+z \bigr\rangle\bigr|\,\bigl|\bigl\langle \tfrac{w}{2}, y-z \bigr\rangle\bigr|}{|h(y)|}
\leq \frac{\|w\|^2 \|y+z\|\,\|y-z\|}{2|h(y)|}\\
&\leq \frac{\|w\|^2 \|y\|^2 \|y-z\|}{|h(y)|\,\|y\|}
\leq 3 \|w\|^2 \frac{\|y - z\|}{\|y\|}.
\end{align*}

For the second grouped term,
\begin{align*}
&\left|g(z)\,\frac{e^{-\|y\|^2/2}-e^{-\|z\|^2/2}}{h(y)\,h(z)}\right|
\leq 2\sin^2\!\langle \tfrac{w}{2}, z \rangle\cdot
\frac{\tfrac12 \|y-z\|\,\|y+z\|}{h(y)\,h(z)}\\
&\leq \frac{\|w\|^2}{2}\cdot
\frac{\tfrac12 \|z\|^2 \|y\|^2 \|y-z\|}{h(y)\,h(z)\,\|y\|}
\leq 3 \|w\|^2 \frac{\|y - z\|}{\|y\|}.
\end{align*}
Summing the two bounds completes the proof.
\end{proof}

\begin{remark}[Relative‐Lipschitz Bound Outside the Unit Ball]
\label{rem:LipboundR1_outside}
For points \(y,z\in\R^N\) with \(\|y\|,\|z\|\ge1\), one obtains a simpler (absolute) Lipschitz bound for
\[
f(x)=\frac{1-\cos\langle w,x\rangle}{1-e^{-\|x\|^2/2}}
\]
by observing that both numerator and denominator are \(O(1)\) and smooth away from the origin.  In fact, assuming \(\norm{y-z} \leq 1\), one shows
\[
|f(z)-f(y)|
\;\le\;
3\bigl|\langle w,\,y-z\rangle\bigr|
\;+\;5\,\|y-z\|
\;\le\;
\bigl(3\|w\|+5\bigr)\,\|y-z\|.
\]
\end{remark}

\begin{lemma}[Relative Lipschitzness of the random feature map]\label{lem:LipboundR2}
Let $w_{1},\dots ,w_{t}$
$\stackrel{\mathrm{i.i.d.}}{\sim}\mathcal N(0,I_{N})$
and define
\[
   R(x)\;:=\;\frac{1}{t}
   \sum_{i=1}^{t}
      \frac{1-\cos\langle w_{i},x\rangle}
           {1-\exp\!\bigl(-\tfrac12\|x\|^{2}\bigr)},
   \qquad x\in\R^{N}.
\]
There is an absolute constant $c_{4}>0$ such that, with probability at least
\(
   1-2\exp\!\bigl(-N t/c_{4}^{4}\bigr),
\)
the following \emph{relative-Lipschitz} estimate holds for all
$y,z\in B_{0}(1):=\{x\in\R^{N}:\|x\|\le1\}$:
\[
   |R(z)-R(y)|
   \;\le\;
   12N\,
   \frac{\|y-z\|}{\max\{\|y\|,\|z\|\}}.
\]
\end{lemma}

\begin{proof}[Proof of Lemma~\ref{lem:LipboundR2}]\label{app:proof-lem-LipboundR2}
Invoke Lemma~\ref{lem:LipboundR1} with
$f_{i}(x)=\dfrac{1-\cos\langle w_{i},x\rangle}
              {1-\exp\!\bigl(-\tfrac12\|x\|^{2}\bigr)}$.
For each $i$ and $\|x\|\le\|y\|\le1$ it yields
\[
   |f_{i}(z)-f_{i}(y)|
   \;\le\;
   6\,\|w_{i}\|^{2}\,
   \frac{\|y-z\|}{\|y\|}.
\]
Summing over $i=1,\dots ,t$ gives
\[
   \bigl|R(z)-R(y)\bigr|
   \;\le\;
   \frac{6\|y-z\|}{\|y\|}
   \Bigl(\frac1t\sum_{i=1}^{t}\|w_{i}\|^{2}\Bigr).
\tag{$\ast$}
\]

So, 

\[
\bigl|R(z)-R(y)\bigr|
   \;\le\;
   \frac1t\sum_{i=1}^{t}\|w_{i}\|^{2}\;
   \Bigl(6\,\frac{\|y-z\|}{\|y\|}\Bigr),
   \qquad
   w_{i}\stackrel{\mathrm{i.i.d.}}{\sim}\mathcal N(0,I_{N}).
\]
Then using Lemma~\ref{Concentration for Gaussian Norm Sums} we get,
\[
   \bigl|R(z)-R(y)\bigr|
   \;\le\;
   12\,N\,
   \frac{\|y-z\|}{\|y\|},
\]
so the bound holds with probability at least
$1-2\exp\!\bigl(-N t / c_{4}^{4}\bigr)$.
\end{proof}

\begin{remark}[Lipschitz control of $R$ outside the unit ball]
\label{rem:gradoutside}
\label{rem:LipboundR2_outside}
Let \(w_{1},\dots,w_{t}\overset{\mathrm{i.i.d.}}{\sim}\mathcal N(0,I_{N})\) and define
\[
R(x)\;=\;\frac{1}{t}\sum_{i=1}^{t}
      \frac{1-\cos\langle w_{i},x\rangle}
           {1-\exp\!\bigl(-\tfrac12\|x\|^{2}\bigr)},
\qquad x\in\R^{N}.
\]
Then there is an absolute constant \(c_{4}>0\) such that, with probability at least
\[
1-2\exp\!\bigl(-N t / c_{4}^{4}\bigr),
\]
the following (absolute) Lipschitz estimate holds for all
\(y,z\in\R^{N}\) with \(\|y\|,\|z\|\ge1\):
\[
\bigl|R(z)-R(y)\bigr|
\;\le\;
12N\;\|y-z\|.
\]
\end{remark}

\section{Controlling Differences on a Manifold via Local Charts}
\label{subsec:local-charts}

In this section we build upon the previous section to lift the ratio from tangent planes to local neighborhoods of the manifold, using the notion of reach.

\paragraph*{Local Difference Patches.}
Fix a scale parameter $0 < r < \rch(\mathcal{M})/2$ and let $m:=|\Gamma_r|$, $\Gamma_r=\{z_1,\dots,z_m\}\subset\mathcal{M}$ be an $r$-net as provided by Lemma~\ref{prop:manifold-net}. For each ordered pair $(p,q)=(z_i,z_j)$ with $i,j\in\{1,\dots,m\}$, we consider the local \emph{difference patch}:
\[
\Delta_{pq} := \{ p' - q' \mid p' \in B_{\mathcal{M}}(p,r), q' \in B_{\mathcal{M}}(q,r) \},
\]
where $B_{\mathcal{M}}(p,r)$ denotes the intrinsic ball on $\mathcal{M}$.

\paragraph*{Approximation by Tangent Differences.}
Lemma~\ref{lem:diff-quadratic} establishes that each difference patch $\Delta_{pq}$ is well-approximated by the linear space of tangent differences: $\Delta_{pq} \subset (T_p\mathcal{M} - T_q\mathcal{M}) + B_{r^2/\rch(\mathcal{M})}(0)$ and $d(\Delta_{pq}, p-q) \leq 2r + r^2/\rch(\mathcal{M})$. This shows that locally, the set of differences lies close to a $(2d)$-dimensional flat space whose geometry is explicitly determined by the tangent spaces at $p$ and $q$.

\paragraph*{Propagating Ratio Control.}
Since the ratio $R(\Delta)$ is relatively Lipschitz (Lemma~\ref{lem:LipboundR2}), its variation between $\Delta_{pq}$ and $T_p\mathcal{M} - T_q\mathcal{M}$ is at most $O(r/\rch(\mathcal{M}))$. By considering all $|\Gamma_r|^2$ ordered pairs of the net, we obtain global multiplicative control of $R$ over all difference vectors in $\mathcal{M}-\mathcal{M}$, without requiring a manifold structure on $\mathcal{M}-\mathcal{M}$ itself.

\subsection{Lifting the Ratio from Tangent Planes to the Manifold}
\label{subsec:lifting_ratio}

The ratio function \( R(\Delta) = \frac{D_{\hat{K}}(\Delta)^2}{D_K(\Delta)^2} \) exhibits singular behavior near the origin, failing to be Lipschitz continuous in any neighborhood of zero. To handle this, we introduce a relaxed notion of continuity that better captures the behavior of \( R \), while leveraging the manifold structure to control approximation errors.

\begin{definition}[Relative Lipschitz Continuity]
A function \( f \colon \mathbb{R}^k \to \mathbb{R} \) is called \emph{relatively Lipschitz} with constant \( C \) if for all \( x, y \in \mathbb{R}^k \setminus \{0\} \),
\[
|f(x) - f(y)| \leq C \cdot \frac{\|y - x\|}{\max(\|x\|, \|y\|)}.
\]
\end{definition}

The following lemma quantifies how relative Lipschitz functions behave when lifted from tangent spaces to the manifold:

\begin{lemma}[Function Lifting via Projection]
\label{lem:projection_bound}
Let \( \mathcal{M} \subset \mathbb{R}^N \) be a \( d \)-dimensional manifold with reach \( \rch(\mathcal{M}) \), and let \( f \colon \mathbb{R}^N \to \mathbb{R} \) be relatively Lipschitz with constant \( C \). For any \( p, q \in \mathcal{M} \) with \( \|p - q\| \leq r < \rch(\mathcal{M}) \):
\begin{enumerate}
    \item The projection lies in the tangent tube:
    \[
    \pi_{T_p\mathcal{M}}(q) \in \mathcal{T}^{\rch(\mathcal{M})}_p\mathcal{M}.
    \]
    \item The function variation is quadratically small:
    \[
    \left| f(q) - f(\pi_{T_p\mathcal{M}}(q)) \right| \leq \frac{C}{2\rch(\mathcal{M})} \cdot \frac{\|p - q\|^2}{\|q\|}.
    \]
\end{enumerate}
\end{lemma}

\begin{proof}[Proof of Lemma~\ref{lem:projection_bound}]\label{app:proof-lem-projection-bound}
The first claim follows immediately from the definition of the tangent tube. For the second claim, let \( q' = \pi_{T_p\mathcal{M}}(q) \) and observe that by relative Lipschitz continuity and Lemma~\ref{lem:distance_tangent_space},
\[
|f(q) - f(q')| \leq C \cdot \frac{\|q - q'\|}{\max(\|q\|, \|q'\|)} \leq C \cdot \frac{\|q - q'\|}{\|q\|} \leq \frac{C}{2\rch(\mathcal{M})} \cdot \frac{\|p - q\|^2}{\|q\|}. \qedhere
\]
\end{proof}

A special case occurs when projecting to the tangent space at the origin, where the relative Lipschitz condition is most stringent:

\begin{corollary}[Origin-Centered Projection]
\label{cor:projection_bound}
Under the assumptions of Lemma~\ref{lem:projection_bound}, if \( p = 0 \) and \( 0 < \|q\| \leq r < \rch(\mathcal{M}) \), then:
\begin{enumerate}
    \item \( \pi_{T_0\mathcal{M}}(q) \in \mathcal{T}^{\rch(\mathcal{M})}_0\mathcal{M} \).
    \item \( \left| f(q) - f(\pi_{T_0\mathcal{M}}(q)) \right| \leq \frac{Cr}{2\rch(\mathcal{M})} \).
\end{enumerate}
\end{corollary}

\begin{figure}[H]
    \centering
\includegraphics[width=\textwidth]{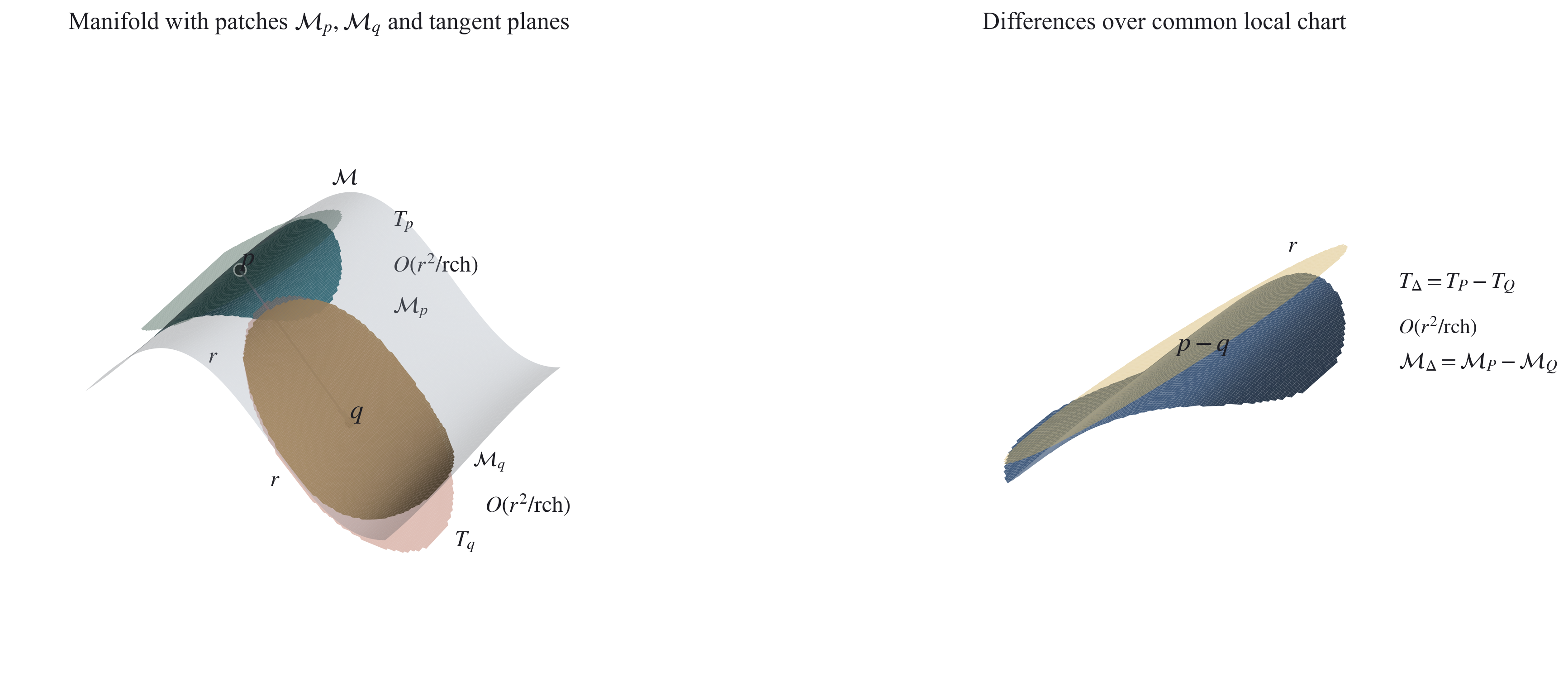}
    \caption{
        Left: a $d$–dimensional manifold $\mathcal M$ with two intrinsic balls
        $B_{\mathcal M}(p,r)$ and $B_{\mathcal M}(q,r)$ (shown as patches
        $\mathcal M_{p}$ and $\mathcal M_{q}$) together with their tangent
        planes $T_p\mathcal{M}$ and $T_q\mathcal{M}$ at the centers $p$ and $q$.  
        Right: the corresponding difference patches in a common local chart,
        illustrating the manifold difference $\mathcal M_{p}-\mathcal M_{q}$
        and the tangent difference $T_p\mathcal{M}-T_q\mathcal{M}$, with the deviation between
        them of order $O(r^{2}/\rch(\mathcal{M}))$.
    }
    \label{fig:manifold-patches-differences}
\end{figure}
\subsection{Local Tangent Approximation of Manifold Differences}
\label{subsec:inside-unit-ball}

Let $\mathcal{M}' := \mathcal{M} \cap B_1(0)$ denote the part of the manifold inside the unit ball. Our goal is to obtain a uniform bound for the kernel distance ratio over all pairwise differences $\Delta = p'-q'$ for $p', q' \in \mathcal{M}'$.

For manifolds with positive reach, the distance between points and their tangent space projections admits a quadratic bound. This property will allow us to control the variation of \( R \) between the manifold and its tangent spaces. The following lemma is used later to contain each local difference patch in a small neighborhood of a tangent-difference space, so that the ratio $R(\Delta)$ can be controlled there and then propagated globally via the net.

Manifold differences are vectors obtained by subtracting two points on the manifold; tangent differences are vectors obtained by subtracting a tangent vector at one point from a tangent vector at another. Because the manifold is curved, a manifold difference need not lie in the space of tangent differences, but when both points stay close to fixed base points the next lemma shows that the manifold difference remains close to that flat tangent-difference space.

\begin{lemma}[Deviation of Point Differences from Tangent Differences]
\label{lem:diff-quadratic}
Let $\mathcal{M} \subset \mathbb{R}^N$ be a $C^2$ submanifold with reach $\rch(\mathcal{M}) > 0$. Fix $p,q \in \mathcal{M}$ and consider points:
\[
p', q' \in \mathcal{M} \quad \text{with} \quad \|p'-p\| \leq \varepsilon, \|q'-q\| \leq \varepsilon, \quad 0 < \varepsilon < \rch(\mathcal{M}).
\]
Define the linear space of tangent differences:
\[
L := T_p\mathcal{M} - T_q\mathcal{M} = \{ u - v \mid u \in T_p\mathcal{M}, v \in T_q\mathcal{M} \} \subset \mathbb{R}^N.
\]
Then:
\begin{enumerate}
    \item\label{it:distance} The distance from $p'-q'$ to $L$ satisfies:
    \[
    d(p'-q', L) \leq \frac{\varepsilon^2}{\rch(\mathcal{M})}.
    \]
    
    \item\label{it:foot} For the orthogonal projection $\eta := \pi_L(p'-q')$, we have:
    \[
    \|\eta - (p-q)\| \leq 2\varepsilon + \frac{\varepsilon^2}{\rch(\mathcal{M})}.
    \]
\end{enumerate}
\end{lemma}

\begin{proof}[Proof of Lemma~\ref{lem:diff-quadratic}]\label{app:proof-lem-diff-quadratic}
\textbf{Part \eqref{it:distance}: Distance to tangent differences.}
By reach estimate, for any $x \in \mathcal{M}$ and $x^* \in \mathcal{M}$ with $\|x^*-x\| \leq \varepsilon < \rch(\mathcal{M})$:
    \[
    d(x^*-x, T_x\mathcal{M}) \leq \frac{\varepsilon^2}{2\rch(\mathcal{M})}.
    \]
Applying this at $x = p$ and $x = q$ yields $d(p'-p, T_p\mathcal{M}) \leq \frac{\varepsilon^2}{2\rch(\mathcal{M})} \quad \text{and} \quad d(q'-q, T_q\mathcal{M}) \leq \frac{\varepsilon^2}{2\rch(\mathcal{M})}.$ Choose $u \in T_p\mathcal{M}$ and $v \in T_q\mathcal{M}$ such that:
    \[
    \|p'-p-u\| \leq \frac{\varepsilon^2}{2\rch(\mathcal{M})} \quad \text{and} \quad \|q'-q-v\| \leq \frac{\varepsilon^2}{2\rch(\mathcal{M})}.
    \]
Define $\zeta := (p-q) + (u-v) \in (p-q) + L$. Then:
    \[
    \|p'-q'-\zeta\| \leq \frac{\varepsilon^2}{2\rch(\mathcal{M})} + \frac{\varepsilon^2}{2\rch(\mathcal{M})} = \frac{\varepsilon^2}{\rch(\mathcal{M})}.
    \]
Since $\eta = \pi_L(p'-q')$ minimizes distance to $L$:
    \[
    d(p'-q', L) = \|p'-q'-\eta\| \leq \frac{\varepsilon^2}{\rch(\mathcal{M})}.
    \]

\textbf{Part \eqref{it:foot}: Projection proximity.}
By the triangle inequality:
    \[
    \|\eta - (p-q)\| \leq \underbrace{\|\eta - (p'-q')\|}_{\leq \varepsilon^2/\rch(\mathcal{M})} + \underbrace{\|(p'-q') - (p-q)\|}_{\leq 2\varepsilon}.
    \]
Combining these bounds gives:
    \[
    \|\eta - (p-q)\| \leq 2\varepsilon + \frac{\varepsilon^2}{\rch(\mathcal{M})}. \qedhere
    \]
\end{proof}
This lemma is used in the proof of the uniform ratio approximation over the manifold (Lemma~\ref{lem:uniform-ratio-approx-in}) to confine each local difference patch near a tangent-difference space before applying the relative-error control.

\begin{lemma}[Uniform Ratio Approximation]
\label{lem:uniform-ratio-approx-in}
For any $\varepsilon > 0$, with $r := \rch(\mathcal{M})\,\varepsilon/(12N) < \varepsilon < \rch(\mathcal{M})$, the approximate kernel distance satisfies
\[
\frac{D_{\hat{K}}(\Delta)^2}{D_K(\Delta)^2} \in [1-2\varepsilon,\; 1+2\varepsilon]
\]
for all $\Delta \in (\mathcal{M} - \mathcal{M})\cap B_1(0)$, with failure probability at most
\[
|\Gamma_r|^2 \cdot O\left(\frac{2d+1}{\varepsilon}\exp\left(-\frac{t\varepsilon^2}{2d+1}\right)\right).
\]
Here, $|\Gamma_r|$ is the covering number of an $r$-net for $\mathcal{M}$.

\end{lemma}

\begin{proof}[Proof of Lemma~\ref{lem:uniform-ratio-approx-in}]\label{app:proof-lem-uniform-ratio-in}
Given an $r$-net $\Gamma_r \subset \mathcal{M}'$ (with $r = O(\varepsilon)$), for each pair $(p,q) \in \Gamma_r \times \Gamma_r$, consider the local \emph{difference patch}
\[
\Delta_{pq} := \left\{ p'-q' : p' \in B_{\mathcal{M}}(p, r),\, q' \in B_{\mathcal{M}}(q, r) \right\}.
\]
By Lemma~\ref{lem:diff-quadratic}, every such patch is contained in an $\varepsilon$-neighborhood (with $\varepsilon \asymp r$) of the affine space $S_{pq}:=T_p\mathcal{M} - T_q\mathcal{M}$. This affine space has dimension at most $2d$, as $T_p\mathcal{M}$ and $T_q\mathcal{M}$ are both $d$-dimensional linear subspaces. Moreover, the set $\Delta_{pq}$ always contains the origin when $p = q$.

\paragraph*{Application of the relative error bound.}
For each patch, define the enclosing ball
\[
S_{pq} := \left( T_p\mathcal{M} - T_q\mathcal{M} \right) \cap B_{1 + O(r)}(0),
\]
which is a $(2d)$-dimensional slice (or $(2d+1)$-dimensional affine subspace if including the origin). The relative error theorem in~\cite{chen2017relative} gives, for all $\Delta \in S_{pq}$,
\[
\frac{D_{\hat{K}}(\Delta)^2}{D_K(\Delta)^2} \in [1-\varepsilon, 1+\varepsilon]
\]
with failure probability at most $O\left(\frac{2d+1}{\varepsilon}\exp\left(-\frac{t\varepsilon^2}{2d+1}\right)\right)$ per patch.

\paragraph*{Control over Difference Patches.}
Because every $\Delta \in \Delta_{pq}$ lies within $O(r)$ of $S_{pq}$, by relative Lipschitz continuity of the ratio, the bound holds within an $O(r)$ margin. For $r = O(\varepsilon)$, this results in the final bracket $[1-2\varepsilon, 1+2\varepsilon]$ for all $\Delta \in \Delta_{pq}$. Indeed, assume $r := \rch(\mathcal{M})\,\varepsilon/(12N) < \varepsilon < \rch(\mathcal{M})$. Applying the relative error result (see, e.g.,~\cite{chen2017relative}) on the smallest subspace containing $S_{pq}$ and the origin, setting the radius as $1$,
\[
\frac{D_{\hat{K}}(\Delta)^2}{D_K(\Delta)^2} \in [1-\varepsilon,\,1+\varepsilon],\quad \forall \Delta \in S_{pq}
\]
with probability at least 
\[
1 - O\Bigl(\frac{2d+1}{\varepsilon}\exp\Bigl(-t\varepsilon^2/(2d+1)\Bigr)\Bigr).
\]
Hence
\[
\frac{D_{\hat{K}}(\Delta)^2}{D_K(\Delta)^2} \in [1-\varepsilon - 12Nr/(2\,\rch(\mathcal{M})),\,1+\varepsilon+12Nr/(2\,\rch(\mathcal{M}))] \subset [1-2\varepsilon, 1+2\varepsilon],\quad \forall \Delta \in \Delta_{pq}
\]
with probability at least 
\[
1 - O\bigl(\Bigl(\frac{2d+1}{\varepsilon}\Bigr)\exp\Bigl(-t\varepsilon^2/(2d+1)\Bigr)\bigr).
\]

\paragraph*{Union Bound over All Pairs.}
Since the control applies to each difference patch associated to a pair $(p, q)$, and the number of such pairs is $|\Gamma_r|^2$, we take a union bound over all pairs. 

Thus, the \emph{total failure probability} is bounded by
\[
|\Gamma_r|^2 \cdot O\left(\frac{2d+1}{\varepsilon}\exp\left(-\frac{t\varepsilon^2}{2d+1}\right)\right).
\]

\end{proof}

\begin{remark}
The $(2d)$-dimensional structure arises since the difference of two $d$-dimensional tangent spaces is at most $2d$-dimensional. The union bound is squared because every pair of net points defines a patch, and we require uniform control for all such patches.
\end{remark}

\subsection{Uniform Approximation When the Difference Patch Is Outside the Unit Ball}

\begin{lemma}[Uniform Ratio Approximation]
\label{lem:uniform-ratio-approx-out}
For any $\varepsilon > 0$, with $r := \rch(\mathcal{M})\,\varepsilon/(12N) < \varepsilon < \rch(\mathcal{M})$, the approximate kernel distance satisfies
\[
\frac{D_{\hat{K}}(\Delta)^2}{D_K(\Delta)^2} \in [1-2\varepsilon,\; 1+2\varepsilon]
\]
for all $\Delta \in (\mathcal{M} - \mathcal{M})\setminus B_1(0)$, with failure probability at most
\[
|\Gamma_r'|^2 \cdot O\left(2d \exp\left( - \frac{c_H \varepsilon^2 t}{668} \right)\right),
\]
Here, $|\Gamma_r'|$ is the covering number of an $r$-net for $\mathcal{M}''$.

\end{lemma}

\begin{proof}[Proof of Lemma~\ref{lem:uniform-ratio-approx-out}]\label{app:proof-lem-uniform-ratio-out}

Let $r = \rch(\man)\,\varepsilon/(12N) < \varepsilon < \rch(\man)$. Define 
\[
\mathcal{M}'' := \mathcal{M} \setminus B_1(0)
\]
as the manifold outside the unit ball. Proceed as above:

\paragraph*{Construction of Covering Net}
Let $\Gamma_r'$ be an $r$-net covering $\mathcal{M}''$. For each pair $(p, q) \in \Gamma_r' \times \Gamma_r'$, consider the difference patch
\[
\Delta_{pq}'' := \{p' - q' : p' \in B_{\mathcal{M}}(p, r),\; q' \in B_{\mathcal{M}}(q, r)\}.
\]

\paragraph*{Local Flat Approximation}
Each patch $\Delta_{pq}''$ is well-approximated by $(T_p\mathcal{M}'' - T_q\mathcal{M}'')$, which is $2d$-dimensional. The local kernel distance ratio in each flat patch is controlled (see Lemma~\ref{Local Uniform Control in Euclidean Balls}) with failure probability $O\left(2d \exp(-c_H \varepsilon^2 t / 668)\right)$.

\paragraph*{Projection to Manifold}
Using Lemma~\ref{Local Uniform Control in Euclidean Balls}, setting the radius as $r$,
\[
\begin{split}
\frac{D_{\hat{K}}(\Delta)^2}{D_K(\Delta)^2}
\in \bigl[&1- \varepsilon - r\sqrt{d}(\varepsilon + rC_L),\\
&\quad 1+ \varepsilon + r\sqrt{d}(\varepsilon + rC_L)\bigr],
\end{split}
\]
for all $\Delta \in \tp{p}{r}''$, with probability at least 
\[
1 - 2d \exp\left( - \frac{c_H \varepsilon^2 t}{668} \right)
- 2\exp\left(-\frac{t \varepsilon^2}{13} \right).
\]
Using Remark~\ref{rem:gradoutside}, we get
\[
\begin{split}
\frac{D_{\hat{K}}(\Delta)^2}{D_K(\Delta)^2}
\in \Bigl[&1- \varepsilon - r\sqrt{d}(\varepsilon + rC_L) - \frac{12Nr}{2\,\rch(\mathcal{M})},\\
&\quad 1+ \varepsilon + r\sqrt{d}(\varepsilon + rC_L) + \frac{12Nr}{2\,\rch(\mathcal{M})}\Bigr],
\end{split}
\]
for all $\Delta \in \man_p^{''r}$.
with probability at least 
\[
1 - O\bigl(2d \exp\left( - \frac{c_H \varepsilon^2 t}{668} \right)\bigr),
\]

As $C_L=O(d)$ and $d \leq N$ we get
\[
\frac{D_{\hat{K}}(\Delta)^2}{D_K(\Delta)^2} \in [1-2\varepsilon, 1+2\varepsilon] \quad \forall \Delta \in \Delta_{pq}''
\]
with the same per-patch failure probability.

\paragraph*{Global Union Bound}
A union bound over all $|\Gamma_r'|^2$ pairs gives total failure probability at most
\[
|\Gamma_r'|^2 \cdot O\left(2d \exp\left( - \frac{c_H \varepsilon^2 t}{668} \right)\right),
\]
so the uniform bound holds for all $\Delta \in \mathcal{M}'' - \mathcal{M}''$.
\end{proof}

\subsection{Final Result: Uniform Multiplicative Approximation on Manifolds}
\label{subsec:final-result}

We combine the local estimates above into the global uniform guarantee already stated as Theorem~\ref{thm:intro_uniform} in the Introduction.

\begin{proof}[Proof of Theorem~\ref{thm:intro_uniform}]\label{app:proof-thm-intro-uniform}
Combining inside and outside the unit ball, we obtain our main uniform approximation. For any $\varepsilon > 0$, with $r = O(\varepsilon)$, the approximate kernel distance satisfies
\[
\frac{D_{\hat{K}}(\Delta)^2}{D_K(\Delta)^2} \in [1-2\varepsilon + O(\varepsilon^2),\; 1+2\varepsilon+O(\varepsilon^2)]
\]
for all $\Delta \in (\mathcal{M} - \mathcal{M})$, with failure probability at most
\[
|\Gamma_r|^2 \cdot O\left(\frac{2d}{\varepsilon}\exp\left(-\frac{t\varepsilon^2}{2d}\right)\right).
\]
Here, $|\Gamma_r|$ is the covering number of an $r$-net for $\mathcal{M}$.

The required $r$-net $\Gamma_r$ satisfies
\[
|\Gamma_r| = O\left(\frac{\vol(\mathcal{M}) \cdot N^d}{\vol(B^d_1(0)) \cdot \rch(\mathcal{M})^d \cdot \varepsilon^d}\right),
\]
where $B^d_1(0)$ denotes the unit ball in $\mathbb{R}^d$. This gives the desired result.

\end{proof}

\begin{remark}[Pairwise Patching Avoids Explicit Construction of $\mathcal{M}-\mathcal{M}$]
A key difficulty in analyzing the set of all differences $\mathcal{M}-\mathcal{M} = \{p-q : p, q \in \mathcal{M}\}$ is that $\mathcal{M}-\mathcal{M}$ generally lacks a manifold structure, and its geometric quantities such as reach are unknown or intractable. Our approach sidesteps this challenge entirely: instead of covering $\mathcal{M}-\mathcal{M}$ as a whole, we use a local chart strategy by constructing an $\varepsilon$-net $\Gamma_\varepsilon \subset \mathcal{M}$ and considering all pairwise local difference patches $\Delta_{pq}$ between neighborhoods around $p, q \in \Gamma_\varepsilon$.

Each $\Delta_{pq}$ is tightly controlled—by Lemma~\ref{lem:diff-quadratic}, it lies in a small neighborhood of the affine subspace $T_p\mathcal{M} - T_q\mathcal{M}$, which has known dimension at most $2d$. The kernel distance ratio $R(\Delta)$ is then shown to be well-approximated and uniformly controlled over all such patches, without needing any global structure or reach on $\mathcal{M}-\mathcal{M}$ itself. This pairwise tangent-based construction not only covers all differences in $\mathcal{M}-\mathcal{M}$, but also uses the manifold geometry and reach $\rch(\mathcal{M})$ of $\mathcal{M}$—quantities that are accessible—rather than that of the potentially stratifold like difference set $\mathcal{M}-\mathcal{M}$.
\end{remark}

\section{Approximation of Persistent Modules of a Manifold}
\label{sec:approx-persistent-modules}

The transition from geometric to topological approximation requires preserving not just pairwise distances but the full structure of the persistent homology modules. The foundational work of \cite{phillips1geometric} established the Gaussian Kernel Power Distance (GKPD) as a stable, distance-like function whose topology can be captured via weighted Rips complexes. Subsequent research demonstrated that dimensionality reduction can preserve persistent homology, with \cite{lotz2019persistent} and \cite{arya2021dimensionality} achieving this for Euclidean distances via random projections by leveraging intrinsic data complexity and the convex structure of weighted simplex radii.

The direct application to kernel methods was introduced by \cite{boissonnat2024euclidean}, who showed that Random Fourier Features (RFF) provide an $\varepsilon$-distortion map for pairwise kernel distances, leading to interleaved GKPD-based Čech filtrations via a key Simplex Distortion Lemma. Their later work introduced techniques for the relative approximation of kernel weights under a stable rank condition.

Building on this foundation, we now proceed from the geometric approximation of pairwise kernel distances on a manifold $\mathcal{M}\subset\mathbb{R}^N$ to the topological approximation of the induced persistent modules. This step is non-trivial, as the persistence of a filtration depends critically on the \emph{relative scaling of kernel weights}—quantities defined by the RKHS centroid $\mu_P$ that are not preserved by pairwise distance approximations alone. We first prove these weights are bounded away from zero, then establish their multiplicative stability under RFF projection. By unifying this result with our manifold-level distortion bounds, we prove that the GKPD-weighted Čech and Rips filtrations of the original data and its RFF image are $(1\pm\varepsilon_\star)$–interleaved, providing a complete topological guarantee for kernel-based persistent homology on manifolds.

\begin{theorem}[Weighted Čech/Rips interleaving on manifolds]
\label{thm:manifold_interleaving_final}
Let $\mathcal{M}\subset\mathbb{R}^N$ be a compact, $d$–dimensional, $\mathcal{C}^2$ submanifold with
positive reach $\rch(\mathcal{M})>0$, and let $P\subset\mathcal{M}$ be a finite sample of $n$ points.
Fix accuracy $\varepsilon\in(0,\rch(\mathcal{M})/2)$, confidence $\delta_0\in(0,1)$, and bandwidth $\sigma>0$.
Let $\phi:\mathbb{R}^N\to\mathbb{R}^t$ be a Random Fourier Feature (RFF) map for the Gaussian kernel
$K(x,y)=\exp\!\bigl(-\|x-y\|^2/2\bigr)$ with
\[
t
\;=\;
\Omega\!\Biggl(
\frac{d}{\varepsilon^2}\,
\log\!\Biggl(
\frac{\vol(\mathcal{M})^{2}\,N^{\,2d}}{\vol(B^d_1(0))^{2}\,\rch(\mathcal{M})^{\,2d}\,\varepsilon^{\,2d+1}\,\delta_0}
\Biggr)\Biggr).
\]
Define the kernel centroid $\mu_P=\tfrac{1}{|P|}\sum_{y\in P}\phi(y)$ and the constant
\[
c_P \;:=\; \frac{2}{\bigl(1-\|\mu_P\|\bigr)^2}\,.
\]
Then, with probability at least $1-\delta_0$, writing $\varepsilon_\star := \max\{\,2\varepsilon,\;\varepsilon c_P\,\}$,
the weighted Čech filtrations $\check{C}_\alpha(\widehat{P})$ (built with the Gaussian kernel power distance $D_K$)
and $\check{C}_\alpha\bigl(\phi(P)\bigr)$ (weights recomputed in the image) are $(1\pm\varepsilon_\star)$–interleaved.
Consequently, the corresponding weighted Rips filtrations $VR_\alpha(\widehat{P})$ and $VR_\alpha\bigl(\phi(P)\bigr)$
are also $(1\pm\varepsilon_\star)$–interleaved.
\end{theorem}

In the preceding sections, we established multiplicative distortion bounds for pairwise 
kernel distances under Random Fourier Feature (RFF) embeddings
of a submanifold $\mathcal{M}\subset\mathbb{R}^N$. 
We now proceed from geometric approximation to topological approximation, namely to the 
\emph{persistent modules} induced by these Čech and Rips filtrations.

This step is not immediate, since the persistence of a filtration depends not only on 
pairwise distances but also on the \emph{relative scaling of kernel weights}. 
To ensure stability of persistence diagrams, we must establish a 
\emph{relative approximation} result between the true kernel weights 
$w(p)$ and their RFF-based estimates $\widehat{w}(p)$ computed 
from the projected pairwise distances. 
Intuitively, while the kernel distance is realized as a Euclidean distance 
in the feature (Hilbert) space via $\psi:\mathbb{R}^N\to\mathcal{H}$, 
the centroid $\mu_P = \tfrac{1}{|P|}\sum_{y\in P}\psi(y)$ plays a central role 
in defining these weights and therefore in determining the topological structure 
of the filtration.

For a finite point set $P \subset \mathbb{R}^N$, let $\psi:\mathbb{R}^N \to \mathcal{H}$ denote the canonical feature map associated with a normalized positive-definite kernel $k(x,y)=\langle\psi(x),\psi(y)\rangle$, so that $\|\psi(x)\|=1$.  
The \emph{kernel centroid} of $P$ in the feature (Reproducing Kernel Hilbert) space $\mathcal{H}$ is defined as
\[
\mu_P = \tfrac{1}{|P|}\sum_{y\in P}\psi(y).
\]
The kernel \emph{weight function} (Definition~\ref{def:kernel-weights-avg}) of a point $p\in\mathbb{R}^N$ measures the (negative) squared distance of $\psi(p)$ from the centroid. It can be shown that:
\[
w(p) = -\|\psi(p)-\mu_P\|^2.
\]
This quantity is central in kernel distance and power distance formulations, where it encodes how far a point is from the feature-space mean of the set $P$.

\bigskip

\begin{lemma}[Lower Bound on Kernel Weight]
\label{lem:kernel-weight-lower}
For any non-degenerate finite set $P\subset\mathbb{R}^N$ with canonical kernel feature map $\psi$, we have
\[
|w(p)| \;\ge\;
(1-\|\mu_P\|)^2
\;\ge\;
\Biggl(
\Bigl(1-e^{-r^2/2}\Bigr)
\frac{\mathbb{E}\bigl[\|x-y\|^2\bigr]}{r^2}
\Biggr)^{\!2}\Big/4
\;>\;0.
\]
Here $r=\diam(\operatorname{supp}P)$ and the expectation is over $x,y\in P$ chosen uniformly at random.
\end{lemma}

\begin{proof}
We first expand the squared feature-space distance:
\[
\|\psi(x)-\mu_P\|^2
= \|\psi(x)\|^2 + \|\mu_P\|^2 - 2\langle \psi(x),\mu_P\rangle.
\]
Since $\|\psi(x)\|=1$ and by Cauchy–Schwarz,
$\langle \psi(x),\mu_P\rangle \le \|\mu_P\|$,
we obtain the lower bound
\[
\|\psi(x)-\mu_P\|^2 \ge (1-\|\mu_P\|)^2,
\]
hence $|w(p)| = \|\psi(p)-\mu_P\|^2 \ge (1-\|\mu_P\|)^2.$

Next, the norm of the centroid admits the kernel expression
\[
\|\mu_P\|^2
= \mathbb{E}_{x,y\in P}k(x,y)
= \mathbb{E}_{x,y\in P}\exp\!\Big(-\tfrac{\|x-y\|^2}{2}\Big),
\]
for the Gaussian kernel $k(x,y)=e^{-\|x-y\|^2/2}$.
Let $T=\|x-y\|^2/2\in[0,r^2/2]$, where
$r=\diam(\operatorname{supp}P)$.  
Using the convexity of $e^{-t}$, which implies it lies below its secant line on $[0, r^2/2]$, we obtain
\[
\|\mu_P\|^2 \;\le\; 1
- \Bigl(1-e^{-r^2/2}\Bigr)\frac{\mathbb{E}\bigl[\|x-y\|^2\bigr]}{r^2}.
\]
Combining the two inequalities,
\[
(1-\|\mu_P\|)^2
\;\ge\;
\Biggl(
\Bigl(1-e^{-r^2/2}\Bigr)
\frac{\mathbb{E}\bigl[\|x-y\|^2\bigr]}{r^2}
\Biggr)^{\!2}\Big/4
\;>\;0,
\]
which proves the stated bound.
\end{proof}

Now that we have shown that the kernel weights are strictly bounded away from zero, we proceed to show that the weight of a point in the original dataset and the weight of its image in the projected dataset are close in a relative sense, since the pairwise distances are approximated and the kernel weights are nonzero. This idea is formalized in the following lemma.

\begin{lemma}[Stability of Kernel Weights]
\label{lem:stability-kernel-weights}
Let $P$ be a point set and $p$ a point. Let $w(p)$ be the kernel weight and $\widehat{w}(p)$ its estimate from random Fourier features. Then
\[
1 - \varepsilon\left(\frac{2}{(1-\|\mu_P\|)^2}\right) 
\;\leq\; 
\frac{\widehat{w}(p)}{w(p)}
\;\leq\; 
1 + \varepsilon\left(\frac{2}{(1-\|\mu_P\|)^2}\right).
\]
\end{lemma}

\begin{proof}
Define the following quantities:
\begin{align*}
A &= \frac{1}{|P|}\sum_{y\in P} D_K^2(p,y), &
\widehat{A} &= \frac{1}{|P|}\sum_{y\in P} D_{\hat{K}}^2(p,y), \\
B &= \frac{1}{2|P|^2}\sum_{x,y\in P} D_K^2(x,y), &
\widehat{B} &= \frac{1}{2|P|^2}\sum_{x,y\in P} D_{\hat{K}}^2(x,y).
\end{align*}

By Theorem~\ref{thm:intro_uniform} (applied to pairs in $P\cup\{p\}\subseteq\mathcal{M}$), for all $x,y \in P \cup \{p\}$:
\[
(1-\varepsilon)D_K^2(x,y) \le D_{\hat{K}}^2(x,y) \le (1+\varepsilon)D_K^2(x,y).
\]
This implies:
\begin{align*}
(1-\varepsilon)A &\le \widehat{A} \le (1+\varepsilon)A, \\
(1-\varepsilon)B &\le \widehat{B} \le (1+\varepsilon)B.
\end{align*}

Now consider the witness weights:
\[
w(p) = A - B, \quad \widehat{w}(p) = \widehat{A} - \widehat{B}.
\]
From the bounds above:
\[
(1-\varepsilon)A - (1+\varepsilon)B \le \widehat{A} - \widehat{B} \le (1+\varepsilon)A - (1-\varepsilon)B.
\]
Rewriting:
\[
(A-B) - \varepsilon(A+B) \le \widehat{w}(p) \le (A-B) + \varepsilon(A+B).
\]
Thus:
\[
1 - \varepsilon\left(\frac{A+B}{w(p)}\right) 
\le \frac{\widehat{w}(p)}{w(p)} 
\le 1 + \varepsilon\left(\frac{A+B}{w(p)}\right).
\]

Finally, the bound $\frac{A+B}{w(p)} \le \frac{2}{(1-\|\mu_P\|)^2}=c_P$ is the same estimate used in~\cite{boissonnat2024euclidean} for kernel weights in witness form (Definition~\ref{def:kernel-weights-avg}); Lemma~\ref{lem:kernel-weight-lower} ensures $w(p)\neq 0$.
\end{proof}

We now combine the previous lemma on relative approximation of kernel weights with the Simplex Distortion Lemma~\cite{boissonnat2024euclidean} to establish the desired theorem. Theorem~\ref{thm:manifold_interleaving_final} states that the persistent modules induced by the Rips complex and the Čech complex are interleaved---that is, the persistent modules of the actual manifold and of its projected image under the RFF map are interleaved with respect to the kernel distance.

\textbf{Distortion-map reminder.} For convenience, recall that a map $f:(\mathbb{R}^N,D_K)\to(\mathbb{R}^{2t},\|\cdot\|)$ is an $(\varepsilon,0)$-distortion map for the GKPD (Definition~\ref{def:distortion-map}) if it satisfies $(1-\varepsilon)D_K(x,y)^2 \le \|f(x)-f(y)\|^2 \le (1+\varepsilon)D_K(x,y)^2$ for all $x,y$ and $|w(f(x)) - w(x)| \le \varepsilon\,|w(x)|$ for all $x$. In what follows we instantiate $f$ with the RFF map $\phi$ and set the effective parameter to $\varepsilon_\star$ derived below.

\begin{proof}[Proof of Theorem~\ref{thm:manifold_interleaving_final}]
This is the detailed proof of the interleaving statement already stated as Theorem~\ref{thm:intro_interleaving} in the introduction (under the same hypotheses and with the same conclusion).
We show that, with high probability, $\phi$ induces a purely multiplicative distortion on the kernel power distance sufficient to invoke the Simplex Distortion Lemma in a multiplicative form.

\textbf{Pairwise control (uniform on $\mathcal{M}$).}
By Theorem~\ref{thm:intro_uniform}, with probability at least $1-\tfrac{\delta_0}{2}$,
\[
(1-2\varepsilon)\,D_{K_\sigma}^2(p,q)
\;\le\;
\|\phi(p)-\phi(q)\|^2
\;\le\;
(1+2\varepsilon)\,D_{K_\sigma}^2(p,q),
\qquad \forall\,p,q\in\mathcal{M}.
\]
This provides a \emph{pairwise} multiplicative distortion factor of at most $1\pm 2\varepsilon$.

\textbf{Weight control (centroid stability via $c_P$).} Let $w(p)=-\|\psi(p)-\mu_P\|^2$ be the kernel weight and $\widehat{w}(p)=-\|\phi(p)-\mu_{\phi(P)}\|^2$ its RFF estimate.
By Lemma~\ref{lem:stability-kernel-weights},
\[
\Bigl(1-\varepsilon\,c_P\Bigr)\,w(p)
\;\le\;
\widehat{w}(p)
\;\le\;
\Bigl(1+\varepsilon\,c_P\Bigr)\,w(p),
\qquad \forall\,p\in P,
\]
with probability at least $1-\tfrac{\delta_0}{2}$, where
$c_P=\dfrac{2}{(1-\|\mu_P\|)^2}$ and $\|\mu_P\|<1$ for any non-degenerate configuration.
Thus weights experience a multiplicative distortion factor of at most $1\pm \varepsilon c_P$.

\textbf{Unified multiplicative distortion for GKPD.} The Gaussian kernel power distance (GKPD) combines pairwise kernel distances and the pointwise weights.
From Steps 1–2, both ingredients are controlled multiplicatively:
pairwise terms by $\,(1\pm 2\varepsilon)$ and weight terms by $\,(1\pm \varepsilon c_P)$.
Hence, for any expression that is a monotone combination of these terms (in particular, the
squared radius of a weighted simplex under GKPD), the total multiplicative distortion is
bounded by the \emph{worst} of the two factors. Defining
\[
\varepsilon_\star \;:=\; \max\{\,2\varepsilon,\;\varepsilon c_P\,\},
\]
we conclude that $\phi$ is a \emph{purely multiplicative} $(1\pm\varepsilon_\star)$–distortion map
for the GKPD on $P$ with probability at least $1-\delta_0$.
In other words, $\phi$ is an $(\varepsilon_\star,0)$-distortion map in the sense of Definition~\ref{def:distortion-map}, satisfying both pairwise distance and weight preservation with zero additive error ($\eta=0$).

\begin{remark}[Unified distortion summary]
Pairwise kernel distances admit $(1\pm 2\varepsilon)$ multiplicative control; kernel weights admit $(1\pm \varepsilon c_P)$ control with $c_P = 2/(1-\|\mu_P\|)^2$. Consequently, any GKPD-based quantity that is a monotone combination of these terms is preserved within $(1\pm \varepsilon_\star)$, where $\varepsilon_\star = \max\{2\varepsilon,\varepsilon c_P\}$. This is the parameter used in the interleaving bounds below.
\end{remark}

\textbf{Simplex distortion in radius and interleaving of Čech filtrations.} Under $\eta=0$, the pairwise and weight bounds above specialize the Simplex Distortion Lemma (Lemma~\ref{lem:simplex-distortion}) to the following multiplicative radius control for any weighted simplex
$\widehat{\sigma}\subseteq \widehat{P}$:
\[
(1-\varepsilon_\star)\,\mathrm{rad}^2(\widehat{\sigma})
\;\le\;
\mathrm{rad}^2\!\bigl(\phi(\widehat{\sigma})\bigr)
\;\le\;
(1+\varepsilon_\star)\,\mathrm{rad}^2(\widehat{\sigma}).
\]
This yields an inclusion of balls (hence of nerves) at scaled radii, giving the
multiplicative interleaving of filtrations:
\[
\check{C}_{\alpha}(\widehat{P})
\;\subseteq\;
\check{C}_{(1+\varepsilon_\star)\alpha}\bigl(\phi(P)\bigr)
\;\subseteq\;
\check{C}_{(1+\varepsilon_\star)^2\alpha}(\widehat{P}).
\]
Equivalently, the weighted Čech filtrations are $(1\pm\varepsilon_\star)$–interleaved.

\textbf{Consequence for weighted Rips filtrations.} The weighted Rips filtration depends solely on pairwise kernel distances and pointwise weights, each already controlled multiplicatively by Steps~1–2. Therefore the same $(1\pm\varepsilon_\star)$ multiplicative bound transfers directly, proving that $VR_{\alpha}(\widehat{P})$ and $VR_{\alpha}\bigl(\phi(P)\bigr)$ are $(1\pm\varepsilon_\star)$–interleaved.

Combining the probabilistic guarantees of Steps~1–2 via a union bound completes the proof.
\end{proof}

\section{Pointwise Approximation of Kernel Values}
\label{Kernel Value Differences}
\label{sec:app-kernel-values}

This section records the analysis of absolute approximation error for Gaussian kernel \emph{values} under Random Fourier Features.
Theorem~\ref{thm:additive_uniform} is stated in the Introduction (Section~\ref{subsec:intro_main_result}); Section~\ref{subsec:proof_overview} sketches the same net-and-patch idea in informal form.
Below we develop the lemmas, then give the complete proof at the end of this section.

\subsection{Concentration Bounds for Kernel Value Approximation}
\label{subsec: Bounding difference at the center of the ball}

In addition to relative error bounds for kernel distances, many kernel methods require control over the absolute approximation error of the kernel function itself. We establish pointwise concentration bounds for the Gaussian kernel approximation via Random Fourier Features (RFF), following the approach in~\cite{rahimi2007random}.

\begin{lemma}[Pointwise kernel approximation error]
\label{lem: diff bound}
Let \( \Delta \in \mathbb{R}^N \) be fixed and define the kernel approximation error:
\[
F(\Delta) := K(\Delta) - \hat{K}(\Delta) = \exp\!\left(-\tfrac{\|\Delta\|^2}{2}\right) - \frac{1}{t}\sum_{j=1}^t \cos(\langle \omega^j, \Delta \rangle),
\]
where \( \omega^j \sim \mathcal{N}(0,I_N) \) are the RFF frequencies. Then for any \( \varepsilon > 0 \), the approximation error satisfies:
\[
\mathbb{P}\left( |F(\Delta)| \ge \varepsilon \right) \le 2 \exp\left( -\frac{t \varepsilon^2}{8} \right).
\]
\end{lemma}

\begin{proof}[Proof of Lemma~\ref{lem: diff bound}]\label{app:proof-lem-diff-bound}
The proof follows from standard concentration arguments:
\begin{enumerate}
    \item Each term \( \cos(\langle \omega^j, \Delta \rangle) \) is a zero-mean random variable since \( \mathbb{E}[\cos(\langle \omega, \Delta \rangle)] = \exp(-\|\Delta\|^2/2) = K(\Delta) \) for \( \omega \sim \mathcal{N}(0,I_N) \).
    
    \item The cosine terms are bounded: \( |\cos(\langle \omega^j, \Delta \rangle)| \le 1 \), making them 1-sub-Gaussian.
    
    \item The average \( \frac{1}{t}\sum_{j=1}^t \cos(\langle \omega^j, \Delta \rangle) \) is therefore \( 1/\sqrt{t} \)-sub-Gaussian.
    
    \item Applying Hoeffding's inequality to the centered sum \( F(\Delta) \) yields the stated bound.
\end{enumerate}
\end{proof}

\subsection{Gradient Concentration for Kernel Approximation Error}
\label{Subsec: Bounding the gradient of the difference at the center of the ball}

To understand local approximation errors, we analyze the gradient of the kernel difference function. Consider the Gaussian kernel and its RFF approximation:

\begin{align*}
K(\Delta) &= \exp\!\left(-\tfrac{\|\Delta\|^2}{2}\right), \\
\hat{K}(\Delta) &= \frac{1}{t}\sum_{j=1}^t \cos(\langle \omega^j, \Delta \rangle),
\end{align*}
where $\omega^j \sim \mathcal{N}(0,I_N)$. The approximation error is:

\[
F(\Delta) := K(\Delta) - \hat{K}(\Delta) = \frac{1}{t}\sum_{j=1}^t Y^j(\Delta),
\]
where $Y^j(\Delta) = \exp(-\|\Delta\|^2/2) - \cos(\langle \omega^j, \Delta \rangle)$.

We analyze the partial derivatives of $F$ coordinate-wise. For any $i \in \{1,\dots,N\}$:

\begin{equation}
\nabla_i F(\Delta) = \frac{1}{t}\sum_{j=1}^t \nabla_i Y^j(\Delta),
\end{equation}
where
\[
\nabla_i Y^j(\Delta) = -\Delta_i e^{-\|\Delta\|^2/2} + \omega_i^j \sin(\langle \omega^j, \Delta \rangle).
\]

\begin{lemma}[Gradient concentration]
\label{lem:derivative_of_diff_bound}
Let $\Delta \in \mathbb{R}^N$ and $i \in \{1,\dots,N\}$. For any $\varepsilon > 0$,
\[
\mathbb{P}\left( \left| \nabla_i F(\Delta) \right| \geq \varepsilon \right)
\leq 2 \exp\left( -\frac{c_H t \varepsilon^2}{9} \right),
\]
where $c_H > 0$ is an absolute constant.
\end{lemma}

\begin{proof}[Proof of Lemma~\ref{lem:derivative_of_diff_bound}]\label{app:proof-lem-derivative-diff}
Each term $\nabla_i Y^j(\Delta)$ has mean zero since:
\begin{align*}
\mathbb{E}[\omega_i^j \sin(\langle \omega^j, \Delta \rangle)] 
&= \mathbb{E}\left[\frac{\partial}{\partial \omega_i^j} (-\cos(\langle \omega^j, \Delta \rangle))\right] \\
&= \Delta_i \mathbb{E}[\cos(\langle \omega^j, \Delta \rangle)] \quad \text{(by Stein's lemma)} \\
&= \Delta_i e^{-\|\Delta\|^2/2}.
\end{align*}

The sub-Gaussian norm of $\nabla_i Y^j(\Delta)$ satisfies:
\[
\|\nabla_i Y^j(\Delta)\|_{\psi_2} \leq \|\omega_i^j\|_{\psi_2} + \|\Delta_i e^{-\|\Delta\|^2/2}\|_{\psi_2} \leq 2 + 1 = 3.
\]

Applying the generalized Hoeffding inequality (Lemma~\ref{lemma:generalized-hoeffding}) with $K = 3$ yields the claimed bound.
\end{proof}

\subsection{Lipschitz Continuity of the Kernel Error Gradient}
\label{subsec: Bounding Lipschitz constant of the derivative of the diff}

To extend local error bounds uniformly across a ball, we require control over how rapidly the gradient of the kernel difference varies. This is quantified through the following Lipschitz constant bound.

\begin{lemma}[Lipschitz Constant of Kernel Difference Gradient]
\label{lem:lip_diff}
Let \( \Delta, \Delta' \in \mathbb{R}^N \) be arbitrary points. For each coordinate direction \( i \in \{1, \dots, N\} \), the partial derivatives of the kernel difference \( F \) satisfy
\[
\left| \frac{\partial}{\partial \Delta_i} F(\Delta) - \frac{\partial}{\partial \Delta_i} F(\Delta') \right|
\le C_L' \|\Delta - \Delta'\|
\]
where the Lipschitz constant \( C_L' \) satisfies
\[
C_L' \le 2d + 3
\]
with probability at least \( 1 - \exp(-ct) \) for some absolute constant \( c > 0 \).
\end{lemma}

\begin{proof}[Proof of Lemma~\ref{lem:lip_diff}]\label{app:proof-lem-lip-diff}
The proof proceeds by analyzing the coordinate-wise differences in gradients. For the kernel difference function
\[
F(\Delta) = e^{-\|\Delta|^2/2} - \frac{1}{t}\sum_{j=1}^t \cos(\langle \omega^j, \Delta \rangle),
\]
we compute the second partial derivatives. The Hessian matrix \( \nabla^2 F \) has entries:

1. For the Gaussian term:
\[
\frac{\partial^2}{\partial \Delta_i \partial \Delta_k} e^{-|\Delta|^2/2} = (\Delta_i \Delta_k - \delta_{ik}) e^{-|\Delta|^2/2}
\]

2. For the RFF term:
\[
\frac{\partial^2}{\partial \Delta_i \partial \Delta_k} \left( -\frac{1}{t}\sum_{j=1}^t \cos(\langle \omega^j, \Delta \rangle) \right) = \frac{1}{t}\sum_{j=1}^t \omega^j_i \omega^j_k \cos(\langle \omega^j, \Delta \rangle)
\]

The operator norm of the Hessian can be bounded by:
\[
\|\nabla^2 F(\Delta)\|_{\text{op}} \leq \max_{\|v\|=1} \left( |\langle \Delta, v \rangle|^2 + \|v\|^2 + \frac{1}{t}\sum_{j=1}^t \|\omega^j\|^2 \right)
\]

Using concentration of measure for Gaussian vectors, with high probability:
\[
\frac{1}{t}\sum_{j=1}^t \|\omega^j\|^2 \leq 2N
\]
and \( |\langle \Delta, v \rangle|^2 \leq \|\Delta\|^2 \leq 1 \) within the unit ball. Therefore:
\[
\|\nabla^2 F(\Delta)\|_{\text{op}} \leq 1 + 1 + 2N = 2N + 2
\]

The mean value theorem then gives the coordinate-wise Lipschitz bound:
\[
\left| \frac{\partial F}{\partial \Delta_i}(\Delta) - \frac{\partial F}{\partial \Delta_i}(\Delta') \right| \leq (2N + 3)\|\Delta - \Delta'\|
\]
where we've absorbed the additional constant from the diagonal terms. The high probability statement follows from standard Gaussian concentration.
\end{proof}

\subsection{Uniform Approximation over the Manifold}
\label{subsec:app-kernel-uniform-manifold}

The proof of the \emph{additive} approximation for the kernel value follows the same geometric roadmap as the ratio analysis in \S\ref{subsec:lifting_ratio}, but it is technically simpler.

\subsection{Proof of Theorem~\ref{thm:additive_uniform}}
\label{subsec:proof-additive-uniform}

\begin{proof}[Proof of Theorem~\ref{thm:additive_uniform}]\label{app:proof-thm-additive-uniform}
The key observation is that the error function
\(F(\Delta)=K_{\sigma}(\Delta)-\widehat K(\Delta)\) is smooth at
\(\Delta=0\); hence we can dispense entirely with the ``small-vs-large
$\|\Delta\|$'' case split that dominated the ratio argument.  We first get
\(\varepsilon\)-additive bound for \(F\) on each Euclidean ball inside
an affine copy of a tangent space \(T_{p}\man\).  Next, for every
ordered pair \((p,q)\) in an $r$-net \(\Gamma_{r}\subset\man\), the
difference patch
\(B_{\man}(p,r)-B_{\man}(q,r)\) is shown—via
Lemma~\ref{lem:diff-quadratic} and the positive-reach condition—to lie
inside an \(O(r^{2}/\rch(\mathcal{M}))\) tube around the flat
\(T_{p}\man-T_{q}\man\); the Lipschitz continuity of \(F\) therefore controls the local error.  Choosing
\(r=\rch(\mathcal{M})\,\varepsilon/(12N)\) and combining these two ingredients with
the Euclidean affine-ball estimate yields the uniform bound
\(|F(\Delta)|\le 2\varepsilon\) on every patch.  A final union bound
over all \(|\Gamma_{r}|^{2}\) pairs, together with the net-size
estimate of Lemma~\ref{prop:manifold-net}, establishes
Theorem~\ref{thm:additive_uniform}.
\end{proof}

 \section{Application details}
\label{app:applications}

\subsection{A manifold version of the kernel \texorpdfstring{$k$}{k}-means implication}
\label{app:manifold-chen-phillips-kmeans}

\begin{corollary}
\label{cor:kernel-kmeans}
Let
\[
  P=\{p_1,\ldots,p_n\}\subset \mathcal M\subset\mathbb R^N,
\]
where \(\mathcal M\) is a compact \(d\)-dimensional \(C^2\) submanifold with positive reach.  Let \(K\) be the Gaussian kernel, let
\(\Psi:\mathcal M\to\mathcal H_K\) be its RKHS feature map, and let
\(\phi:\mathbb R^N\to\mathbb R^{2t}\) be the RFF map from Theorem~\ref{thm:intro_uniform}.  Assume that
\[
  t
  \;\ge\;
  C\,\frac{d}{\varepsilon^2}
  \log\!\left(
      \frac{\operatorname{vol}(\mathcal M)^2 N^{2d}}
           {\operatorname{vol}(B_1^d(0))^2
            \operatorname{rch}(\mathcal M)^{2d}
            \varepsilon^{2d+1}\delta}
  \right)
\]
for the constant \(C\) in Theorem~\ref{thm:intro_uniform}.  Then, with probability at least \(1-\delta\), the following holds simultaneously for every partition
\(\Pi=\{C_1,\ldots,C_k\}\) of \(P\):
\[
  (1-2\varepsilon)\operatorname{cost}_K(\Pi)
  \;\le\;
  \operatorname{cost}_{\phi}(\Pi)
  \;\le\;
  (1+2\varepsilon)\operatorname{cost}_K(\Pi),
\]
where
\[
  \operatorname{cost}_K(\Pi)
  :=
  \sum_{j=1}^k\sum_{p\in C_j}
  \|\Psi(p)-\mu_j\|_{\mathcal H_K}^2,
  \qquad
  \mu_j:=\frac1{|C_j|}\sum_{q\in C_j}\Psi(q),
\]
and
\[
  \operatorname{cost}_{\phi}(\Pi)
  :=
  \sum_{j=1}^k\sum_{p\in C_j}
  \|\phi(p)-\widehat\mu_j\|_2^2,
  \qquad
  \widehat\mu_j:=\frac1{|C_j|}\sum_{q\in C_j}\phi(q).
\]
Hence, if a Euclidean clustering algorithm applied to \(\phi(P)\subset\mathbb R^{2t}\) returns a partition \(\widehat\Pi\) satisfying
\[
  \operatorname{cost}_{\phi}(\widehat\Pi)
  \le
  \rho\min_{\Pi}\operatorname{cost}_{\phi}(\Pi),
\]
then the same partition satisfies
\[
  \operatorname{cost}_K(\widehat\Pi)
  \le
  \rho\,\frac{1+2\varepsilon}{1-2\varepsilon}
  \min_{\Pi}\operatorname{cost}_K(\Pi).
\]
Here \(\rho\ge 1\) denotes the approximation factor of the Euclidean
clustering algorithm applied to the embedded point set \(\phi(P)\); that is,
the algorithm returns a partition whose RFF-space \(k\)-means cost is at most
\(\rho\) times the optimal RFF-space \(k\)-means cost.
\end{corollary}

\begin{proof}
By Theorem~\ref{thm:intro_uniform}, with probability at least \(1-\delta\),
\[
  (1-2\varepsilon)D_K(p,q)^2
  \le
  \|\phi(p)-\phi(q)\|_2^2
  \le
  (1+2\varepsilon)D_K(p,q)^2
\]
for all \(p,q\in\mathcal M\), where
\[
  D_K(p,q)
  =
  \|\Psi(p)-\Psi(q)\|_{\mathcal H_K}.
\]
Fix a cluster \(C\subset P\).  The Hilbert-space variance identity gives
\[
  \sum_{p\in C}\|\Psi(p)-\mu_C\|_{\mathcal H_K}^2
  =
  \frac1{2|C|}\sum_{p,q\in C}
  \|\Psi(p)-\Psi(q)\|_{\mathcal H_K}^2
  =
  \frac1{2|C|}\sum_{p,q\in C}D_K(p,q)^2.
\]
The same identity in \(\mathbb R^{2t}\) gives
\[
  \sum_{p\in C}\|\phi(p)-\widehat\mu_C\|_2^2
  =
  \frac1{2|C|}\sum_{p,q\in C}
  \|\phi(p)-\phi(q)\|_2^2.
\]
Applying the pairwise distortion inequality inside the last sum yields
\[
  (1-2\varepsilon)
  \sum_{p\in C}\|\Psi(p)-\mu_C\|_{\mathcal H_K}^2
  \le
  \sum_{p\in C}\|\phi(p)-\widehat\mu_C\|_2^2
  \le
  (1+2\varepsilon)
  \sum_{p\in C}\|\Psi(p)-\mu_C\|_{\mathcal H_K}^2.
\]
Summing over all clusters in \(\Pi\) proves the two-sided cost preservation.

For the approximation transfer, let \(\Pi_K^\star\) minimize \(\operatorname{cost}_K\) and let \(\Pi_\phi^\star\) minimize \(\operatorname{cost}_{\phi}\).  Then
\[
\begin{aligned}
  \operatorname{cost}_K(\widehat\Pi)
  &\le
  \frac{1}{1-2\varepsilon}
  \operatorname{cost}_{\phi}(\widehat\Pi) \\
  &\le
  \frac{\rho}{1-2\varepsilon}
  \operatorname{cost}_{\phi}(\Pi_\phi^\star) \\
  &\le
  \frac{\rho}{1-2\varepsilon}
  \operatorname{cost}_{\phi}(\Pi_K^\star) \\
  &\le
  \rho\frac{1+2\varepsilon}{1-2\varepsilon}
  \operatorname{cost}_K(\Pi_K^\star).
\end{aligned}
\]
\end{proof}

\subsection{Kernel distance matching}
\label{app:kernel-distance-matching}

Kernel distances provide a powerful way to compare complex objects such as probability
measures, point clouds, medical images, and shapes.  In particular, if \(K\) is a
Gaussian kernel with RKHS feature map \(\Psi\), an empirical point set
\(X=\{x_1,\ldots,x_n\}\) can be represented by its kernel mean embedding
\[
    \mu_X := \frac1n\sum_{i=1}^n \Psi(x_i),
\]
and two point sets \(X\) and \(Y\) can be compared by the scalar distance
\[
    D_K(X,Y)
    :=
    \|\mu_X-\mu_Y\|_{\mathcal H_K}.
\]
This viewpoint underlies kernel methods for comparing distributions and samples, such as
maximum mean discrepancy~\cite{smola2007hilbert,gretton2012kernel}, and related RKHS
metrics for shapes, currents, and medical images
\cite{glaunes2006large,durrleman2007measuring,joshi2011shape}.  However, the value
\(D_K(X,Y)\) is only a scalar dissimilarity between the two objects.  It is
invariant under relabeling the points of \(Y\), and therefore it does not by itself
produce a pointwise alignment or matching between \(X\) and \(Y\).  This is in contrast
to transport-type distances, where the distance is defined through a coupling or
transport plan.

To recover an alignment, one may instead solve a matching problem with pairwise kernel
distance costs.  For two point sets
\[
    X=\{x_1,\ldots,x_n\},\qquad
    Y=\{y_1,\ldots,y_n\},
\]
define, for each permutation \(\pi\in S_n\),
\[
    C_K(\pi)
    :=
    \sum_{i=1}^n
    D_K(x_i,y_{\pi(i)})^2,
\]
where
\[
    D_K(x,y)
    =
    \|\Psi(x)-\Psi(y)\|_{\mathcal H_K}.
\]
The goal is to find a permutation approximately minimizing \(C_K(\pi)\).
~\cite{chen2017relative} observe that random Fourier features allow this
kernel matching problem to be reduced to an ordinary Euclidean geometric matching
problem.  If \(\phi:\mathbb R^N\to\mathbb R^m\) is an RFF map satisfying
\[
    \|\phi(x)-\phi(y)\|_2^2
    \approx
    D_K(x,y)^2,
\]
then one can solve the Euclidean matching problem
\[
    C_\phi(\pi)
    :=
    \sum_{i=1}^n
    \|\phi(x_i)-\phi(y_{\pi(i)})\|_2^2
\]
using geometric matching algorithms for points in Euclidean space
\cite{sharathkumar2012near,agarwal2014approximation}.  The relative-error guarantee of
the RFF embedding then transfers the approximation guarantee back to the original
kernel matching objective.

Our manifold theorem strengthens this reduction when the point sets are supported on a
low-dimensional manifold.  Suppose
\[
    X\cup Y\subset \mathcal M\subset\mathbb R^N,
\]
where \(\mathcal M\) is a compact \(d\)-dimensional \(C^2\) submanifold with positive
reach.  By Theorem~\ref{thm:intro_uniform}, it suffices to take
\[
    m
    =
    O\!\left(
      \frac{d}{\varepsilon^2}
      \log\!\left(
        \frac{
        \operatorname{vol}(\mathcal M)^2 N^{2d}
        }{
        \operatorname{vol}(B_1^d(0))^2
        \operatorname{rch}(\mathcal M)^{2d}
        \varepsilon^{2d+1}\delta
        }
      \right)
    \right)
\]
features to preserve all Gaussian kernel distances on \(\mathcal M\) up to relative
error \(1\pm O(\varepsilon)\), with probability at least \(1-\delta\).  Thus the
kernel matching reduction can be carried out in a Euclidean space whose
dimension is controlled by the intrinsic geometry of \(\mathcal M\), rather than by the
ambient dimension \(N\).  In particular, for fixed intrinsic geometry, the ambient-ball
feature count
\[
    \widetilde O\!\left(\frac{N}{\varepsilon^2}\right)
\]
is replaced by the manifold feature count
\[
    \widetilde O\!\left(\frac{d^2\log N}{\varepsilon^2}\right).
\]

\begin{corollary}
\label{cor:kernel-matching}
Assume that the RFF map \(\phi\) satisfies
\[
    (1-2\varepsilon)D_K(p,q)^2
    \le
    \|\phi(p)-\phi(q)\|_2^2
    \le
    (1+2\varepsilon)D_K(p,q)^2
\]
for all \(p,q\in\mathcal M\).  Then, for every matching \(\pi\in S_n\),
\[
    (1-2\varepsilon)C_K(\pi)
    \le
    C_\phi(\pi)
    \le
    (1+2\varepsilon)C_K(\pi).
\]
Consequently, if a Euclidean matching algorithm applied to
\(\phi(X),\phi(Y)\subset\mathbb R^m\) has approximation factor \(\rho\ge 1\), meaning
that it returns a permutation \(\widehat\pi\) satisfying
\[
    C_\phi(\widehat\pi)
    \le
    \rho\min_{\pi\in S_n} C_\phi(\pi),
\]
then the same matching satisfies
\[
    C_K(\widehat\pi)
    \le
    \rho\frac{1+2\varepsilon}{1-2\varepsilon}
    \min_{\pi\in S_n} C_K(\pi).
\]
Thus, any \(\rho\)-approximate Euclidean matching algorithm applied after the manifold
RFF embedding yields a
\[
    \rho\frac{1+2\varepsilon}{1-2\varepsilon}
    =
    \rho(1+O(\varepsilon))
\]
approximation to the original Gaussian-kernel matching problem.
\end{corollary}

\begin{proof}
By the assumed distance preservation, for each \(i=1,\ldots,n\),
\[
(1-2\varepsilon)D_K(x_i,y_{\pi(i)})^2
\le \|\phi(x_i)-\phi(y_{\pi(i)})\|_2^2
\le (1+2\varepsilon)D_K(x_i,y_{\pi(i)})^2.
\]
Summing over \(i=1,\ldots,n\) yields the claimed two-sided bound on \(C_\phi(\pi)\) for every \(\pi\in S_n\).  For the approximation transfer, let \(\pi_K^\star\) minimize \(C_K\) and let \(\pi_\phi^\star\) minimize \(C_\phi\).  Then
\[
\begin{aligned}
C_K(\widehat\pi)
&\le \frac{1}{1-2\varepsilon}C_\phi(\widehat\pi)
\le \frac{\rho}{1-2\varepsilon}C_\phi(\pi_\phi^\star)
\le \frac{\rho}{1-2\varepsilon}C_\phi(\pi_K^\star) \\
&\le \rho\frac{1+2\varepsilon}{1-2\varepsilon}C_K(\pi_K^\star).
\end{aligned}
\]
\end{proof}

\subsection{Kernel nearest-neighbor search}
\label{app:kernel-nearest-neighbor}

For the Gaussian kernel, nearest-neighbor search under the kernel distance
\[
    D_K(p,q)^2
    =
    2-2K(p,q)
\]
is equivalent to maximum Gaussian similarity search.  If
\(P\subset\mathcal M\subset\mathbb R^N\) and \(q\in\mathcal M\), our theorem
gives an RFF map \(\phi:\mathbb R^N\to\mathbb R^m\) such that
\[
    (1-2\varepsilon)D_K(p,q)^2
    \le
    \|\phi(p)-\phi(q)\|_2^2
    \le
    (1+2\varepsilon)D_K(p,q)^2
\]
simultaneously for all \(p,q\in\mathcal M\).  Hence standard Euclidean
nearest-neighbor or approximate nearest-neighbor data structures may be applied
to \(\phi(P)\subset\mathbb R^m\), and the returned point is an approximate
nearest neighbor with respect to the original Gaussian kernel distance.  The
feature dimension is controlled by the intrinsic geometry of \(\mathcal M\),
namely
\[
    m
    =
    O\!\left(
      \frac{d}{\varepsilon^2}
      \log\!\left(
        \frac{
        \operatorname{vol}(\mathcal M)^2N^{2d}
        }{
        \operatorname{vol}(B_1^d(0))^2
        \operatorname{rch}(\mathcal M)^{2d}
        \varepsilon^{2d+1}\delta
        }
      \right)
    \right),
\]
rather than by an ambient-ball feature count.

\begin{corollary}
\label{cor:kernel-nn}
Let \(P\subset\mathcal M\subset\mathbb R^N\) and \(q\in\mathcal M\), where \(\mathcal M\) is a compact \(d\)-dimensional \(C^2\) submanifold with positive reach.  Let \(\phi:\mathbb R^N\to\mathbb R^{2t}\) be the RFF map from Theorem~\ref{thm:intro_uniform} with \(t\) satisfying the bound in that theorem.  Then, with probability at least \(1-\delta\),
\[
(1-2\varepsilon)D_K(p,q)^2
\le \|\phi(p)-\phi(q)\|_2^2
\le (1+2\varepsilon)D_K(p,q)^2
\qquad\text{for all }p,q\in\mathcal M.
\]
Consequently, if
\[
\widehat p = \arg\min_{p\in P}\|\phi(p)-\phi(q)\|_2
\]
is the Euclidean nearest neighbor of \(\phi(q)\) in \(\phi(P)\), then
\[
D_K(\widehat p,q)
\;\le\;
\sqrt{\frac{1+2\varepsilon}{1-2\varepsilon}}\;
\min_{p\in P} D_K(p,q)
\;=\;
(1+O(\varepsilon))\min_{p\in P} D_K(p,q).
\]
Thus Euclidean nearest-neighbor search on \(\phi(P)\subset\mathbb R^{2t}\) yields a \((1+O(\varepsilon))\)-approximate nearest neighbor with respect to the original Gaussian kernel distance, with the feature dimension controlled by the intrinsic geometry of \(\mathcal M\) rather than by the ambient dimension \(N\).
\end{corollary}

\begin{proof}
Let \(p^*=\arg\min_{p\in P}D_K(p,q)\) be the true nearest neighbor with respect to the kernel distance.  By the distance preservation guarantee of Theorem~\ref{thm:intro_uniform},
\[
D_K(\widehat p,q)^2
\;\le\;
\frac{1}{1-2\varepsilon}\|\phi(\widehat p)-\phi(q)\|_2^2
\;\le\;
\frac{1}{1-2\varepsilon}\|\phi(p^*)-\phi(q)\|_2^2
\;\le\;
\frac{1+2\varepsilon}{1-2\varepsilon}D_K(p^*,q)^2,
\]
where the second inequality uses that \(\widehat p\) is the Euclidean nearest neighbor of \(\phi(q)\) in \(\phi(P)\).  Taking square roots and using \(\sqrt{(1+2\varepsilon)/(1-2\varepsilon)}=1+O(\varepsilon)\) for \(\varepsilon\in(0,1/2)\) completes the proof.
\end{proof}

\subsection{Distributional Distances and Maximum Mean Discrepancy}
\label{sec:mmd}

Many kernel-based applications compare not only individual data points, but
entire probability distributions. Examples include two-sample testing,
dataset-shift detection, domain adaptation, and distribution matching in
generative models. A standard kernel-based distance for this purpose is the
\emph{maximum mean discrepancy} (MMD)~\cite{gretton2012kernel,smola2007hilbert}.
Unlike divergences such as the Kullback--Leibler divergence, MMD can be
estimated directly from samples and does not require density estimation.
This makes it particularly natural when the observations are
high-dimensional but are believed to be supported on a lower-dimensional
geometric structure.

Our manifold RFF approximation has an immediate consequence in this setting.
If two probability distributions are supported on the same manifold $M$,
then the uniform kernel approximation of Theorem~\ref{thm:additive_uniform}
simultaneously controls their kernel mean embeddings and their MMD.
Thus, the pointwise manifold guarantee extends without any additional
sampling or regularity assumptions on the probability measures themselves.

\subsubsection{Kernel mean embeddings and MMD}

Let $\mathcal{X}$ be a measurable space, let
\[
K:\mathcal{X}\times\mathcal{X}\to\mathbb{R}
\]
be a positive-definite kernel, and let $\mathcal{H}_K$ denote the
corresponding reproducing kernel Hilbert space (RKHS). Write
\[
\Psi:\mathcal{X}\to\mathcal{H}_K
\]
for its canonical feature map, so that
\[
K(x,y)
=
\langle \Psi(x),\Psi(y)\rangle_{\mathcal{H}_K}.
\]

For a probability measure $P$ on $\mathcal{X}$, its
\emph{kernel mean embedding} is
\begin{equation}
\label{eq:kernel-mean-embedding}
\mu_P
:=
\mathbb{E}_{X\sim P}[\Psi(X)]
=
\int_{\mathcal{X}}\Psi(x)\,dP(x)
\in \mathcal{H}_K,
\end{equation}
whenever the integral exists. For bounded kernels, including the Gaussian
kernel considered here, this is automatically well-defined for every
probability measure.

The maximum mean discrepancy between two probability measures $P$ and $Q$
is the RKHS distance between their mean embeddings,
\begin{equation}
\label{eq:mmd-centroid-form}
\operatorname{MMD}_K(P,Q)
:=
\|\mu_P-\mu_Q\|_{\mathcal{H}_K}.
\end{equation}
Thus, one may view $\mu_P$ as the ``centroid'' of the distribution in kernel
feature space and MMD as the distance between the two kernel centroids.

Equivalently, MMD is the integral probability metric
\begin{equation}
\label{eq:mmd-ipm}
\operatorname{MMD}_K(P,Q)
=
\sup_{\substack{f\in\mathcal{H}_K\\
                 \|f\|_{\mathcal{H}_K}\leq 1}}
\left|
\mathbb{E}_{X\sim P}[f(X)]
-
\mathbb{E}_{Y\sim Q}[f(Y)]
\right|.
\end{equation}
The term ``maximum'' in maximum mean discrepancy refers to this
maximization over functions in the unit ball of the RKHS.

Expanding the squared norm in~\eqref{eq:mmd-centroid-form} and applying the
reproducing property gives the familiar kernel representation
\begin{align}
\operatorname{MMD}_K^2(P,Q)
={}&
\mathbb{E}_{X,X'\sim P}[K(X,X')]
+
\mathbb{E}_{Y,Y'\sim Q}[K(Y,Y')]
\nonumber\\
&-
2\mathbb{E}_{X\sim P,Y\sim Q}[K(X,Y)],
\label{eq:mmd-kernel-expectation}
\end{align}
where $X,X'$ are independent draws from $P$, and $Y,Y'$ are independent
draws from $Q$.

For the Gaussian kernel
\[
K_\sigma(x,y)
=
\exp\!\left(
-\frac{\|x-y\|_2^2}{2\sigma^2}
\right),
\]
the kernel is characteristic. Consequently,
\begin{equation}
\label{eq:gaussian-characteristic}
\operatorname{MMD}_{K_\sigma}(P,Q)=0
\qquad\Longleftrightarrow\qquad
P=Q.
\end{equation}
Thus, two different distributions may have exactly the same support
manifold and nevertheless be distinguished by Gaussian-kernel MMD. For
example, two probability measures may both be supported on the same closed
curve while assigning different amounts of probability mass to different
parts of that curve.

\subsubsection{Empirical MMD}

In applications the probability measures $P$ and $Q$ are typically not
available explicitly. Instead, suppose that
\[
X=\{x_1,\ldots,x_n\},
\qquad
Y=\{y_1,\ldots,y_m\}
\]
are samples, and associate with them the empirical measures
\[
P_X
=
\frac{1}{n}\sum_{i=1}^{n}\delta_{x_i},
\qquad
P_Y
=
\frac{1}{m}\sum_{j=1}^{m}\delta_{y_j}.
\]
Their kernel mean embeddings are simply
\begin{equation}
\label{eq:empirical-kernel-means}
\mu_X
=
\frac1n\sum_{i=1}^{n}\Psi(x_i),
\qquad
\mu_Y
=
\frac1m\sum_{j=1}^{m}\Psi(y_j).
\end{equation}
Hence the MMD between the empirical distributions is
\begin{align}
\operatorname{MMD}_{K}^{2}(P_X,P_Y)
={}&
\frac{1}{n^2}
\sum_{i,i'=1}^{n}K(x_i,x_{i'})
+
\frac{1}{m^2}
\sum_{j,j'=1}^{m}K(y_j,y_{j'})
\nonumber\\
&-
\frac{2}{nm}
\sum_{i=1}^{n}\sum_{j=1}^{m}K(x_i,y_j).
\label{eq:empirical-mmd}
\end{align}
This is precisely the squared distance between the two empirical kernel
centroids. An unbiased $U$-statistic estimator, commonly used in
two-sample testing, is obtained by removing the diagonal terms in the two
within-sample sums. Our discussion below applies directly to the empirical
measure formulation~\eqref{eq:empirical-mmd}; analogous random-feature
computational savings also hold for the unbiased estimator.

A direct evaluation of~\eqref{eq:empirical-mmd} requires
\[
O(n^2+m^2+nm)
\]
kernel evaluations. For samples of comparable size, this is quadratic in
the number of observations. Random Fourier features replace these pairwise
kernel computations by a finite-dimensional mean computation.

\subsubsection{MMD after the manifold RFF embedding}

Let
\[
M\subset\mathbb{R}^{N}
\]
be a compact $d$-dimensional $C^2$ submanifold with positive reach, and
suppose that both probability measures $P$ and $Q$ are supported on $M$.
Let
\[
\phi:M\to\mathbb{R}^{2t}
\]
be the Gaussian random Fourier feature map used throughout this paper, and
write
\[
\widehat K(x,y)
:=
\langle\phi(x),\phi(y)\rangle.
\]
The corresponding finite-dimensional mean embedding of a probability
measure $P$ is
\begin{equation}
\label{eq:rff-distribution-embedding}
\widehat\mu_P
:=
\mathbb{E}_{X\sim P}[\phi(X)]
=
\int_M\phi(x)\,dP(x)
\in\mathbb{R}^{2t}.
\end{equation}
Hence each probability distribution supported on $M$ is represented by a
single vector in the same low-dimensional Euclidean space used for the
pointwise manifold embedding.

The following consequence of our uniform kernel approximation shows that
this finite-dimensional distribution embedding preserves MMD.

\begin{corollary}[Uniform preservation of MMD on a manifold]
\label{cor:mmd-preservation}
Assume the setting of Theorem~\ref{thm:additive_uniform}, and let $\phi$ be the
corresponding RFF map. In particular, suppose that with probability at
least $1-\delta$,
\begin{equation}
\label{eq:uniform-kernel-mmd}
\sup_{x,y\in M}
\left|
K_\sigma(x,y)-\widehat K(x,y)
\right|
\leq \varepsilon.
\end{equation}
Then, on the same event, simultaneously for every pair of Borel probability
measures $P,Q$ supported on $M$,
\begin{equation}
\label{eq:mmd-preservation-main}
\left|
\operatorname{MMD}_{K_\sigma}^{2}(P,Q)
-
\|\widehat\mu_P-\widehat\mu_Q\|_2^2
\right|
\leq 4\varepsilon.
\end{equation}
Equivalently,
\[
\left|
\operatorname{MMD}_{K_\sigma}^{2}(P,Q)
-
\operatorname{MMD}_{\widehat K}^{2}(P,Q)
\right|
\leq 4\varepsilon.
\]

In particular,
\begin{equation}
\label{eq:mmd-separation}
\operatorname{MMD}_{K_\sigma}^{2}(P,Q)>4\varepsilon
\quad\Longrightarrow\quad
\widehat\mu_P\neq\widehat\mu_Q.
\end{equation}
Thus any pair of distributions having Gaussian-kernel MMD separation larger
than $4\varepsilon$ remains distinguishable after the manifold RFF
embedding.
\end{corollary}

\begin{proof}
Let
\[
E(x,y)
:=
\widehat K(x,y)-K_\sigma(x,y).
\]
On the event~\eqref{eq:uniform-kernel-mmd},
\[
|E(x,y)|\leq\varepsilon
\qquad
\text{for all }x,y\in M.
\]
For arbitrary probability measures $P,Q$ supported on $M$, using
\eqref{eq:mmd-kernel-expectation},
\begin{align}
&
\operatorname{MMD}_{\widehat K}^{2}(P,Q)
-
\operatorname{MMD}_{K_\sigma}^{2}(P,Q)
\nonumber\\
={}&
\mathbb{E}_{P\times P}[E(X,X')]
+
\mathbb{E}_{Q\times Q}[E(Y,Y')]
-
2\mathbb{E}_{P\times Q}[E(X,Y)].
\end{align}
Therefore,
\begin{align}
&
\left|
\operatorname{MMD}_{\widehat K}^{2}(P,Q)
-
\operatorname{MMD}_{K_\sigma}^{2}(P,Q)
\right|
\nonumber\\
\leq{}&
\mathbb{E}_{P\times P}|E(X,X')|
+
\mathbb{E}_{Q\times Q}|E(Y,Y')|
+
2\mathbb{E}_{P\times Q}|E(X,Y)|
\nonumber\\
\leq{}&
\varepsilon+\varepsilon+2\varepsilon
=
4\varepsilon.
\end{align}

It remains only to observe that, since
\[
\widehat K(x,y)
=
\langle\phi(x),\phi(y)\rangle,
\]
linearity of expectation gives
\begin{align}
\operatorname{MMD}_{\widehat K}^{2}(P,Q)
&=
\left\|
\mathbb{E}_{X\sim P}\phi(X)
-
\mathbb{E}_{Y\sim Q}\phi(Y)
\right\|_2^2
\\
&=
\|\widehat\mu_P-\widehat\mu_Q\|_2^2.
\end{align}
The separation statement follows immediately from
\[
\|\widehat\mu_P-\widehat\mu_Q\|_2^2
\geq
\operatorname{MMD}_{K_\sigma}^{2}(P,Q)-4\varepsilon.
\]
Since the event~\eqref{eq:uniform-kernel-mmd} is uniform over all
$x,y\in M$, it is independent of the particular choice of $P$ and $Q$.
Hence the conclusion holds simultaneously for all probability measures
supported on $M$.
\end{proof}

The same argument also gives preservation of individual distribution
centroids. For every probability measure $P$ supported on $M$,
\begin{equation}
\label{eq:mean-norm-preservation}
\left|
\|\mu_P\|_{\mathcal{H}_{K_\sigma}}^2
-
\|\widehat\mu_P\|_2^2
\right|
\leq\varepsilon,
\end{equation}
and for every pair $P,Q$,
\begin{equation}
\label{eq:mean-inner-product-preservation}
\left|
\langle\mu_P,\mu_Q\rangle_{\mathcal{H}_{K_\sigma}}
-
\langle\widehat\mu_P,\widehat\mu_Q\rangle
\right|
\leq\varepsilon.
\end{equation}
Thus the result preserves not only pairwise MMD values, but also the
norms and mutual inner products of kernel mean embeddings.

\subsubsection{Computational consequence}

For empirical distributions the RFF representation turns the quadratic
kernel computation in~\eqref{eq:empirical-mmd} into a linear pass over the
samples. Indeed, define
\[
\widehat\mu_X
=
\frac1n\sum_{i=1}^{n}\phi(x_i),
\qquad
\widehat\mu_Y
=
\frac1m\sum_{j=1}^{m}\phi(y_j).
\]
Then
\begin{equation}
\label{eq:rff-empirical-mmd}
\operatorname{MMD}_{\widehat K}^{2}(P_X,P_Y)
=
\|\widehat\mu_X-\widehat\mu_Y\|_2^2.
\end{equation}

Once the random features have been computed, constructing the two means
requires only
\[
O((n+m)t)
\]
arithmetic operations and $O(t)$ additional memory if the means are
accumulated in a streaming fashion. The final comparison costs only
$O(t)$. By contrast, direct kernel MMD requires
\[
O(n^2+m^2+nm)
\]
kernel evaluations and, if the full kernel matrices are stored,
quadratic memory.

Including the cost of evaluating the random features, a direct
implementation requires
\[
O((n+m)Nt)
\]
time to embed points in $\mathbb{R}^N$, followed by
$O((n+m)t)$ time to form the means. Thus for samples of comparable size
$n$, the dependence on the sample size changes schematically from
\[
O(n^2 N)
\qquad\text{to}\qquad
O(nNt),
\]
up to the cost model used for Gaussian-kernel and feature evaluations.
Consequently, when
\[
n\gg t,
\]
the random-feature formulation provides a substantial computational
advantage.

The manifold assumption is relevant here because the number of random
features required by Theorem~\ref{thm:additive_uniform} is controlled by the
intrinsic geometry of $M$ rather than by a linear dependence on the ambient
dimension. Schematically, for fixed geometric parameters,
\[
t
=
\widetilde O\!\left(
\frac{d^2\log N}{\varepsilon^2}
\right),
\]
rather than an ambient-dimensional feature count. Hence, for
$d\ll N$ and a large number of samples, probability distributions supported
on $M$ can be represented by low-dimensional vectors
$\widehat\mu_P\in\mathbb{R}^{2t}$ whose mutual Euclidean distances
approximate their Gaussian-kernel MMD.

This has an additional benefit when many distributional comparisons are
required. Once a dataset has been summarized by its RFF mean vector
$\widehat\mu_P$, the original sample need not be revisited for subsequent
MMD comparisons: comparing two already-computed distributions costs only
$O(t)$ time. Thus a collection of large datasets supported on the same
low-dimensional manifold can be compressed into one $2t$-dimensional
vector per dataset while retaining their Gaussian-kernel MMD geometry up
to the additive error of Corollary~\ref{cor:mmd-preservation}.  \bibliography{RFFbib}

\appendix

\section{Concentration Inequality for Gaussian Kernel distances}
\label{sec:conc-ineq-gcc}
In this section, we shall state a concentration inequality for certain trigonometric functions of projections of Gaussian random vectors, which will be crucially used in our
bound on the distortion of the GKPD weight function under the RFF map. The inequality is described in the following general framework.
        Let $\Delta \in \R^D$ be a single vector. For $k=1,\ldots,t$, let
$g_k$ be independent and identically distributed standard normal vectors in $\R^D$, and
define
\begin{eqnarray*}
   L_t = L_t(\Delta) &:= & \frac{1}{4} \cdot \frac{1}{t}\sum_{k=1}^t\pth{1-\cos(\langle \Delta,g_k\rangle)}.
\end{eqnarray*}
The first two lemmas below are for general random variables.
Lemma~\ref{l:lem1} is a standard optimization used typically in proofs of concentration inequalities.
Lemma~\ref{l:lem2} is a slight generalization and improvement
of similar lemmas in~\cite{Freedman75,DBLP:journals/toc/BansalDGL19}.

\begin{lemma}
\label{l:lem1}
    Let $X$ be a random variable such that there exists $A>0$ such that for all $\lambda >0$, $\Ex{e^{\lambda X}}\leq \exp\pth{\lambda \Ex{X}}\exp\pth{(e^{\lambda}-\lambda-1)A}$.
    Let $Y =(\sum_{i=1}^t X_i)/t$ be the average of $t$ independent copies of $X$, given by $X_1,\ldots,X_t$.
    Then for any $A>tB>0$,
    \begin{eqnarray}
       \Prob{Y-\Ex{Y}\geq B} &\leq& \exp\pth{-\frac{tB^2}{2A}+\frac{t^2B^3}{6A^2}} \;\;\leq\;\; \exp\pth{-\frac{tB^2}{3A}}.
    \end{eqnarray}
\end{lemma}

\begin{proof}[Proof of Lemma~\ref{l:lem1}]
    We have
    \begin{eqnarray}
       \Prob{Y-\Ex{Y} \geq B} &=& \Prob{e^{\lambda(Y-\Ex{Y})}\geq e^{\lambda B}} \\
                              &\leq& \Ex{e^{\lambda(Y-\Ex{Y})}}\cdot e^{-\lambda B} \\
                              &=& \Ex{e^{\lambda Y}}e^{-\lambda\Ex{Y}}\cdot e^{-\lambda B}
    \end{eqnarray}
By the premise of the lemma, $\Ex{e^{\lambda Y}}$ can be simplified as
    \begin{eqnarray}
       \Ex{e^{\lambda Y}} &=&    \Ex{e^{\frac{1}{t}\sum_{i=1}^t \lambda X_i}} \\
                          &=&    \prod_{i=1}^t\Ex{e^{\lambda X_i/t}} \;\;=\;\;    \pth{\Ex{e^{\lambda X/t}}}^t \\
                          &\leq& \pth{e^{\lambda \Ex{X}/t}e^{(e^{\lambda/t}-\lambda/t-1)A}}^t
                          \;\;=\;\; \pth{e^{\lambda \Ex{X}}e^{t(e^{\lambda/t}-\lambda/t-1)A}} \\
                          &\leq& e^{\lambda \Ex{Y}}\pth{e^{(e^{\lambda}-\lambda-1)A/t}} \label{eqn:expr-lambda-y}
    \end{eqnarray}
    where the second line above follows from the fact that $X_i$ are identically distributed and independent copies of $X$, and
    the first inequality was by using the condition in the statement of the lemma, and the last line followed from that
    $\Ex{X}=\Ex{Y}$, together with the fact that for any $\lambda>0$, $t\geq 1$,
    $t(e^{\lambda/t}-\lambda/t-1)\leq (e^{\lambda}-\lambda-1)/t$, which can be easily seen using basic
    calculus.
The right-hand side of ~\eqref{eqn:expr-lambda-y} can be bound by optimizing the choice of
    $\lambda$. From elementary calculus we get that the optimum is when
    $\lambda = \ln\pth{1+tB/A}$. Substituting this value of $\lambda$ in the right-hand side of the last expression, we get
    \begin{eqnarray}
       \Prob{Y-\Ex{Y} \geq B} &\leq& \exp\pth{B-(A+B)\ln\pth{1+tB/A}}  \\
                              &\leq& \exp\pth{-\frac{tB^2}{2A}+\frac{t^2B^3}{6A^2}} ,
    \end{eqnarray}
    where in the last line we used the Taylor series expansion for $\ln(1+x)$ about $x=0$.
    This proves the first inequality in the statement of the Lemma~\ref{l:lem1}. For the second inequality in the lemma, we
    just use that $tB<A$ to get that $t^2B^3/6A^2 < B^2/6A$, and substitute in the tail bound.

\end{proof}

\begin{lemma}
\label{l:lem2}
Let $X$ be a random variable such that $|X|\leq 1/2$. Then the following
inequality holds true.
    \[ \Ex{e^{\lambda X}} \leq \exp\pth{(e^{\lambda}-\lambda-1)\Var{X}}\cdot\exp\pth{\lambda \Ex{X}}.\]
\end{lemma}

\begin{proof}[Proof of Lemma~\ref{l:lem2}]
Let $x_0:=\Ex{X}$. Now since $|X|\leq 1/2$, we have $\Ex{X}=x_0\in [-1/2,1/2]$, so that $(X-x_0) \in [-1,1]$.
   Consider the function $f(y)=\frac{e^{\lambda y}-\lambda y-1}{y^2}$, $y\neq 0$, and $f(0)=1/2$. From elementary calculus, for $y\in [-1,1]$ $f(y)$ is increasing.
Therefore for $y\in [-1,1]$, $f(y)\leq f(1)=(e^{\lambda}-\lambda-1)$.
Taking $y=(X-x_0)$, we get $f(y)\leq (e^{\lambda}-\lambda-1)y^2$.
Now taking expectations gives
\[ \Ex{e^{\lambda (X-x_0)}-\lambda(X-x_0)-1} \leq (e^{\lambda}-\lambda-1)\Ex{(X-x_0)^2} = (e^{\lambda}-\lambda-1)\Var{X} ,\]
or
\begin{eqnarray}
    \Ex{e^{\lambda (X-x_0)}} &\leq& 1+\Ex{\lambda(X-x_0)}+(e^{\lambda}-\lambda-1)\Var{X}  \\
                                        &=& 1+(e^{\lambda}-\lambda-1)\Var{X} \\
                                        &\leq& \exp\pth{(e^{\lambda}-\lambda-1)\Var{X}}.
\end{eqnarray}
Thus we get $\Ex{e^{\lambda X}} \leq \exp\pth{\lambda x_0}\cdot\exp\pth{(e^{\lambda}-\lambda-1)\Var{X}}$.
\end{proof}

Since $\langle\Delta,g\rangle = \|\Delta\| g_1$ with $g_1\sim\calN(0,1)$, we have
$L = \frac{1}{4}(1-\cos(\|\Delta\| g_1))$.  Note that $|L|\leq 1/2$, so applying Lemma~\ref{l:lem2} directly to $X=L$ gives, for any $\lambda>0$,
\[ \Ex{e^{\lambda L}} \leq \exp\pth{\lambda\Ex{L}}\cdot\exp\pth{(e^{\lambda}-\lambda-1)\Var{L}}. \]

Thus we obtain the following concentration inequality for $L_t$.

\begin{theorem}
\label{thm:cos-conc-ineq}
   For any $\e\in [0,1]$, the following holds.
   \begin{eqnarray}
       \Prob{|L_t-\Ex{L_t}|\geq \e\Ex{L_t}} \leq 2\cdot\exp\pth{-\frac{\e^2t\Ex{L}^2}{3\Var{L}}}.
   \end{eqnarray}
\end{theorem}

\begin{proof}[Proof of Theorem~\ref{thm:cos-conc-ineq}]
    The proof follows directly from Lemmas~\ref{l:lem2} and~\ref{l:lem1}. We focus on the upper tail, as the lower tail can be bounded by the same argument.
Applying Lemma~\ref{l:lem2} to $L$ gives $\Ex{e^{\lambda L}}\leq \exp\pth{\lambda\Ex{L}}\cdot\exp\pth{(e^{\lambda}-\lambda-1)\Var{L}}$. Now
applying Lemma~\ref{l:lem1} with $A = \Var{L}$ and $B=\e\Ex{L}$, and recalling that $L_t$ is the sum of $t$ independent copies of $L$,
gives
      \[ \Prob{|L_t-\Ex{L_t}|\geq \e\Ex{L_t}} \leq 2\cdot\exp\pth{-\frac{\e^2t\Ex{L}^2}{3\Var{L}}},\]
which is the statement of the theorem.
\end{proof}

It therefore remains to compute $\Ex{L}$ and $\Var{L}$ explicitly.

\begin{lemma}
\label{l:var-cos}
For $g\sim\calN(0,1)$ and $\theta\in\R$,
\begin{equation}
\label{eqn:var-cos}
   \Var{\cos(\theta g)} \;=\; \frac{1}{2}\pth{1-e^{-\theta^2}}^2.
\end{equation}
\end{lemma}

\begin{proof}[Proof of Lemma~\ref{l:var-cos}]
The characteristic function of the standard Gaussian gives $\Ex{e^{i\theta g}} = e^{-\theta^2/2}$.
Taking real parts yields $\Ex{\cos(\theta g)}=e^{-\theta^2/2}$.
Using the double-angle identity $\cos^2(\theta g)=\frac{1}{2}(1+\cos(2\theta g))$
and applying the same formula at $2\theta$:
\[
   \Ex{\cos^2(\theta g)} = \frac{1+e^{-2\theta^2}}{2}.
\]
Therefore,
\[
   \Var{\cos(\theta g)}
   \;=\; \Ex{\cos^2(\theta g)}-\pth{\Ex{\cos(\theta g)}}^2
   \;=\; \frac{1+e^{-2\theta^2}}{2} - e^{-\theta^2}
   \;=\; \frac{1}{2}\pth{1-e^{-\theta^2}}^2.
\]
\end{proof}

\begin{lemma}[Expectation of $L$]
\label{l:exp-L-singleton}
With $L = \frac{1}{4}(1-\cos(\langle\Delta,g\rangle))$ and $g\sim\calN(0,I_D)$,
\begin{equation}
\label{eqn:exp-L-singleton}
   \Ex{L} \;=\; \frac{1}{4}\pth{1 - e^{-\|\Delta\|^2/2}}.
\end{equation}
In particular, $\Ex{L_t}=\Ex{L}$ for every $t\geq 1$.
\end{lemma}

\begin{proof}[Proof of Lemma~\ref{l:exp-L-singleton}]
Since $\langle\Delta,g\rangle = \|\Delta\|g_1$ with $g_1\sim\calN(0,1)$, we have
$\Ex{\cos(\theta g_1)}=e^{-\theta^2/2}$,
\[
   \Ex{L}
   = \frac{1}{4}\pth{1-\Ex{\cos(\|\Delta\|g_1)}}
   = \frac{1}{4}\pth{1-e^{-\|\Delta\|^2/2}}.
\]
Since $L_t$ is the average of $t$ independent copies of $L$, we have $\Ex{L_t}=\Ex{L}$.
\end{proof}

We now have all the ingredients to state the final explicit concentration bound.
Since $\Var{L} = \frac{1}{16}\Var{\cos(\|\Delta\|g_1)}$, substituting $\Ex{L}$ from
Lemma~\ref{l:exp-L-singleton} and $\Var{\cos}$ from Lemma~\ref{l:var-cos}
into Theorem~\ref{thm:cos-conc-ineq} gives the following.

\begin{corollary}[Explicit concentration bound for $L_t(\Delta)$]
\label{cor:explicit-singleton}
For any $\e\in[0,1]$ and $t\geq 1$,
\begin{equation}
\label{eqn:final-singleton}
   \Prob{|L_t-\Ex{L_t}|\geq \e\Ex{L_t}}
   \;\leq\;
   2\exp\!\pth{-\frac{2\e^2 t}{3\pth{1+e^{-\|\Delta\|^2/2}}^2}}.
\end{equation}
\end{corollary}

\begin{proof}[Proof of Corollary~\ref{cor:explicit-singleton}]
Theorem~\ref{thm:cos-conc-ineq} gives
\[
   \Prob{|L_t-\Ex{L_t}|\geq \e\Ex{L_t}}
   \leq 2\exp\!\pth{-\frac{\e^2 t\,\Ex{L}^2}{3\Var{L}}}.
\]
Since $L = \frac{1}{4}(1-\cos(\|\Delta\|g_1))$, we have $\Var{L} = \frac{1}{16}\Var{\cos(\|\Delta\|g_1)}$.
From Lemma~\ref{l:var-cos} with $\theta=\|\Delta\|$, $\Var{\cos(\|\Delta\|g_1)} = \frac{1}{2}(1-e^{-\|\Delta\|^2})^2$, so
$\Var{L} = \frac{1}{32}(1-e^{-\|\Delta\|^2})^2$.
From Lemma~\ref{l:exp-L-singleton}, $\Ex{L}=\frac{1}{4}(1-e^{-\|\Delta\|^2/2})$, so
$\Ex{L}^2 = \frac{1}{16}(1-e^{-\|\Delta\|^2/2})^2$.
Therefore,
\[
   \frac{\e^2 t\,\Ex{L}^2}{3\Var{L}}
   \;=\;
   \frac{\e^2 t\cdot\frac{1}{16}(1-e^{-\|\Delta\|^2/2})^2}{3\cdot\frac{1}{32}(1-e^{-\|\Delta\|^2})^2}
   \;=\;
   \frac{2\e^2 t\,(1-e^{-\|\Delta\|^2/2})^2}{3\,(1-e^{-\|\Delta\|^2})^2}.
\]
We now simplify using the factorisation
\[
   1-e^{-\|\Delta\|^2}
   = \pth{1-e^{-\|\Delta\|^2/2}}\pth{1+e^{-\|\Delta\|^2/2}},
\]
which gives $(1-e^{-\|\Delta\|^2})^2=(1-e^{-\|\Delta\|^2/2})^2(1+e^{-\|\Delta\|^2/2})^2$.
Cancelling the common factor $(1-e^{-\|\Delta\|^2/2})^2 > 0$:
\[
   \frac{\e^2 t\,\Ex{L}^2}{3\Var{L}}
   \;=\;
   \frac{2\e^2 t}{3\pth{1+e^{-\|\Delta\|^2/2}}^2},
\]
which is exact.  Substituting into the tail bound gives~\eqref{eqn:final-singleton}.
\end{proof}

\subsection{Application: Relative Error for the Gaussian Kernel Distance}
\label{sec:kernel-distance}

We now connect the concentration inequality for $L_t$ to a relative error
bound for the Gaussian kernel distance under random Fourier features (RFF).
This recovers a lemma from~\cite{chen2017relative} with a cleaner,
self-contained proof.

\medskip
\noindent\textbf{Gaussian kernel and RFF setup.}
Let $\sigma>0$ be the kernel bandwidth. For $x,y\in\R^D$, set
$\Delta := (x-y)/\sigma$.
The Gaussian kernel and its associated kernel distance are
\[
K(x,y)=e^{-\frac{\|x-y\|^2}{2\sigma^2}}=e^{-\frac{1}{2}\|\Delta\|^2},\qquad
D_K(x,y)=\sqrt{2-2K(x,y)}=\sqrt{2-2e^{-\frac{1}{2}\|\Delta\|^2}}.
\]

Draw $t$ independent Gaussian vectors $\omega_1,\ldots,\omega_t\sim\calN(0,\sigma^{-2}I_D)$
and define the RFF embedding $\hat{\phi}:\R^D\to\R^{2t}$ coordinate-wise by
\[
\bigl[\hat{\phi}(x)_{2k-1};\; \hat{\phi}(x)_{2k}\bigr]
   =\frac{1}{\sqrt{t}}\bigl[\cos(\langle\omega_k,x\rangle);\;
                            \sin(\langle\omega_k,x\rangle)\bigr],\qquad k=1,\ldots,t.
\]
The approximate distance is
\[
D_{\hat{K}}(x,y)=\|\hat{\phi}(x)-\hat{\phi}(y)\|.
\]

A standard computation (using $\cos a\cos b+\sin a\sin b=\cos(a-b)$
and $\langle\hat{f}_k(x),\hat{f}_k(x)\rangle=1$) gives
\begin{equation}\label{eqn:Dhatk-via-cos}
D_{\hat{K}}(x,y)^2
   = 2-\frac{2}{t}\sum_{k=1}^{t}\cos\bigl(\langle\omega_k,\,x-y\rangle\bigr)
   = \frac{2}{t}\sum_{k=1}^{t}\bigl(1-\cos(\langle\omega_k,\,x-y\rangle)\bigr).
\end{equation}

By rotational invariance of the Gaussian, $\langle\omega_k,x-y\rangle$
has the same distribution as $\|\Delta\|g_k$ where $g_k\sim\calN(0,1)$
are i.i.d.\ standard Gaussians. Let $\tilde{g}_k\sim\calN(0,I_D)$;
then $\langle\Delta,\tilde{g}_k\rangle = \|\Delta\|g_k$.
Comparing with~\eqref{eqn:Dhatk-via-cos} and the definition of $L_t$:
\begin{equation}\label{eqn:bridge}
D_{\hat{K}}(x,y)^2
   = \frac{2}{t}\sum_{k=1}^{t}\bigl(1-\cos(\|\Delta\| g_k)\bigr)
   = 8\,L_t(\Delta),
\qquad
D_K(x,y)^2
   = 2\bigl(1-e^{-\frac{1}{2}\|\Delta\|^2}\bigr)
   = 8\,\Ex{L_t(\Delta)}.
\end{equation}

Thus the squared-distance ratio is exactly the ratio of $L_t$ to its mean:
\[
\frac{D_{\hat{K}}(x,y)^2}{D_K(x,y)^2}
   = \frac{L_t(\Delta)}{\Ex{L_t(\Delta)}}.
\]

\medskip
\noindent\textbf{From concentration of $L_t$ to relative error.}
Corollary~\ref{cor:explicit-singleton} states that for any $\varepsilon\in[0,1]$,
\[
\Pr\!\Bigl[\,|L_t-\Ex{L_t}|\geq \varepsilon\,\Ex{L_t}\Bigr]
   \leq 2\exp\!\Bigl(-\frac{2\varepsilon^2 t}{3(1+e^{-\|\Delta\|^2/2})^2}\Bigr).
\]

When $\|x-y\|\leq\sigma$, we have $\|\Delta\|\leq 1$, hence
$e^{-\|\Delta\|^2/2}\in[e^{-1/2},1]\subseteq[0.6,1]$ and
$(1+e^{-\|\Delta\|^2/2})^2\leq 4$.  Therefore
\[
\Pr\!\Bigl[\,|L_t-\Ex{L_t}|\geq \varepsilon\,\Ex{L_t}\Bigr]
   \leq 2\exp\!\Bigl(-\frac{\varepsilon^2 t}{6}\Bigr).
\]

Choosing $t = \frac{6}{\varepsilon^2}\ln\frac{2}{\delta}
        = \Omega\!\bigl(\frac{1}{\varepsilon^2}\log\frac{1}{\delta}\bigr)$
makes this probability at most $\delta$.

On the complement event,
$(1-\varepsilon)\Ex{L_t}\leq L_t\leq(1+\varepsilon)\Ex{L_t}$,
and by~\eqref{eqn:bridge} this is equivalent to
\[
(1-\varepsilon)\,D_K(x,y)^2 \;\leq\; D_{\hat{K}}(x,y)^2 \;\leq\; (1+\varepsilon)\,D_K(x,y)^2.
\]

Taking square roots and using
$\sqrt{1-\varepsilon}\geq 1-\varepsilon$,
$\sqrt{1+\varepsilon}\leq 1+\varepsilon$ (valid for $\varepsilon\in[0,1]$),
we obtain:

\begin{lemma}[Relative error bound]\label{lem:chen-phillips-small}
If $\|x-y\|\leq\sigma$ and
$t \geq \frac{6}{\varepsilon^2}\log\frac{2}{\delta}$,
then
\[
\Pr\!\Bigl(\frac{D_{\hat{K}}(x,y)}{D_K(x,y)}\in[1-\varepsilon,\,1+\varepsilon]\Bigr)
   \geq 1-\delta.
\]
\end{lemma}

\begin{proof}
Follows directly from Corollary~\ref{cor:explicit-singleton}
and the identities~\eqref{eqn:bridge}, as detailed above.
\end{proof}

\end{document}